\documentclass[a4paper,twocolumn,11pt,unpublished]{quantumarticle}
\pdfoutput=1

\usepackage[utf8]{inputenc}
\usepackage[english]{babel}
\usepackage[T1]{fontenc}
\usepackage{amsmath}
\usepackage{subcaption}
\usepackage{float}
\usepackage{enumitem}
\usepackage{quantikz}
\usepackage{tikz}
\definecolor{primalorange}{RGB}{230,120,20}
\definecolor{dualblue}{RGB}{35,75,230}
\definecolor{intersectionpurple}{RGB}{255,0,92}

\definecolor{pgreen}{RGB}{67, 143, 100}
\definecolor{porange}{RGB}{199, 103, 42}
\usepackage[
colorlinks,
citecolor=pgreen,
linkcolor=porange,
urlcolor=pgreen
]{hyperref}

\usepackage{lipsum}
\usepackage{amsmath,amssymb,mathtools,amsthm}

\allowdisplaybreaks
\usepackage{dsfont}

\newcommand{\Id}{\mathds{1}}

\DeclareMathOperator{\Span}{span}
\newtheorem{definition}{Definition}
\newtheorem{lemma}{Lemma}
\newtheorem{proposition}{Proposition}
\newtheorem{theorem}{Theorem}
\newtheorem{remark}{Remark}

\newtheorem{result}{Result}
\newcommand{\sign}[0]{\text{sgn}}
\newcommand{\supp}[1]{\text{supp}\left(#1\right)}

 \usepackage[dvipsnames]{xcolor}

\newcommand{\carlos}[1]{}  
\newcommand{\gustavo}[1]{}
\newcommand{\denis}[1]{}
\newcommand{\marco}[1]{}
\newcommand{\amanda}[1]{}
\begin{document}

\title{Self-testing the toric code against classical communication}

\author{Gustavo Fróes}
\affiliation{CPHT, LIX, CNRS, Inria, École polytechnique, Institut Polytechnique de Paris, 91120 Palaiseau, France}
\affiliation{Télécom Paris, Institut Polytechnique de Paris, 19 Place Marguerite Perey, 91120 Palaiseau, France}
\affiliation{Instituto de Física Gleb Wataghin, Universidade Estadual de Campinas, CEP 13083-859 Campinas, Brazil}
\author{Amanda Wei}
\affiliation{CPHT, LIX, CNRS, Inria, École polytechnique, Institut Polytechnique de Paris, 91120 Palaiseau, France}
\affiliation{Nokia Bell Labs, 12 rue Jean Bart, 91300 Massy, France}
\author{Carlos de Gois}
\affiliation{CPHT, LIX, CNRS, Inria, École polytechnique, Institut Polytechnique de Paris, 91120 Palaiseau, France}
\author{Denis Rochette}
\affiliation{CPHT, LIX, CNRS, Inria, École polytechnique, Institut Polytechnique de Paris, 91120 Palaiseau, France}
\author{Marc-Olivier Renou}
\affiliation{CPHT, LIX, CNRS, Inria, École polytechnique, Institut Polytechnique de Paris, 91120 Palaiseau, France}

\maketitle
\begin{abstract}

\marco{ALL CAN BE COMMENTED IN/OUT IN THE COMMAND BEFORE THE begin\{document\} LINE}

Device-independent self-testing certifies a multipartite quantum state from the correlations it produces alone, without assumptions on the inner workings of the devices. 
Most existing self-tests, however, break down as soon as the parties can exchange messages. 
We consider a relaxed setting in a quantum network, where in the honest version of the protocol, the parties never communicate, yet certification remains sound even in an untrusted case where the untrusted devices secretly exchange classical messages within a bounded distance in the network. 
We show that the code space of the toric code, distributed over a network with the geometry of a torus, can be self-tested in this setting in a noise-robust way, against communication up to distances proportional to the network diameter. 
We also develop an algorithmic method for constructing such protocols and extend our results to the surface code and cluster states, robust to one round of nearest-neighbor classical communication
\end{abstract}

Quantum theory allows a state to be certified from observed measurement statistics alone, requiring neither a model of the experimental apparatus nor any trust in it.
This is known as self-testing~\cite{Mayers2004}. It is the strongest form of quantum certification~\cite{Supic2020-ox,Eisert2020QuantumCertification} and a key tool in device-independent verification of quantum computation~\cite{Gheorghiu2018Verification} and in quantum complexity theory~\cite{McKague2016,Ji2021MIP}.

Self-testing protocols exist for all pure bipartite entangled states~\cite{Coladangelo2017AllPure} and for all pure multipartite entangled states of qubits~\cite{BalanzoJuando2026AllPureMultipartite}.
In quantum networks involving independent sources, it extends to all pure~\cite{Supic2023QuantumNetworks} and even mixed~\cite{Sarkar2026UniversalScheme} states.
Self-testing can also certify entire subspaces, beyond individual states.
This is relevant for quantum error correction, where information is encoded in a \emph{code space} rather than in a single state.
Subspace self-tests are known for the five-qubit code~\cite{Baccari2020}, the toric code~\cite{Baccari2020,Hart2025,Vallee2025-ty} and for certain genuinely entangled stabilizer subspaces~\cite{Makuta2021}. 

Crucially, self-testing relies on the assumption that the subsystems cannot communicate during the protocol.
Enforcing this in practice is demanding even for two parties~\cite{Storz2025CompleteSelfTesting}. 
The difficulty grows in the multipartite setting, where each party must be individually addressed and read out, with all measurement events pairwise spacelike separated.
To our knowledge, no multipartite experiment simultaneously closing the detection and locality loopholes has been reported~\cite{Zhang2018Experimentally,Wu2021Robust,Xu2022Experimental,Wu2022ClosingLocality}.
Similar issues also arise in non-adversarial settings.
In device-independent property certification~\cite{Barreiro2013Demonstration,Wang2025Probing,Cai2026GeneralizedMermin}, for example, the absence of communication cannot be assumed from the start, since crosstalk can cause the measurement basis chosen on one party to affect the statistics of others~\cite{Sarovar2020DetectingCrosstalk,Tabia2025Almost}.

Rather than enforcing the no-communication assumption in practice, one can relax it in the theoretical model.
A natural way to do so is to suppose each untrusted devices can secretly exchange classical messages with all subsystems within some fixed distance of it.
It has been shown that the correlations of some multipartite states cannot be reproduced classically even when the communication distance grows sublinearly with the number of qubits~\cite{Barrett2007,Meyer2023}.
This underlies quantum advantages for shallow circuits~\cite{Bravyi2018}, distributed computing~\cite{legall2019b}, and randomness certification~\cite{Coudron2021}.
First self-test results in a model allowing classical communication were obtained only recently, for the cycle and honeycomb graph states, and for arbitrary graph states only indirectly via auxiliary qubits~\cite{Meyer2026}.
No such protocol is known for other states, and in particular none that certifies a subspace.

Here we provide communication robust protocols to self-test the toric code subspace, and an algorithmic method to construct similar protocols for large families of stabilizer codes, assuming classical communication between the parties.
Our construction for the toric code is based on a new self-test for the non-communicating scenario which requires only single-qubit Pauli measurements (Section~\ref{sec:without-comm}).
In contrast to the previous protocols \cite{Baccari2020,Hart2025}, which do not survive communication (Appendix~\ref{sec:communicating-classical-strategies}), ours can be made robust to classical communication up to a distance $r$ on any lattice of size $L \geq 6r+3$ (Section~\ref{sec:comm}), which is asymptotically optimal.
The protocols are robust to noise (Section~\ref{sec:noise}), with deviations of at most $\epsilon$ from the ideal statistics certifying a distance $\mathcal{O}(L^2\sqrt{\epsilon})$ from the toric code subspace.
We further develop a systematic method to search for self-tests of stabilizer subspaces whose generators are products of Pauli $X$ and $Z$ operators (Section~\ref{sec:MILP}).
This can be written as an integer linear program, and as examples we find additional self-tests for the toric code which are smaller than our analytical construction (Section~\ref{sec:toric-mip}) and self-tests robust to $1$ round of classical communication for the surface code and the two-dimensional cluster state (Secs.~\ref{sec:surface-mip} and \ref{sec:graph-states-mip}).

\section{Preliminaries \label{sec:prelim}}

\subsection{Toric code}
The toric code is a prototypical example of a topological error correction code \cite{KITAEV20032, Dennis2002}.
It is defined on a square lattice of size $L \times L$ with periodic boundary conditions, where each edge represents a physical qubit.
The encoded subspace is defined by two types of stabilizers.
The \textit{star stabilizers} $A_v$ are defined on each vertex $v$ of the lattice by
\begin{equation}
    A_v = \prod_{\partial e \ni v} X_e,
    \label{eq:star-def}
\end{equation}
where $\partial e \ni v$ denotes the edges $e$ connected to the vertex $v$, and $X_e$ denotes the Pauli $X$ operator acting on the edge $e$. Similarly, the \textit{plaquette stabilizers} $B_f$ are defined on each face $f$ by
\begin{equation}
    B_f = \prod_{e \in \partial f} Z_e,
    \label{eq:plaquette-def}
\end{equation}
where $e \in \partial f$ denotes the edges $e$ composing the face $f$, and $Z_e$ is the Pauli $Z$ operator on the system in $e$.

The \textit{toric code subspace} is defined by all states $\ket{\psi}$ stabilized by all the star and plaquette stabilizers on the lattice, i.e.
\begin{align}
    &A_v \ket{\psi} = +1 \ket{\psi} \quad \forall v, \\
    & B_f \ket{\psi} = +1 \ket{\psi} \quad \forall f.
\end{align}

 \begin{figure}[t]
     \centering
     \includegraphics[width=0.65\linewidth]{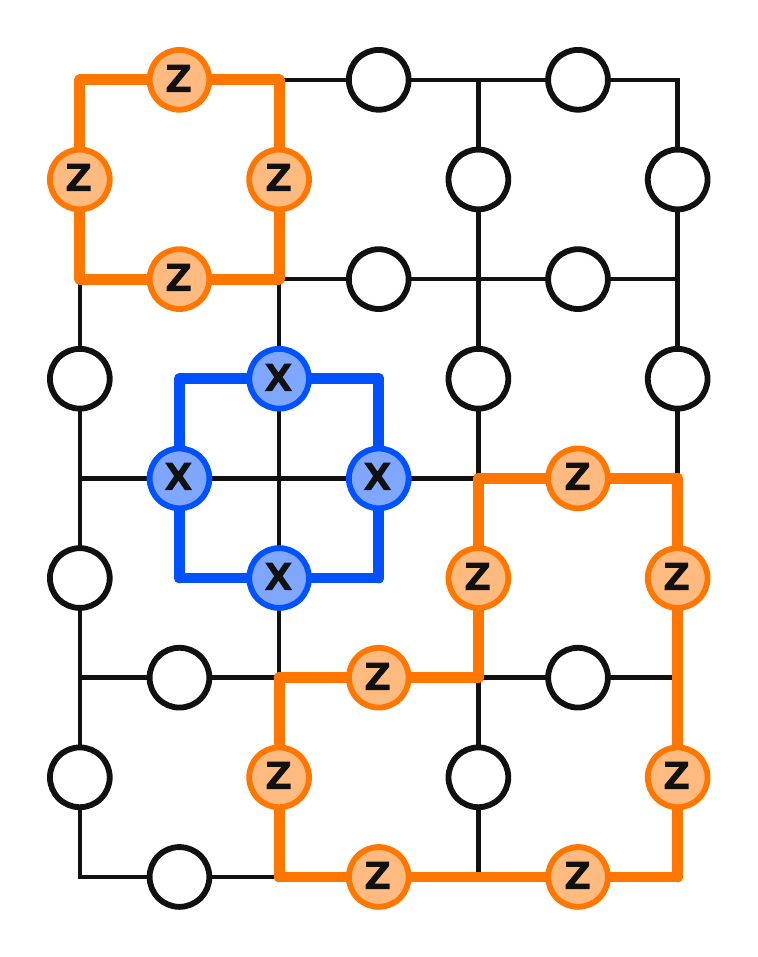}
    \caption{Illustration of the primal and dual loops acting on a patch toric code. Physical qubits are associated with the edges of the primal lattice. The orange paths represent contractible loops in the primal lattice, and the corresponding product of Pauli $Z$ operators. The blue path represents a contractible loop in the dual lattice, and the corresponding product of Pauli $X$ operators acting on the primal edges crossed by the loop.}
    \label{fig:toric-code-diagram}
 \end{figure}

A useful way to represent these stabilizers is in terms of loops on the lattice. These loops specify the support of a collection of Pauli operators. For a path $\gamma$ on the original lattice (also called the primal lattice), we associate the operator
\begin{equation}
    Z(\gamma)=\prod_{e\in\gamma} Z_e.
\end{equation}
As illustrated in Fig.~\ref{fig:toric-code-diagram} by the orange loop, a plaquette stabilizer can be identified with the smallest such loop, namely, $B_f = Z(\partial f)$. More generally, the product of plaquette stabilizers associated with a collection of faces corresponds to a loop along the boundary of the region formed by those faces.

Similarly, the star stabilizers define loops on the \textit{dual lattice}, obtained by shifting the original lattice by half an edge in both directions. A dual-lattice path $\gamma^\ast$ represents the operator
\begin{equation}
    X(\gamma^\ast)
    =
    \prod_{e\in\gamma^\ast} X_e ,
\end{equation}
that is, an $X$ operator is applied to every physical qubit whose edge is crossed by the dual loop. In particular, the boundary of a single dual face surrounds one vertex $v$ of the primal lattice and corresponds precisely to the star stabilizer $A_v$, as illustrated by the blue loop in Fig.~\ref{fig:toric-code-diagram}. The products of star stabilizers behave similarly to the plaquettes, and the resulting operator is supported only on the boundary of the corresponding collection of dual faces. We refer to all such closed loops as \textit{contractible loops}, since they can be generated by products of stabilizers and therefore correspond themselves to stabilizer operators.

Finally, to denote the edges of the lattice, we will use the following convention for the indices modulo $L$. Let $(i,j)$ and $(i+1,j)$ be the left and right vertices of the lattice at the boundary of a horizontal edge; we then denote this edge by 
\begin{equation}
    H(i,j) = \left\{(i,j),(i+1,j)\right\}.
\end{equation}
Similarly, let $(i,j)$ and $(i,j+1)$ be the bottom and top vertices of the lattice at the boundary of a vertical edge; we then denote this edge by 
\begin{equation}
    V(i,j) = \left\{(i,j),(i,j+1)\right\}.
\end{equation} 
Furthermore, the vertex $(0,0)$ can be defined arbitrarily, and by convention, we always display it on the bottom left of the figures (see Fig.~\ref{fig:comm-graph-diagram}).

\subsection{Self-testing}

Self-testing is a form of certification where a quantum state and the measurements performed on it are uniquely determined from observations, up to unavoidable degrees of freedom, without requiring a characterization of the experimental apparatus \cite{Mayers2004,Supic2020-ox}.
The idea was later extended to other structures, such as circuits \cite{Magniez2006SelfTestingQuantumCircuits} and supermaps \cite{Barizien2026SelfTestingSupermaps}.
Of particular interest for us is its extension to the case where the certified object is not a single state, but an entire subspace \cite{Baccari2020}.

In the usual non-communicating scenario, we consider a set of $N$ untrusted parties that are spacelike separated, so that no communication between them is possible. A referee assigns to each party $i$ a classical input $x_i$, chosen uniformly at random, and the party returns an output $a_i$. If, for example, these outcomes are dichotomic and labeled $a_i\in\{\pm1\}$,  the referee can estimate the expectation value of the product of the outputs associated with each input string $\mathbf{x}$. Within the quantum formalism, the operation performed by party $i$ upon receiving input $x_i$ is represented by an unknown observable $\widetilde{O}_{x_i}$, acting on the uncharacterized shared state $\ket{\Psi}$, so that the corresponding expectation value is given by
\begin{equation}
\left \langle \prod_i a_i \right \rangle_\mathbf{x} = \bra{\Psi} \bigotimes_i \widetilde{O}_{x_i} \ket{\Psi}. 
\label{eq:expect-eq}
\end{equation}
We call this experiment the \textit{physical experiment} (PE). Informally, a self-testing protocol certifies that if the observed expectation values from a PE match those of a perfectly characterized and known \textit{reference experiment} (RE), then the physical state and measurements used in the PE are those of the reference experiment up to some local degrees of freedom. Throughout this article, we use tilded letters to refer to the uncharacterized measurements and observables made in the physical experiment, and untilded letters to refer to the characterized and trusted measurements made in the reference experiment.

In the model introduced in \cite{Barrett2007} (see also \cite{Meyer2023, Meyer2026}), the non-communicating assumption is relaxed, and, after receiving their inputs, the parties are allowed to communicate classical messages to their neighbors in a given communication graph, up to some specified distance. This is analogous to the LOCAL model of distributed algorithms \cite{Linial1987, Linial1992}, where at each \textit{round of communication} the parties are allowed to synchronously send a message of arbitrary size to their immediate neighbors. Therefore, in $r$ rounds of such a protocol, the parties will potentially have access to all the classical information present within a distance $r$. In this way, a self-testing protocol robust to $r$ rounds of communication must thus remain valid even if every party has access to all this information. 

Importantly, in the communication model considered here, as well as in \cite{Barrett2007,Meyer2023,Meyer2026}, the classical messages exchanged between the parties are not allowed to depend on the outcomes of quantum measurements. This condition is necessary for the communication to remain genuinely classical rather than merely being implemented through a classical channel. Indeed, since the uncharacterized parties are not assumed to have bounded local dimension, they may share an arbitrary amount of entangled states. Therefore, if they were allowed to locally perform arbitrary quantum operations and communicate the resulting measurement outcomes, they could use this pre-shared entanglement together with these classical messages to teleport quantum systems between neighboring parties, effectively enabling quantum communication. Consequently, to account for all the classical information locally available to the parties, it suffices to allow them to communicate their inputs, since any other classical message can be generated from these inputs together with shared randomness. Throughout the remainder of the manuscript, we therefore refer to this restricted communication model simply as classical communication.

To start making the above definitions more precise, let us introduce the notion of submeasurements, which were shown to be a useful tool for the certification of nonlocality in many-body systems (see, e.g., \cite{Hart2025, Daniel2022}). In fact, the authors in \cite{Barrett2007} showed that for graph states, there are communicating classical strategies capable of reproducing the expectation value of all Pauli measurements performed on graph states, but the same strategies fail to reproduce all the statistics of the submeasurements therein.

\begin{definition}[Submeasurements]
 Given a measurement ${M}_\mathbf{x} = \bigotimes^n_{i=1} {O}_{x_i}$, we call another measurement ${S}_\mathbf{x}$ a submeasurement of ${M}_\mathbf{x}$ if there exists a subset $I\subseteq \{1,\ldots,n\}$ such that
\begin{equation}
S_{\mathbf{x}}
=
\bigotimes_{i=1}^{n} S_{x_i},
\qquad
S_{x_i}
=
\begin{cases}
O_{x_i}, & i\in I,\\
\Id, & i\notin I.
\end{cases}
\end{equation}
In this case, we denote
\begin{equation}
S_{\mathbf{x}}
=
\left.M_{\mathbf{x}}\right|_{I}.
\end{equation}
\end{definition}
For example, $ A \otimes B \otimes \Id$ is a submeasurement of $A \otimes B \otimes C$.
Operationally, a submeasurement corresponds to the referee sending the inputs to every party and averaging out the outputs of the parties that are not in the set $I$.
This is a crucial tool for our proof, as the parties do not know which outputs are disregarded and thus have to always behave the same if the information if all the information available to them does not change. 

When the RE and PE are \textit{compatible}, meaning that they have the same number of parties and observables available at each party, we can compare them via the following definition.

\begin{definition}[Simulation]
Consider a reference experiment with state $\ket{\psi}$, and fix a set of submeasurements $\mathcal{S}=\{S_{\mathbf{x}}\}_{\mathbf{x}}$.
Let $\ket{\Psi}$ be the state of a compatible physical experiment, and let $\{\widetilde{S}_{\mathbf{x}}\}_{\mathbf{x}}$ denote the corresponding physical submeasurements, obtained by retaining the same subsets of parties in the respective measurements.
We say that the physical experiment simulates $\mathcal{S}$ if
\begin{equation}
    \bra{\Psi}\widetilde{S}_{\mathbf{x}}\ket{\Psi}
    =
    \bra{\psi}S_{\mathbf{x}}\ket{\psi}
\end{equation}
for every $S_{\mathbf{x}}\in\mathcal{S}$.
\label{def:simulation}
\end{definition}

Furthermore, to self-test a subspace and not any individual state inside it, any comparison between the PE and the RE must be made through submeasurements that are \emph{constant} over this subspace. More precisely, a submeasurement $S_x$ is said to be constant over a subspace $\mathcal{C}$ if, for every pair of states $\ket{\phi},\ket{\phi'}\in\mathcal{C}$, they satisfy
\begin{equation}
    \bra{\phi'}{S}_\mathbf{x}\ket{\phi'} = \bra{\phi}{S}_\mathbf{x}\ket{\phi}.
\end{equation}

We can now define self-testing of a subspace in a manner similar to \cite{Baccari2020}. The idea is that any physical experiment reproducing the statistics of a reference experiment, consisting of a state in the target subspace and measurements that are constant over it, can be mapped by local isometries to a state belonging to that same space.

\begin{definition}[Self-testing subspaces]
    Let 
\begin{equation}
    \mathcal{C} = \Span(\{\ket{\psi_\kappa}\}_\kappa),
\end{equation} 
be a subspace of $
 \bigotimes_i \mathcal{H}'_i$, and
 \begin{equation}
{M}_\mathbf{x} = \bigotimes_i {O}_{x_i}     
 \end{equation}
define the measurements and submeasurements associated with the input string $\mathbf{x}$, and that are constant over $\mathcal{C}$.
Furthermore, let 
\begin{equation}
\ket{\Psi} \in  \left ( \bigotimes_i \mathcal{H}_{i} \right) \otimes \mathcal{H}_P
\end{equation}
be any purification of a state that, together with the measurements $\bigl\{\widetilde{M}_\mathbf{x}\bigr\}_\mathbf{x}$, simulates the chosen set of submeasurements $\mathcal{S}=\{S_{\mathbf{x}}\}_{\mathbf{x}}$ of the reference experiment $\left(\ket{\psi},\left\{{M}_\mathbf{x}\right\}_\mathbf{x}\right)$, where $\ket{\psi}$ is any pure state in $\mathcal{C}$, and where the operators $\widetilde{M}_\mathbf{x}$ are understood to act trivially on the purifying system $\mathcal{H}_P$.

We say that the observed statistics $\{\langle \widetilde{S}_\mathbf{x}\rangle \}_{\mathbf{x}}$ self-tests the entangled subspace $\mathcal{C}$, if one can deduce that there exist local isometries $\Phi_i: \mathcal{H}_{_i} \rightarrow \mathcal{H}_i'\otimes \mathcal{H}''_i$ and states $\ket{\xi_\kappa} \in \mathcal{H}'' \otimes \mathcal{H}_P$ such that
    \begin{equation}
         \left(\Phi\otimes \Id_P\right) \ket{\Psi} = \sum_\kappa c_\kappa \ket{\psi_\kappa} \otimes \ket{\xi_\kappa},
    \label{eq:self-testing-def}     
    \end{equation}
where
\begin{equation}
    \Phi = \bigotimes_i \Phi_i.
\end{equation}
\label{def:self-testing-subspace}
\end{definition}

\begin{definition}[Communication-robust protocol]
    A self-testing protocol is robust to classical communication up to distance $r$ if Eq.~\eqref{eq:self-testing-def} can still be established when the untrusted parties have access, throughout the protocol, to all classical information available within distance $r$.
\end{definition}

As highlighted in \cite{Baccari2020}, Definition~\ref{def:self-testing-subspace} introduces an additional degree of freedom compared with standard notions of self-testing, encoded in the coefficients $c_i$. This additional freedom arises because the definition certifies only that the extracted state lies within a given subspace, without necessarily specifying which particular state in that subspace is realized.

To prove that a given protocol self-tests the toric subspace according to this definition, the following proposition will be useful.
\begin{proposition}
    For self-testing the toric subspace according to Definition~\ref{def:self-testing-subspace}, it suffices to show that the statistics
\begin{enumerate}
    \item Certify anticommutation locally, i.e., show the existence of two binary observables $\widetilde{X}_e$ and $\widetilde{Z}_e$ at each party $e$ that anticommute on the state $\{\widetilde{X}_e,\widetilde{Z}_e\}\ket{\Psi} = 0$, and
    \item Certify the stabilizer, i.e., show that these observables also reproduce the stabilizers of the toric code
\begin{align}
    &\prod_{\partial e \ni v}\widetilde{X}_e \ket{\Psi} = +1 \ket{\Psi} \quad \forall v, \nonumber\\
    & \prod_{e \in \partial f} \widetilde{Z}_e \ket{\Psi} = +1 \ket{\Psi} \quad \forall f. \nonumber
\end{align}

\end{enumerate}
\label{prop: self-testing-criteria}
\end{proposition}

Informally, the first condition in Proposition~\ref{prop: self-testing-criteria} certifies that a qubit subsystem can be extracted locally at each party, while the second ensures that the joint state of these extracted qubits is supported on the toric code subspace. This is very similar to the statement made in \cite{Meyer2026} and implicitly used in \cite{McKague2014,McKague2016}.
For completeness, a proof is provided in Appendix ~\ref{app:proof-prop1}.

\section{Self-testing the toric code without communication}
\label{sec:without-comm}

To show the construction of a self-testing protocol where communication is allowed, it will be convenient to start with the case where the parties cannot communicate.
Self-testing protocols for the toric code subspace have already been proposed \cite{Baccari2020,Hart2025}.
However, those constructions are not robust to classical communication (see Appendix~\ref{sec:communicating-classical-strategies}).
In this section, we present an alternative self-test that is generalizable to the communicating scenario and can be implemented using only single-qubit Pauli operations.

The referee provides inputs $x_{e}\in\{0,1,2\}$ to the subsystem on each edge $e$.
These inputs correspond to measurements in the $X,Y,Z$ basis on a state in the toric code subspace in the trusted scenario of the RE. The corresponding observables of the untrusted parties in the PE are denoted by $\widetilde{X}_{e}$, $\widetilde{Y}_{e}$, and $\widetilde{Z}_{e}$, and act on the purified shared state $\ket{\Psi}$. For convenience, we omit the identity operator on the purifying system $\mathcal{H}_P$, on which these observables act trivially. In return, each party $e$ will output a measurement result $a_{e} \in \{ \pm 1 \}$. 

In Fig.~\ref{fig:measurements-vertical}, we show a subset of parties in a $3\times 3$ patch of a lattice of the toric code and inputs defining four submeasurements $S_1, S_2, S_3, S_4$ in the RE and respectively $\widetilde{S}_1, \widetilde{S}_2, \widetilde{S}_3, \widetilde{S}_4$ in the PE.
The parties who are not in the submeasurement receive the input $1$.

\begin{figure*}[t]

    \begin{subfigure}{.97\textwidth}
        \centering
                \caption{}
        \label{fig:measurements-vertical}
        \includegraphics[width=\textwidth]
        {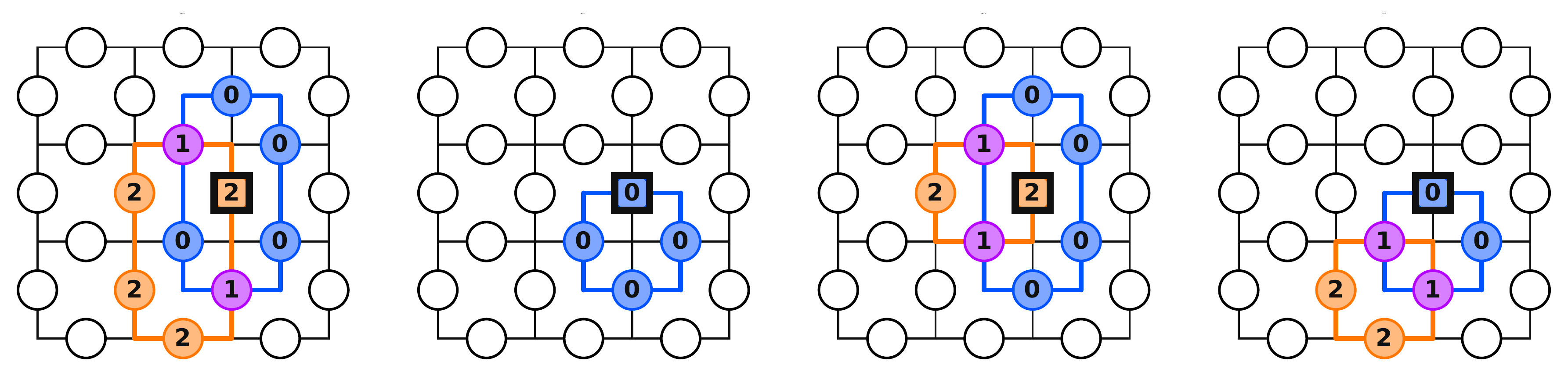}
    \end{subfigure}

    \vspace{3mm}

    \begin{subfigure}{.97\textwidth}
        \centering
        \caption{}
        \label{fig:measurements-horizontal}
        \includegraphics[width=\textwidth]
        {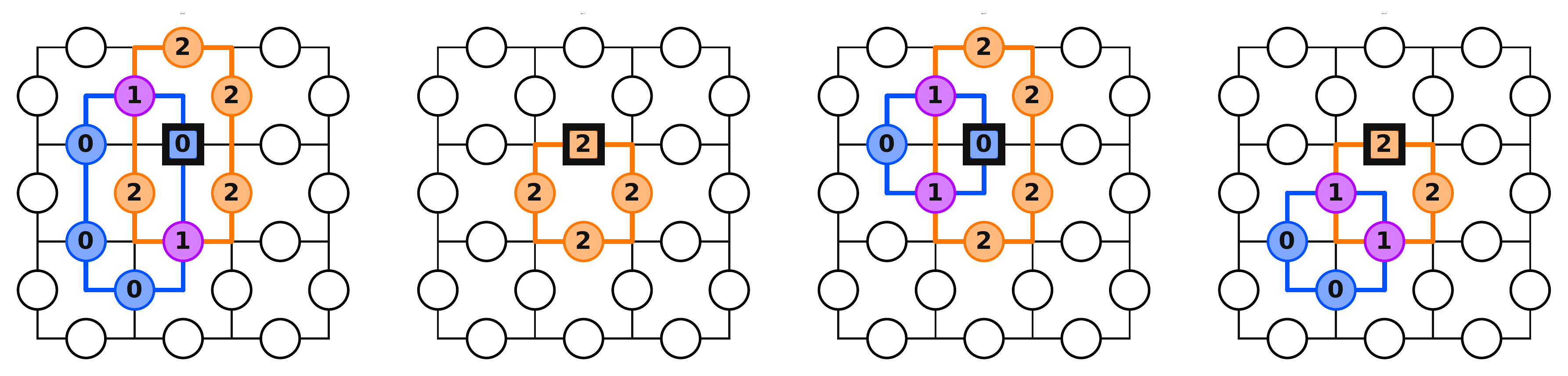}
    \end{subfigure}
    
    \caption{A $3\times 3$ patch of the toric code with the inputs given by the referee to each party in the submeasurements 1 through 4, from left to right, for certifying the anticommutation relations without any communication. For (a) a vertical edge and (b) a horizontal edge. Orange and blue loops represent primal and dual contractible loops, respectively. The numbers inside the qubits indicate the corresponding measurement inputs, while the target qubit (a) $t = V(2,1)$ and (b) $t = H(1,2)$ is highlighted by a square. Parties without explicit inputs receive $1$ and do not participate in the submeasurement.
    }
    \label{fig:measurements}

\end{figure*}

These submeasurements are obtained from the product of contractible loops shown in Fig.~\ref{fig:measurements-vertical}. For this, the parties located at the intersections of the loops would in principle be applying $XZ$ or $ZX$, depending on the product order. Since only a single setting can be assigned to a party, we replace these intersections by a $Y$. In the RE, each replacement introduces a phase factor of $\pm i$, where the sign is determined by the order of the applied operators. As the intersections occur in pairs, these phases combine to an overall factor of $-1$.
Consequently, for any state $\ket{\psi}$ in the toric code subspace, these submeasurements satisfy
\begin{align}
    &S_1 \ket{\psi} = -\ket{\psi}, \\
    &S_2 \ket{\psi} = +\ket{\psi},\\
    &S_3 \ket{\psi} = -\ket{\psi}, \\
    &S_4 \ket{\psi} = -\ket{\psi}.
\end{align}

Therefore, any physical experiment simulating the behavior of these submeasurements must satisfy
\begin{equation}
    \big\langle\widetilde{S}_{k} \big\rangle = \left\langle {S}_k \right\rangle  = \lambda_k, 
    \label{eq:anticomm-expectations-0-rounds}
\end{equation}
for $k=1,2,3,4$, where $\lambda_2 = 1$ and $\lambda_1 = \lambda_3 =\lambda_4 = -1$.
Since every party performs an operator with spectrum $\pm1$, if the PE satisfy Eq.~\eqref{eq:anticomm-expectations-0-rounds}, we have
\begin{equation}
    \widetilde{S}_{k} \ket{\Psi}=\lambda_k \ket{\Psi} \quad \forall k ,
\label{eq:stab-cond}
\end{equation}
which implies
\begin{equation}
    \prod_{k=1}^4 \widetilde{S}_k \ket{\Psi}= \prod_{k=1}^4 \lambda_k \ket{\Psi} = -\ket{\Psi}.
    \label{eq:minus-sign}
\end{equation}

On the other hand, let us compute the product $S_1S_2S_3S_4$ explicitly. For the parties not participating in any submeasurement, this product is trivial and equals the identity. Now, for the other parties, we have
\begingroup
\newcommand{\pstr}[4]{%
    \makebox[.75em]{$#1$}%
    \makebox[.75em]{$#2$}%
    \makebox[.75em]{$#3$}%
    \makebox[.75em]{$#4$}}
\setlength{\arraycolsep}{0pt}
\renewcommand{\arraystretch}{0.95}
\begin{equation}
\begin{array}{r@{:\,\,}l@{\quad\;}r@{:\,\,}l}
    H(1,0)&\pstr{Z}{\Id}{\Id}{Z},
    &V(1,0)&\pstr{Z}{\Id}{\Id}{Z},\\
    H(1,1)&\pstr{X}{X}{Y}{Y},
    &V(2,0)&\pstr{Y}{X}{X}{Y},\\
    H(2,1)&\pstr{X}{X}{X}{X},
    &V(1,1)&\pstr{Z}{\Id}{Z}{\Id},\\
    H(1,2)&\pstr{Y}{\Id}{Y}{\Id},
    &V(2,2)&\pstr{X}{\Id}{X}{\Id},\\
    H(2,2)&\pstr{X}{\Id}{X}{\Id},
    &V(2,1)&\pstr{Z}{X}{Z}{X} .
\end{array}
\label{eq:cancelations}
\end{equation}
\endgroup
With the exception of $V(2, 1)$, all other Pauli strings above evaluate to $\Id$, therefore
\begin{equation}
    \prod_{k=1}^{4} S_k = Z_{V(2,1)}X_{V(2,1)}Z_{V(2,1)}X_{V(2,1)}.
    \label{eq:trusted-alt}
\end{equation}
Crucially, by Naimark dilation, we may assume without loss of generality that the dichotomic measurements performed in the PE are projective \cite{Supic2020-ox}. Their associated observables therefore satisfy $\widetilde{X}_e^2 = \widetilde{Y}_e^2= \widetilde{Z}_e^2 = \Id$. What follows is that Eq.~\eqref{eq:trusted-alt} also holds for the PE
\begin{equation}
    \prod_{k=1}^{4} \widetilde{S}_k = \widetilde{Z}_{V(2,1)}\widetilde{X}_{V(2,1)}\widetilde{Z}_{V(2,1)}\widetilde{X}_{V(2,1)}.
    \label{eq:untrusted-alt}
\end{equation}

Taking Eqs.~\eqref{eq:minus-sign} and \eqref{eq:untrusted-alt} together, and multiplying by $\widetilde{X}_{V(2,1)} \widetilde{Z}_{V(2,1)}$ on the left, we can conclude that
\begin{equation}
        \{\widetilde{Z}_{V(2,1)},\widetilde{X}_{V(2,1)}\}\ket{\Psi} = 0.
\end{equation}

Therefore, these four submeasurements are able to certify anticommuting observables for the party at $V(2,1)$. Because the lattice is periodic, every party on a vertical edge can also have its anticommuting operators certified with a translated version of these four submeasurements. Furthermore, a very similar pattern can be constructed for parties lying on horizontal edges, as we show in Fig.~\ref{fig:measurements-horizontal}. Thus, for a target party $t$, we denote by $S^{(t)}_k$ the four submeasurements used for these certifications, which are enough to conclude the first requirement of Proposition~\ref{prop: self-testing-criteria}.

To prove these operators form the stabilizer, it is sufficient to note that the measurements ${S}^{(t)}_2$ correspond exactly to a star or plaquette, depending on whether $t$ is a vertical or horizontal edge. Moreover, every stabilizer element corresponds to one such submeasurement. By Eq.~\eqref{eq:stab-cond}, the simulation of these measurements for all $t$ also implies that the operators $\widetilde{X}_e$ and $\widetilde{Z}_e$ define all the stabilizers of the subspace. Finally, the loops shown in Fig.~\ref{fig:measurements} do not rely on the specific boundary conditions. Therefore, our results extend to any lattice containing a $3\times3$ patch on which the same construction can be implemented. Together with the previous results, this allows us to conclude the following statement.

\begin{result}
    Let $\ket{\psi}$ be any pure state in the subspace of a toric code of lattice of size $L \geq 3$. Let $\left (\ket{\Psi} , \{\widetilde{M}^{(t)}_k\}_{k,t}\right)$ be any compatible physical experiment that simulates the reference experiment $\left (\ket{\psi}, \{{M}^{(t)}_k\}_{k,t}\right)$, whose submeasurements are translated versions of the ones shown in Fig.~\ref{fig:measurements}. Then the statistics $\left\{\left\langle \widetilde{S}^{(t)}_k\right\rangle\right\}_{t,k} $ self-test the toric code subspace.
    \label{res:not-com-res}
\end{result}

\begin{remark}
    The statistics used in Result~\ref{res:not-com-res} can equivalently be collected into a Bell inequality by summing, over all target edges, the tests introduced above with the signs of their ideal expectation values. The toric code reference experiment attains the algebraic maximum of this expression.
\end{remark}

\section{Self-testing against classical communication \label{sec:comm}}

We are now ready to extend the self-test to the case where the untrusted parties can communicate classically with each other. The scenario is closely related to the one introduced in the previous section, with the key difference that each untrusted party is now allowed to exchange secret classical messages synchronously with some other parties over $r$ rounds of communication. The allowed secrete communication is described by a given \textit{communication graph}, which encodes how the parties are spatially distributed, and where the graph distance between two parties determines the number of rounds of communication needed for them to exchange messages.
We stress that the honest implementation of the protocol (the experiment actually performed in the lab) uses no communication at all. The goal is to show that, even when the devices are untrusted and might have communicated secretly, they are still forced to share a state in the code space of the toric code, up to local isometries.

This communication model differs from the dishonest-party setting considered in \cite{Murta2023}, where a subset of the parties may act dishonestly, communicate among themselves, and perform arbitrary joint quantum operations. In that setting, the systems held by the honest parties can be certified individually, whereas the systems held by the dishonest parties are certified only collectively, through an extraction isometry acting jointly on all systems within the group. In our setting, instead, communication is restricted to classical messages exchanged within a bounded distance on a graph, while the self-testing conclusion retains a separate local isometry for every party. Consequently, a qubit of the target resource is extracted locally at each party.

The bounded classical-communication model we consider in this paper was first studied in the context of nonlocality using graph states \cite{Barrett2007} and more recently in \cite{Meyer2023, Meyer2026}.
There, a natural choice for the communication graph was the same as the underlying graph of the state.
Since here we associate the parties with the edges of a lattice, the natural definition for the neighborhood of each party is in terms of the closest lattice edges, namely, those that are adjacent to it along a common face; see Fig.~\ref{fig:comm-graph-diagram}.
More precisely, a subsystem lying on a vertical edge $V(i,j)$ shares an edge in the communication graph with the parties 
\begin{equation}
\begin{aligned}
    &H(i-1,j+1), \;\; &&H(i,j+1), \\
    &H(i-1,j), &&H(i,j),
\end{aligned}
    \label{eq:neigh-hor}
\end{equation} 
and an horizontal edge $H(i,j)$ has the neighbors 
\begin{equation}
\begin{aligned}
    &V(i,j), \;\; &&V(i+1,j), \\
    &V(i,j-1), &&V(i+1,j-1).
\end{aligned}
    \label{eq:neigh-ver}
\end{equation}

In this model, the operators performed locally by each party in the PE depend not only on the input assigned to that party by the referee, but also on the inputs received by all parties within graph distance $r$. Consequently, each observable is also indexed by the collection of inputs in this neighborhood. To simplify the notation, we denote the operators performed by the party located on edge $e$ by $\widetilde{X}^{\Gamma}_e$, $\widetilde{Y}^{\Gamma}_e$, and $\widetilde{Z}^{\Gamma}_e$ when its assigned input is $0$, $1$, and $2$, respectively, and where
\begin{equation}
    \Gamma=\{x_\ell \mid d(e,\ell)\leq r\}
\end{equation}
denotes the collection of inputs assigned to all parties within graph distance $r$ of $e$.
We will refer to $\Gamma$ as the \emph{context} of the observable. Note that in the definition of $\Gamma$ we omit its dependence on $r$ and $e$, as they can be inferred by the setting.

In the RE, the parties do not communicate, and therefore their local operations cannot depend on the inputs received by neighboring parties.
Nevertheless, so that we can use the same notation for the RE and the PE, we artificially enlarge the set of measurement settings in the RE by introducing a context label that has no operational effect. Thus, for any party $e$ and any two contexts $\Gamma$ and $\Gamma'$, the corresponding characterized observables satisfy
\begin{equation}
    X^\Gamma_e = X^{\Gamma'}_e =X_e,
    \label{eq:irrelevant-context}
\end{equation}
and similarly for $Y$ and $Z$.

\begin{figure}[t]
     \centering
     \includegraphics[width=0.7\linewidth]{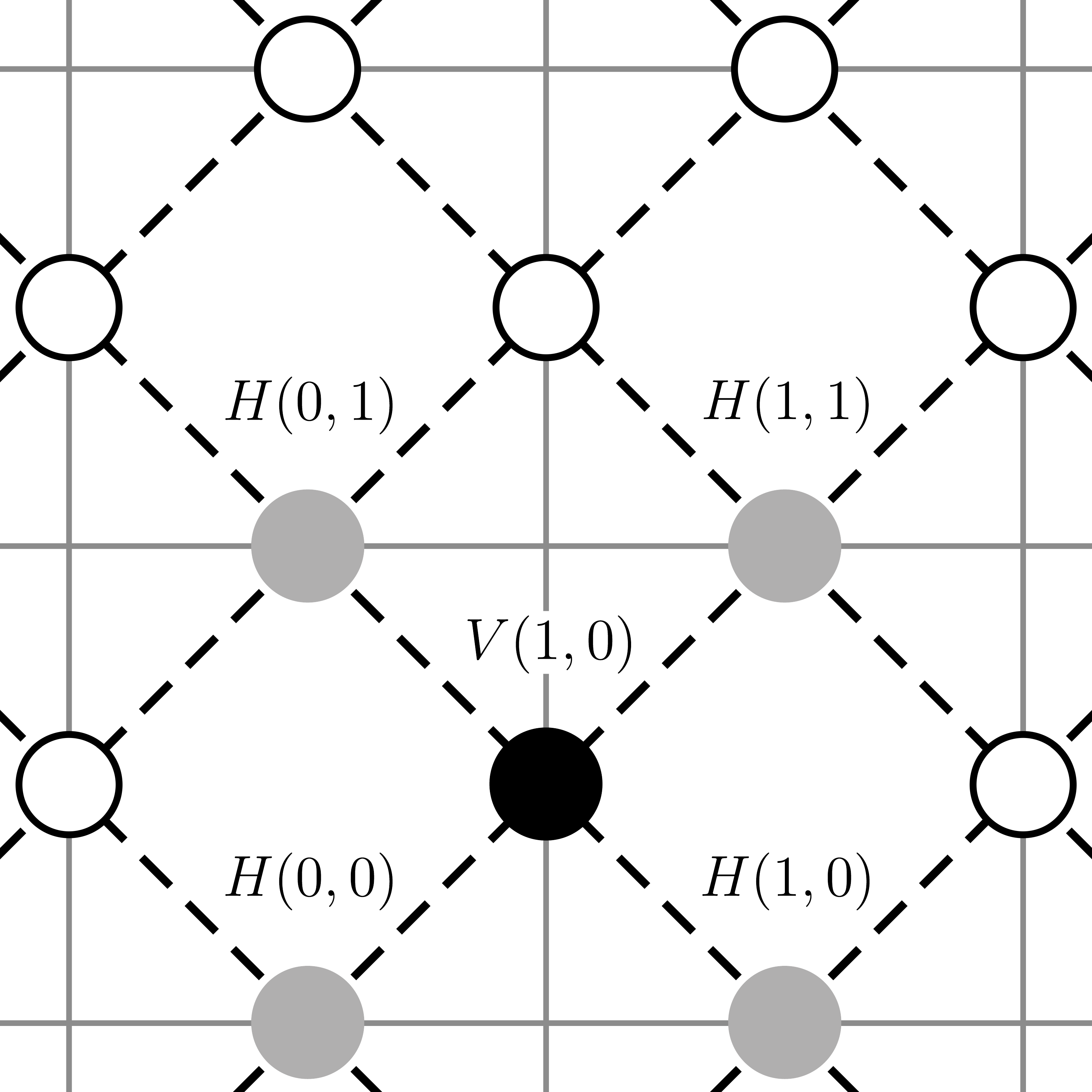}
    \caption{Communication graph (black dashed lines) for the parties on the toric code lattice (gray lines). The four neighbors of the party located at $V(1,0)$, shown in black, are highlighted in gray.}
    \label{fig:comm-graph-diagram}
 \end{figure}
 
Before we start constructing the protocols, it is instructive to understand why the self-test in Section~\ref{sec:without-comm} is not robust to even $1$ round of communication.
In the non-communicating case, each party performs the same local operation whenever it receives the same input, regardless of the inputs assigned to neighboring parties. The primal and dual loops shown in Fig.~\ref{fig:measurements} were chosen so that each local observable appears an even number of times across the four submeasurements. Moreover, for every party except the target, these occurrences are ordered so that identical operators can be paired and squared to the identity when the four submeasurements are multiplied together, as in Eq.~\eqref{eq:cancelations}.

With communication, however, a party's local operation can also depend on its context. So, although its own input remains the same in the paired submeasurements, the surrounding inputs can differ, allowing the party to distinguish these occurrences and perform different operations. The corresponding operators therefore need not coincide, and the cancellations used in the non-communicating proof no longer hold. The same problem arises if one allows the parties to communicate in the self-testing protocols proposed in \cite{Baccari2020, Hart2025}; see Appendix~\ref{sec:communicating-classical-strategies} for a discussion.

To make the protocol robust to communication, we therefore need to construct submeasurements in which both the local input and the accessible context coincide whenever the argument requires the same observable to appear. We achieve this by defining larger loops and suitably choosing the inputs of parties not participating in the submeasurement. Although the referee discards these parties' outputs, their inputs allow us to match the required contexts, which prevents the relevant parties from distinguishing between the paired submeasurements, and preserves the cancellations used in the non-communicating proof.

Let us begin by constructing submeasurements capable of self-testing in the presence of $1$ round of communication. Consider a lattice of size $L=9$. Within this lattice, we define a coarse-grained sub-lattice of size $L'=3$, whose horizontal and vertical edges are given by
\begin{align}
    H'(i,j)  = \{(3i,3j),(3i+3,3j)\},
\end{align}
and 
\begin{align}
    V'(i,j) = \{(3i,3j),(3i,3j+3)\},
\end{align}
where the vertices on the right-hand side are with respect to the original lattice (Fig.~\ref{fig:measurements-1-round}). We can then reproduce, on the coarse-grained lattice, the same primal and dual loops used in the construction of the previous section and shown in Fig.~\ref{fig:measurements-vertical}, and define the submeasurements $S^{(t)}_k$ accordingly. Every coarse-grained edge belonging to one of these loops corresponds to three consecutive edges of the finer lattice. More precisely, if a horizontal coarse-grained edge $H'(i,j)$ belongs to a primal loop $\gamma$, then
\begin{align}
H(3i,3j), H(3i+1,3j), H(3i+2,3j) \in \gamma.
\end{align}
Likewise, if the same horizontal coarse-grained edge belongs to a dual loop $\gamma^\star$, then
\begin{equation}
\begin{aligned}
H(3i+1,&3j-1), H(3i+1, 3j), \\
&\quad H(3i+1,3j+1) \in \gamma^\star.
\end{aligned}
\end{equation}
The corresponding definitions for vertical coarse-grained edges are obtained analogously.

The effect of the larger loop operators is to distance the parties in the submeasurement. If we now follow the same steps as in the non-communicating construction, the four submeasurements defined by these loop operators still satisfy 
\begin{align}
    &S_1^{(t)} \ket{\psi} = -\ket{\psi}, \\
    &S_2^{(t)} \ket{\psi} = +\ket{\psi},\\
    &S_3^{(t)} \ket{\psi} = -\ket{\psi}, \\
    &S_4^{(t)} \ket{\psi} = -\ket{\psi}.
\end{align}
and
\begin{equation}
    S^{(t)}_1 S^{(t)}_2S^{(t)}_3S^{(t)}_4 = Z_t^{\Gamma}X_t^{\Gamma}Z_t^{\Gamma}X_t^{\Gamma}.
\label{eq:missleading-eq}
\end{equation}

We can now analyze what happens when these submeasurements are performed in the PE. 
First, recall that in the RE, the cancellations leading to the above equation happen because the reference observables do not depend on the context. In the PE, however, a pairing such as $\widetilde{X}^{\Gamma}_e\widetilde{X}^{\Gamma'}_e$ reduces to the identity only if $\Gamma = \Gamma'$. Hence, we must choose the inputs so that every party has the same context in each pair of submeasurements where the operators must cancel.
For most parties, the larger loops already ensure that the contexts coincide across the paired submeasurements. For example, the party at $H(4,0)$ has identical contexts in submeasurements 1 and 3, allowing its paired operators to cancel, as shown in Fig.~\ref{fig:measurements-1-round}. These parties therefore cannot distinguish the relevant submeasurements, so their untrusted operators in the PE cancel through pairings analogous to those in Eq.~\eqref{eq:cancelations}. For some parties, however, the paired contexts still differ, as at $V(3,2)$ in submeasurements 1 and 4 (see Fig.~\ref{fig:measurements-1-round}). Thus, Eq.~\eqref{eq:missleading-eq} does not directly hold in the PE. Fortunately, we can modify the inputs of parties outside these loops so that the remaining paired contexts also coincide.

\begin{figure*}[t]
    \centering

    % ========================================================
    % FOUR SUBMEASUREMENTS
    % ========================================================
    \begin{minipage}[c]{0.74\textwidth}
        \centering

        \includegraphics[width=0.485\linewidth]
        {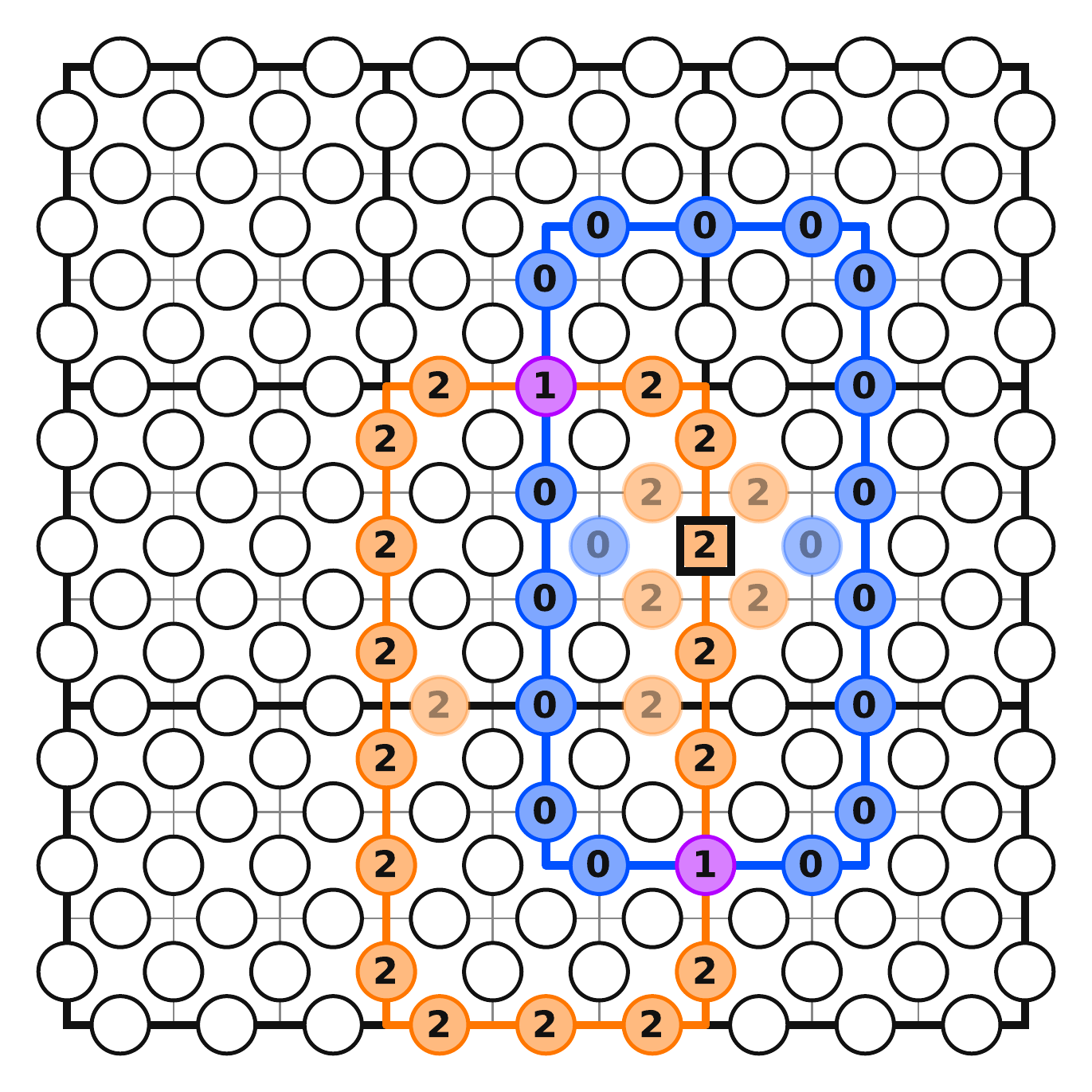}
        \hfill
        \includegraphics[width=0.485\linewidth]
        {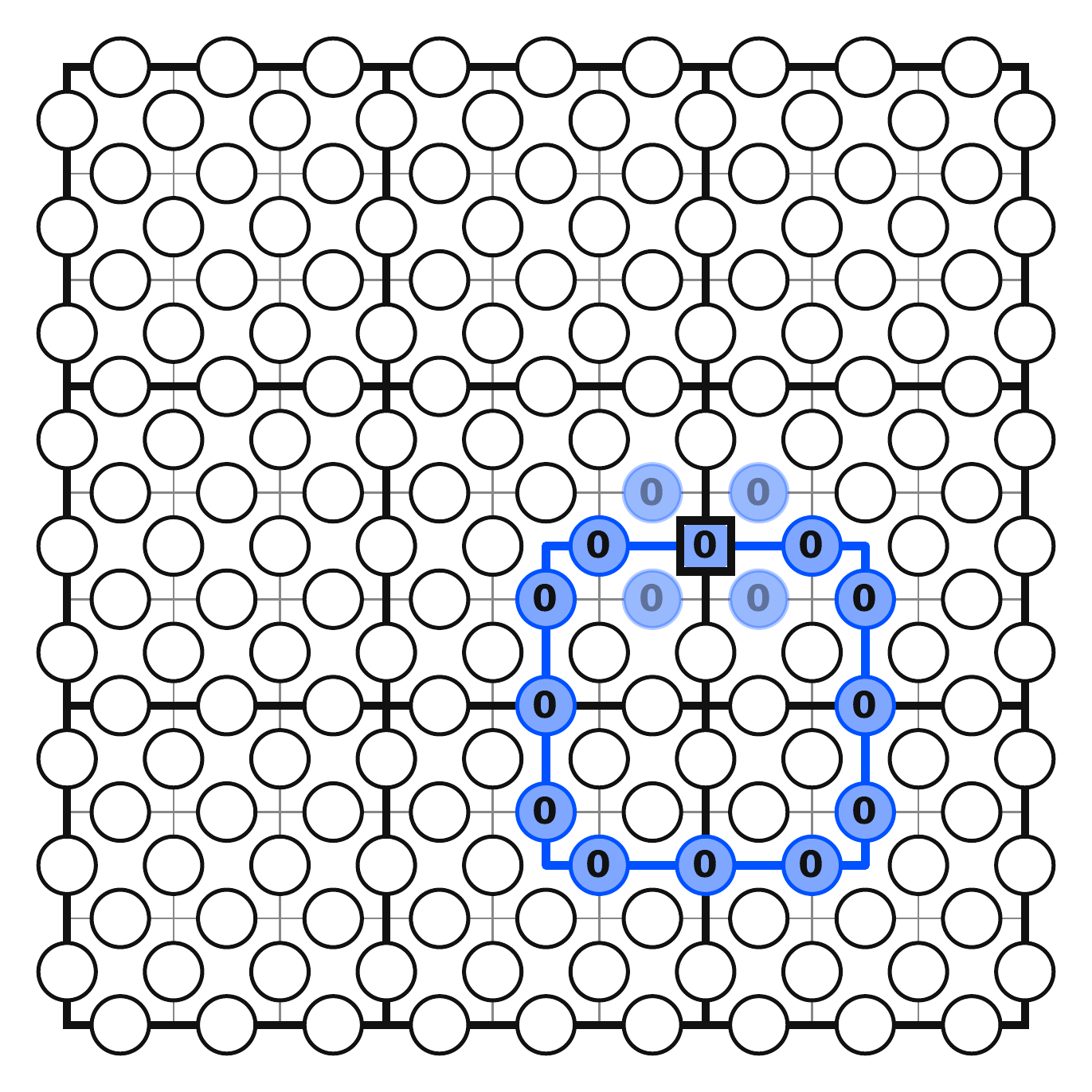}

        \vspace{0.2cm}

        \includegraphics[width=0.485\linewidth]
        {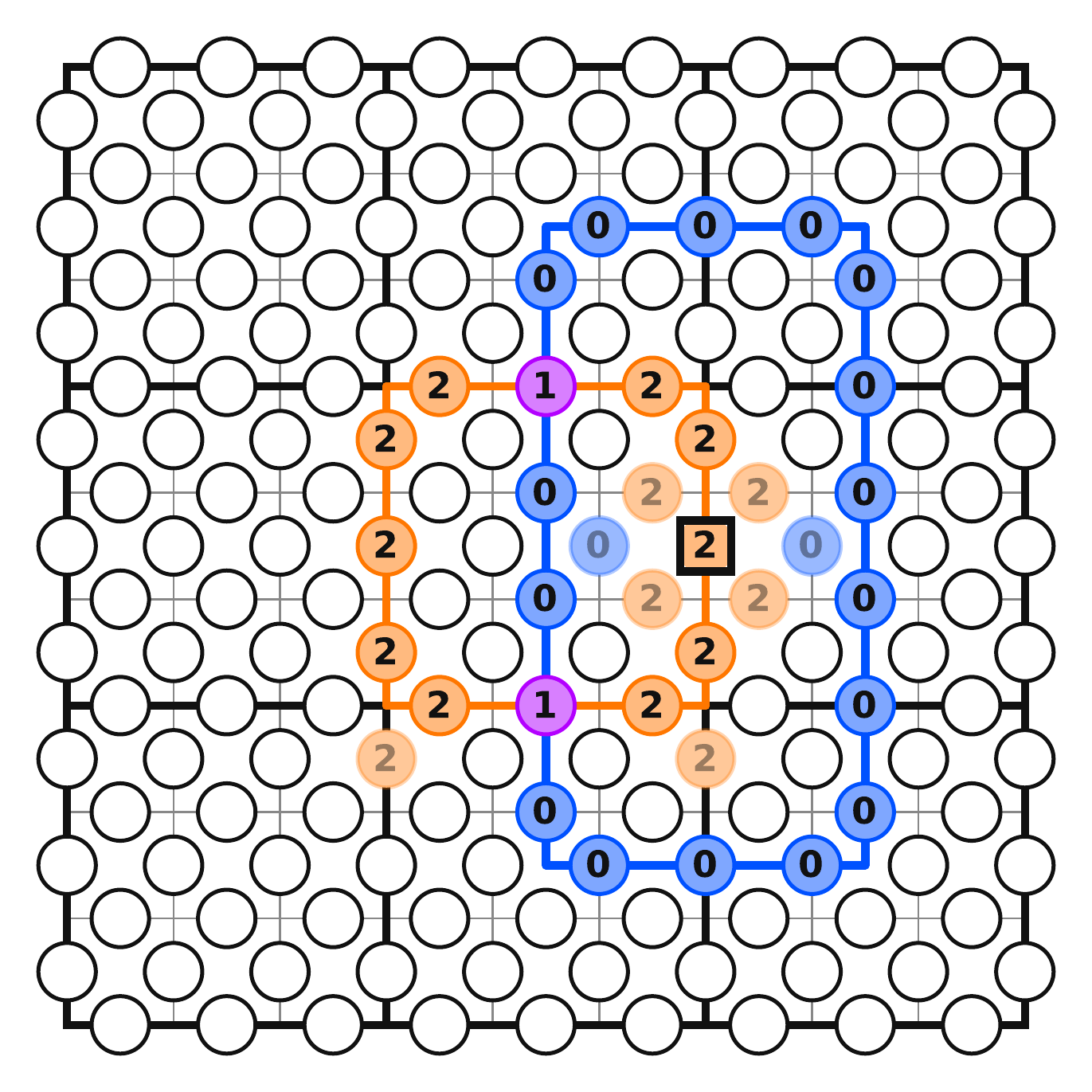}
        \hfill
        \includegraphics[width=0.485\linewidth]
        {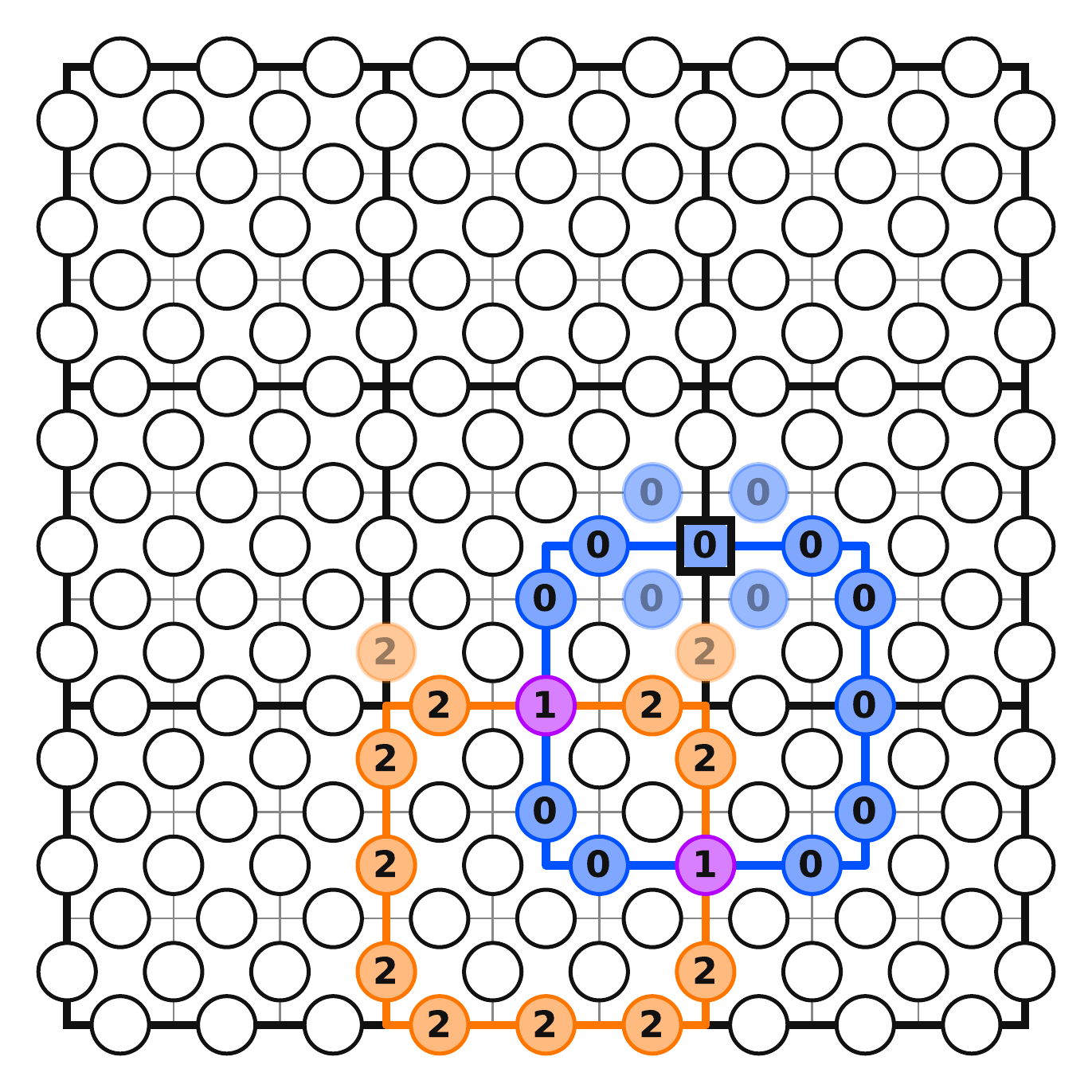}

    \end{minipage}
    \hfill
    % ========================================================
    % LEGEND
    % ========================================================
    \begin{minipage}[c]{0.25\textwidth}
        \centering

        \includegraphics[
            width=\linewidth
        ]{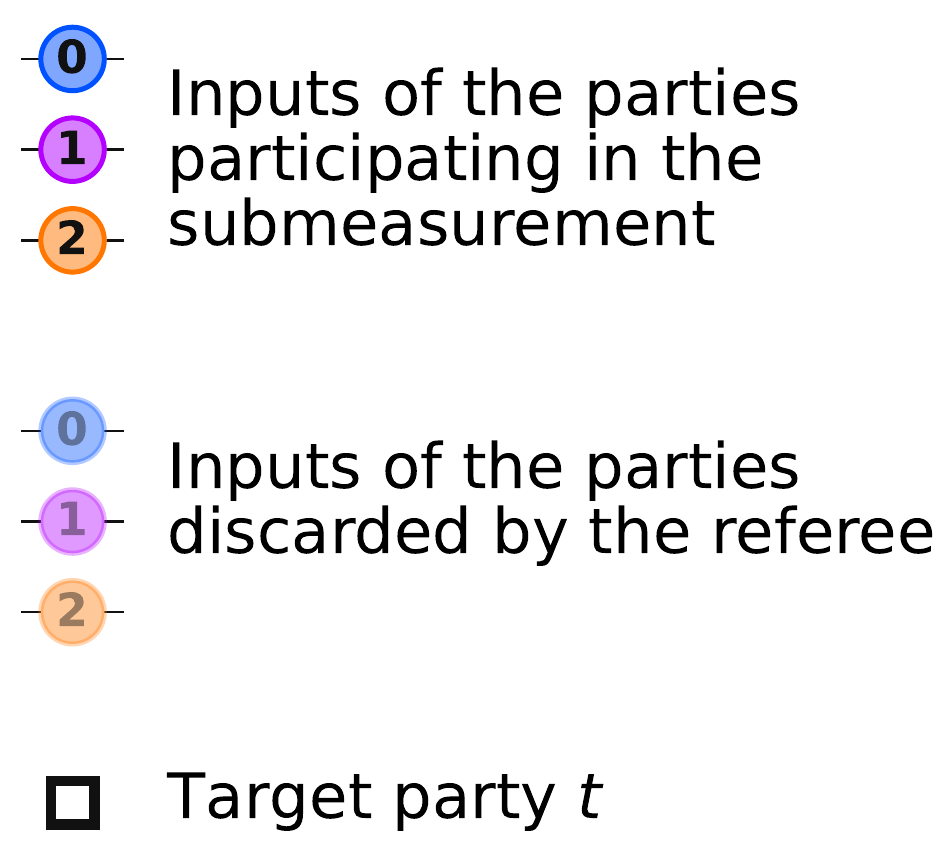}

    \end{minipage}

    \caption{
    A $9\times 9$ patch of the toric code showing the inputs assigned by the referee in measurements 1 through 4, ordered left to right and top to bottom, to certify the anticommutation relations. The fine lattice is shown in gray, while the coarse-grained lattice is highlighted in black. Orange and blue curves represent primal and dual contractible loops, respectively. The numbers displayed on the qubits indicate the corresponding measurement inputs, and the target qubit $t=V(6,4)$ is highlighted by a square. The parties included in each submeasurement are those associated with the primal and dual loops. Faded inputs, outside of the loops, correspond to the additional inputs assigned to parties outside the submeasurement to make the relevant distance-1 contexts coincide across different measurements. Parties with no inputs do not take part in the submeasurement nor are relevant for matching any context and, by convention, receive an input 1.
    }
    \label{fig:measurements-1-round}
\end{figure*}

Let us analyze the parties for which this cancellation fails. For that, we use as an example the submeasurements associated with the target party $t = V(6,4)$, shown in Fig.~\ref{fig:measurements-1-round}. We start with the party located at $H(4,4)$, and the same reasoning also applies for the one in $H(7,4)$. In the four submeasurements, this party is located in a dual loop, and hence receives the input $0$. However, the dual loop appearing in the submeasurements alternates between two possible paths
\begin{align}
    \gamma^\star_1 = \{V'(2,0)&, H'(2,1), H'(2,2), \nonumber \\ &V'(2,2), H'(1,2), H'(1,1)\},    
\end{align}
and 
\begin{equation}
    \gamma^\star_2 = \{V'(2,0), H'(2,1), V'(2,1), H'(1,1)\}.
\end{equation} 
In the two cases, the distance-1 context of $H(4,4)$ is the same, with the exception of the input assigned to the party at $V(5,4)$. In submeasurements 2 and 4, $V(5,4)$ belongs to $\gamma^\star_2$ and therefore receives input $0$, revealing to $H(4,4)$ that the dual loop continues through that neighboring edge. By contrast, in submeasurements 1 and 3, $V(5,4)$ does not belong to any loop and receives input $1$. Thus, after a single round of communication, $H(4,4)$ can determine from the input of $V(5,4)$ which of the two possible loops is being implemented. Consequently, even though $H(4,4)$ receives the same input in all four submeasurements, its context is different, and the corresponding local observables need not coincide. 

To match its context in all four submeasurements, it suffices to change the input assigned to $V(5,4)$ in the measurements $M^{(t)}_1$ and $M^{(t)}_3$ from a $1$ to a $0$, even if the referee disregards its output. In this way, the context visible to $H(4,4)$ in 1 round no longer locally reveals which direction the dual loop follows, as it is the same across the four submeasurements, yielding 
\begin{equation}
    \left( \widetilde{S}_1^{(t)}\widetilde{S}_2^{(t)}\widetilde{S}_3^{(t)}\widetilde{S}_4^{(t)}\right)_{H(4,4)} =\Id.
\end{equation}

We next consider the parties $H(3,3)$, $V(3,2)$, and $V(3,3)$, whose contexts also need to be modified. The same reasoning applies to $H(5,3)$, $V(6,2)$, and $V(6,3)$. In this case, the problem involves three possible primal loops
\begin{equation}
\begin{aligned}
    \gamma_1 = \{H'(1,0)&,V'(2,0),V'(2,1), \\&H'(1,2),V'(1,1),V'(1,0)\}, 
\end{aligned}
\end{equation}
\begin{equation}
    \gamma_2 = \{H'(1,1),V'(2,1),H'(1,2),V'(1,1)\}, 
\end{equation}
and 
\begin{equation}
    \gamma_3 = \{H'(1,0),V'(2,0),H'(1,1),V'(1,0)\}, 
\end{equation}
associated with submeasurements 1, 3, and 4, respectively. Each of the parties $H(3,3)$, $V(3,2)$, and $V(3,3)$ belongs to two of these three loop configurations and is not present in the remaining one. Importantly, $H(3,3)$ lies within the context of both $V(3,2)$ and $V(3,3)$. Hence, these two parties can distinguish between loops from the input assigned to $H(3,3)$. Conversely, $H(3,3)$ can also distinguish them by observing which of $V(3,2)$ and $V(3,3)$ receives input $2$. Therefore, to make the relevant contexts coincide across the corresponding submeasurements, we assign an input $2$ to all three parties in the measurements $M^{(t)}_1$, $M^{(t)}_3$, and $M^{(t)}_4$, regardless of whether their outputs are included in the submeasurement. This makes the direction followed by the loop indistinguishable for these parties, which recovers
\begin{equation}
\left(\widetilde{S}^{(t)}_1\widetilde{S}^{(t)}_2\widetilde{S}^{(t)}_3\widetilde{S}^{(t)}_4\right)_{e} = \Id,
\end{equation}
for $e \in \{ H(3, 3), V(3, 2), V(3, 3) \}$.

With these new input assignments, we covered all the parties whose cancellation relied on extra assumptions about the observables. 

Furthermore, as in the non-communicating case, this pattern is valid for any vertical edge by just translating the coordinates. To certify anticommutation on the horizontal edges, one could use the same idea, applied to the loops in Fig.~\ref{fig:measurements-horizontal}. However, in Section~\ref{sec:without-comm}, these loops were useful because the referee could use them to also certify all the plaquette stabilizers, according to Proposition~\ref{prop: self-testing-criteria}. Since none of the larger loops we defined for the communicating case corresponded to a star or a plaquette stabilizer, this is no longer true. We can therefore simply use the same loops shown in Fig.~\ref{fig:measurements-1-round} rotated by 90°, transforming every horizontal edge into a vertical edge and vice versa, so that the same analysis we performed also applies. This finally implies that, for every target party $t$, we have
\begin{equation}
    \widetilde{S}^{(t)}_1 \widetilde{S}^{(t)}_2\widetilde{S}^{(t)}_3\widetilde{S}^{(t)}_4 = \widetilde{Z}_t^{\Gamma}\widetilde{X}_t^{\Gamma}\widetilde{Z}_t^{\Gamma}\widetilde{X}_t^{\Gamma}.
    \label{eq:anticomm-rel-1-round-bad-context}
\end{equation}

To use Proposition~\ref{prop: self-testing-criteria}, it remains to define submeasurements whose simulation implies that the same observables in Eq.~\eqref{eq:anticomm-rel-1-round-bad-context} also define the stabilizers of the toric code subspace. So far, the context of the operators on the right-hand side of Eq.~\eqref{eq:anticomm-rel-1-round-bad-context} is composed of all the parties being assigned the input $1$. However, a star or plaquette stabilizer containing a certain party $e$ requires that exactly two of the parties in $e$'s distance-1 context also appear in that stabilizer. This makes it impossible for the referee to use a single submeasurement to certify the stabilizers using these observables, as these contexts conflict with each other.

To address this, we call the \textit{canonical contexts} as those associated with the observables whose anticommutation we aim to certify and which also define the toric code stabilizers. We denote these contexts by $\Gamma_X$ and $\Gamma_Z$. A convenient choice for them is to assign
\begin{equation}
    \Gamma_X = \{x_\ell = 0 \mid d(\ell,t) \leq r \},
\end{equation}
\begin{equation}
    \Gamma_Z = \{x_\ell = 2 \mid d(\ell,t) \leq r \},
\end{equation}
as the referee can now compute the statistics of the submeasurements
\begin{equation}
    S^{(v)} = \prod_{\partial e \ni v} X_e^{\Gamma_X} \quad \forall v,
    \label{eq:star-submeasurement}
\end{equation}
and 
\begin{equation}
    S^{(f)} = \prod_{e \in \partial f} Z_e^{\Gamma_Z} \quad \forall f,
    \label{eq:plaquette-submeasurement}
\end{equation}
without any conflict. To define the full measurements $M^{(v)}$ and $M^{(f)}$, we assign the input $1$ to every party not participating in these submeasurements nor in the canonical contexts.

One can see that these inputs can also be assigned in the four measurements certifying the anticommutation relation without interfering with Eq.~\eqref{eq:anticomm-rel-1-round-bad-context}. The complete input assignment defining $M^{(t)}_k$ can be found in Fig.~\ref{fig:measurements-1-round}.

Finally, we can follow the same argument used in the previous section. The four submeasurements $S^{(t)}_k$ still obey
\begin{equation}
    \left \langle S^{(t)}_k \right \rangle = \lambda_k,
\end{equation}
where $\lambda_1 = \lambda_3 = \lambda_4 = -1$ and $\lambda_2 = 1$. Hence, if a PE simulates these submeasurements, then
\begin{equation}
    \left\{\widetilde{Z}^{\Gamma_Z}_t,\widetilde{X}^{\Gamma_X}_t\right\} \ket{\Psi} = 0.
\end{equation}
Furthermore, if it also simulates the submeasurements $S^{(v)}$ and $S^{(f)}$ for every vertex $v$ and face $f$ in the lattice, we have
\begin{align}
    \prod_{\partial e \ni v} \widetilde{X}_e^{\Gamma_X} \ket{\Psi} = +\ket{\Psi}, \\
     \prod_{e \in \partial f} \widetilde{Z}_e^{\Gamma_Z} \ket{\Psi} = +\ket{\Psi}.
\end{align}
Thus, we can use Proposition~\ref{prop: self-testing-criteria} to show that the statistics arising from these submeasurements self-test the toric code subspace and are robust to 1 round of classical communication.

A similar reasoning applies for more rounds of communication. Indeed, in Appendix~\ref{app:generalization} we show that for any lattice of size $L \geq 6r+3$, an analogous construction can be made to self-test the subspace against $r$ rounds of communication.

To conclude, let us denote with $\mathcal{M}_r$ the set of measurements $\bigl\{ \{M^{(t)}_k\}_{k,t},\{M^{(v)}\}_v,\{M^{(f)}\}_f\bigr\}$, where $M^{(t)}_k$ correspond to the generalization of the measurements shown in Fig.~\ref{fig:measurements-1-round} to $r$ rounds of communication. Then, from the above discussion and Appendix~\ref{app:generalization} follows our main result:
\begin{result}
    Let $\ket{\psi}$ be any pure state in the subspace of a toric code of lattice of size $L \geq 6r+3$. Let $\left (\ket{\Psi} ,\widetilde{\mathcal{M}}_r\right)$ be any compatible physical experiment that simulates the reference experiment $\left (\ket{\psi},\mathcal{M}_r\right)$. Then the statistics 
    \begin{equation}
        \left\{\left\{ \big\langle \widetilde{S}^{(t)}_k \big\rangle\right\}_{t,k},\left\{\big\langle \widetilde{S}^{(v)}\big\rangle\right\}_{v},\left\{\big\langle \widetilde{S}^{(f)}\big\rangle\right\}_{f}  \right\}
    \end{equation} 
    provide a self-testing protocol for the toric code subspace robust to $r$ rounds of classical communication. 
    \label{res:comm-res}
\end{result}

Notably, this construction is asymptotically optimal. Indeed, let $L_{\min}(r)$ denote the minimum side length of a toric code lattice for which the code subspace can be self-tested when the parties are allowed $r$ rounds of communication. The analytical construction of Result~\ref{res:comm-res} already establishes the upper bound $L_{\min}(r)=\mathcal{O}(r)$. Conversely, if $r$ were allowed to grow faster than the lattice size, there would exist a lattice for which the toric code could be self-tested even when every party has access to the inputs of all the other parties. However, such a construction is impossible, since in this case any probability distribution over the outcomes could be reproduced classically using shared randomness, and thus $L_{\min}(r)=\Omega(r)$. Combining the two bounds shows that our construction is asymptotically optimal.

\section{Noise robustness}
\label{sec:noise}
The self-testing protocols presented above all rely on reproducing exactly the expectation values of an ideal reference experiment. However, in any realistic experimental implementation, measurement statistics are inevitably affected by noise and experimental imperfections. Thus, to obtain experimentally meaningful results, one must study the robustness of the self-testing protocol to noise. The goal is to show that a small deviation from the ideal simulation conditions in Definition~\ref{def:simulation} still guarantees that the realized resources remain close to those of the reference experiment. 

To do so, we adapt Definition~\ref{def:simulation} and say that the PE $\epsilon$-simulates the RE if we have
\begin{equation}
    \left \lvert \bra{\Psi}\widetilde{S}_\mathbf{x} \ket{\Psi}  - \bra{\psi} {S}_\mathbf{x} \ket{\psi} \right \rvert \leq \epsilon,
\end{equation}
for all the submeasurements considered in the protocol.

To quantify how much a noisy experiment deviates from the ideal case, we use the distance of a pure state $\ket{\psi}$ to a subspace $\mathcal{C}$ given by $d(\ket{\psi},\mathcal{C}) = \left \lVert (\Id - \Pi_\mathcal{C}) \ket{\psi}\right \rVert_2$, where $\Pi_\mathcal{C}$ corresponds to the projector onto the subspace $\mathcal{C}$. For the toric code, this given by 
\begin{equation}
    \Pi_{TC} = \prod_v \frac{(\Id + A_v)}{2}\prod_f \frac{(\Id + B_f)}{2}.
\end{equation}
The following proposition gives a bound on this distance.
\begin{proposition}
Let
    \begin{equation}
    \left \{S^{(t)}_k\right\}_{k,t},\quad \left \{S^{(v)}\right\}_{v}, \quad \text{and}\quad  \left \{S^{(f)}\right\}_{f}    
    \end{equation}
    be the submeasurements used in Results~\ref{res:not-com-res} and \ref{res:comm-res}.
    For any PE $\epsilon$-simulating these submeasurements, there exists an isometry $\Phi = (\bigotimes_i \Phi_i)\otimes \Id_P$ such that
\begin{align}
     \lVert \Phi  \ket{\Psi}  - (\Pi_{TC}\otimes \Id_P ) \Phi \ket{\Psi}  \rVert_2 \leq \delta(\epsilon),
\end{align}
where
\begin{equation}
    \delta(\epsilon) = L^2\sqrt{\epsilon} \left(\sqrt{2} + 8\right). 
\end{equation}
\label{prop:robust-prop}
\end{proposition}
The proof of Proposition~\ref{prop:robust-prop} is a straightforward application of the techniques used in \cite{McKague2014} and can be found in Appendix \ref{app:proof-noise}.

\section{Algorithmic method for finding submeasurements \label{sec:MILP}}

In this section, we show how to systematically search for submeasurements whose statistics can be used to self-test a resource even when the parties have access to classical communication. We first present the formulation for the toric code, which can be adapted to other stabilizer subspaces admitting propositions analogous to Proposition~\ref{prop: self-testing-criteria}. We then illustrate this generality with explicit constructions for the surface-code subspace and the two-dimensional cluster state.

We begin by noting that the structure of the proof of Results \ref{res:not-com-res} and \ref{res:comm-res} is relatively straightforward. In both cases, the strategy consisted of finding measurements $M^{(t)}_k$ with submeasurements $S^{(t)}_k$ composed of products of the toric code stabilizers, up to a minus sign, satisfying
\begin{align}
    &\left\langle S^{(t)}_k\right\rangle = \lambda_k \quad \lambda_k \in \{\pm 1\}, \label{eq:eig-cond}\\
     &\prod_{k=1}^{4} \left(S^{(t)}_k\right)_t =  Z^{\Gamma_Z}_t X^{\Gamma_X}_t Z^{\Gamma_Z}_t X^{\Gamma_X}_t,\label{eq:op-alt} \\
     &\prod_{k=1}^4 \left(S^{(t)}_k\right)_e = \Id \quad \forall e \neq t, \label{eq:op-const}\\
    &\prod_{k=1}^4 \lambda_k = -1. \label{eq:prd-const}
\end{align}
The key observation is that Eq.~\eqref{eq:op-const} does not rely on the specific algebraic structure of Pauli operators, but only on the fact that the local observables are self-inverse, i.e., $\mathcal{A}^2=\Id$. Hence, the cancellation of the non-target parties can be understood purely in terms of how identical observables are paired across the four submeasurements, and the same mechanism can be applied to the uncharacterized observables of the PE. In particular, we only have three possible forms for Eq.~\eqref{eq:op-const}:
\begin{enumerate}[itemsep=-3pt]
    \item Consecutive pairings
    \begin{equation}
        \mathcal{A} \mathcal{A} \mathcal{B}\mathcal{B} = \Id;
        \label{eq:consecutive-pair}
    \end{equation}
    \item Identity-separated pairings
    \begin{equation}
        \mathcal{A} \Id \mathcal{A}\Id = \Id \quad \text{ or }  \quad \Id \mathcal{A} \Id \mathcal{A} = \Id;
        \label{eq:id-sep-pair}
    \end{equation}
    \item Nested pairings
    \begin{equation}
        \mathcal{A} \mathcal{B} \mathcal{B} \mathcal{A} = \Id;
        \label{eq:nested-pair}
    \end{equation}
\end{enumerate}
where $\mathcal{A},\mathcal{B}\in\{\Id, X^{\Gamma}_e,Y^{\Gamma}_e,Z^{\Gamma}_e\}$, and the identities appearing in the left-hand side of Eq.~\eqref{eq:id-sep-pair} arise from the definition of submeasurements, in which the outputs of some parties in the measurements can be disregarded by the referee. 

Finding submeasurements satisfying Eqs.~\eqref{eq:op-alt} and \eqref{eq:op-const} then reduces to a linear constraint problem over the parties' input assignments in a given measurement. To illustrate this, let $M^{(t)}_k$ denote the four measurements which, together with the sets $I^{(t)}_k$, define the submeasurements 
\begin{equation}
    S^{(t)}_k = \left. M^{(t)}_k \right \rvert_{I^{(t)}_k}
\end{equation}
satisfying Eqs.~\eqref{eq:eig-cond}-\eqref{eq:prd-const}, i.e., they certify the anticommutation relation for the target party $t$. Every non-target party $e$ appearing in at least one of the sets
$I_k^{(t)}$ must then satisfy Eq.~\eqref{eq:op-const}, which can be
enforced by requiring its four local observables to realize one of the
allowed pairings above.

For each measurement $M_k^{(t)}$, let $x_k^{(t)}(e)$ denote the input
assigned to party $e$, and introduce a binary variable
$\alpha_k^{(t)}(e)$ indicating whether $e$ participates in the
corresponding submeasurement
\begin{equation}
    \alpha_k^{(t)}(e) = \left\{ 
    \begin{aligned}
        1, \quad e \in I^{(t)}_k,\\
        0, \quad e \not\in I^{(t)}_k.
    \end{aligned}
    \right.
\end{equation}

Therefore, satisfying a pairing like $\mathcal{A}\Id\mathcal{A}\Id$ for a given party $e$ can be done by enforcing the following two constraints over these variables. First, one needs to enforce that the party $e$ appears in submeasurements 1 and 3, and does not appear in the other two
\begin{align}
    &\alpha_1^{(t)}(e) = \alpha_3^{(t)}(e) = 1, \\
    & \alpha_2^{(t)}(e) = \alpha_4^{(t)}(e) = 0.
    \label{eq:alpha-const}
\end{align}
Secondly, the operators appearing in submeasurements 1 and 3 must be the same. This amounts to matching not only the input of the party $e$, but also matching the input of every party at distance $r$ from $e$, i.e.
\begin{align}
    x_1^{(t)}(\ell) = x_3^{(t)}(\ell) \quad \forall \ell : d(e,\ell) \leq r.
    \label{eq:gamma-const}
\end{align}

Similar arguments can also be made for the other pairings, and the explicit linear formulations are given in Appendix~\ref{app:ILP-Formulation}. The choice of pairing itself can also be incorporated as linear constraints by introducing auxiliary binary variables enforcing that one allowed pairing is selected for each party appearing in the submeasurements. Finally, at the target party, Eq.~\eqref{eq:op-alt} is imposed by fixing the inputs in its distance-$r$ context according to the canonical context.

The remaining conditions, Eqs.~\eqref{eq:eig-cond} and \eqref{eq:prd-const}, can likewise be expressed through similar linear constraints. For simplicity, we do not write these constraints explicitly in the main text and instead refer to them in the form of Eqs.~\eqref{eq:eig-cond}-\eqref{eq:prd-const}. Their complete formulation can be found in Appendix~\ref{app:ILP-Formulation}.

The constraints above address the first condition of Proposition~\ref{prop: self-testing-criteria}, namely the certification of local anticommutation relations. To complete the formulation, we now turn to the second condition, which requires certifying the stabilizers of the toric code. So far, we only needed one submeasurement per stabilizer to certify them. This was due to the convenient choice of the canonical context that allowed all involved observables to be measured in the same submeasurement without conflicting contexts. However, in principle, one could use different canonical contexts for every party, and which might not obey this property. In those cases, a strategy similar to the one used for the anticommutation certificate can be used, where we define more measurements $M^{(v)}_j$ ($M^{(f)}_m$) and their corresponding submeasurements $S^{(v)}_j$ ($S^{(f)}_m$) that satisfy analogous constraints to those in Eqs.~\eqref{eq:eig-cond}-\eqref{eq:prd-const}, but for a given star (plaquette) stabilizer located at a vertex $v$ (face $f$). By the same argument as above, these can also be cast as linear constraints analogous to the previous case.

Finally, using only four submeasurements at each step is also not special; it is simply the minimum number needed to find the anticommutation relations. It could be the case that there exists more intricate loop constructions using more submeasurements that allow for solutions which are not possible using only four submeasurements.  We therefore can specify the parameters $K_t$, $J_v$, and $M_f$ corresponding to the numbers of submeasurements used to construct the anticommutation certificate for target $t$, the star stabilizer certificate for vertex $v$, and the plaquette stabilizer certificate for face $f$, respectively. Notably, in these cases, one also needs to adjust the possible pairings accordingly.

Therefore, we can search for submeasurements capable of self-testing by solving the following problem:
\begin{align}
    \textbf{find: }& 
    M^{(t)}_k,I^{(t)}_k, M^{(v)}_j,I^{(v)}_j, M^{(f)}_m,I^{(f)}_m
    \nonumber\\[1em]
    \textbf{such that:}& 
    \nonumber\\[2em]
    \substack{
        \textbf{Qubit}\\
        \textbf{certification}
    }&
    \left\{
    \begin{aligned}
        &  S^{(t)}_k = \left.\left(M^{(t)}_k \right)\right\rvert_{I^{(t)}_k} \quad \forall k \in [K_t] \\
        &\left\langle S^{(t)}_k \right\rangle = \lambda_k^{(t)} \quad \lambda_k^{(t)} \in \{\pm 1\} ,
        \\
        &\prod_{k=1}^{K_t} \left(S^{(t)}_k\right)_t = Z^{\Gamma_Z}_t X^{\Gamma_X}_t Z^{\Gamma_Z}_t X^{\Gamma_X}_t,
        \\
        &\prod_{k=1}^{K_t} \left(S^{(t)}_k\right)_e = \Id \quad \forall e \neq t,
        \\
        &\prod_{k=1}^{K_t} \lambda_k^{(t)} = -1.
    \end{aligned}
    \right. \forall t,
    \nonumber\\
        \substack{
            \textbf{Star stabilizer} \\
        \textbf{certification}
    }&
    \left\{
    \begin{aligned}
        &  S^{(v)}_j = \left.\left(M^{(v)}_j \right)\right\rvert_{I^{(v)}_j} \quad \forall j \in [J_v] \\
        & \left\langle S^{(v)}_j \right\rangle = \mu_j^{(v)}
        \quad \mu_j^{(v)} \in \{\pm 1\},
        \\
        & \prod_{j=1}^{J_v} \left(S_j^{(v)}\right)_e = X_e^{\Gamma_X}
    \quad \forall \partial e \ni v, 
    \\
    & 
     \prod_{j=1}^{J_v} \left(S_j^{(v)}\right)_e = \Id
    \quad \forall \partial e \not\ni v
    \\
    &
    \prod_{j=1}^{J_v} \mu_j^{(v)} = +1,
        \\
    \end{aligned}
    \right. \forall v,
    \nonumber\\
   \substack{
        \textbf{Plaquette}\\
        \textbf{stabilizer}\\ \textbf{certification}
    }&
    \left\{
    \begin{aligned}
        &  S^{(f)}_m = \left.\left(M^{(f)}_m \right)\right\rvert_{I^{(f)}_m} \quad \forall m \in [M_f] \\
        & \left\langle S^{(f)}_m \right\rangle = \chi_m^{(f)}
        \qquad \chi_m^{(f)} \in \{\pm 1\},
        \\
        & \prod_{m=1}^{M_f} \left(S_m^{(f)}\right)_e = Z_e^{\Gamma_Z}
    \qquad \forall e \in \partial f, 
    \\
    & 
     \prod_{m=1}^{M_f} \left(S_m^{(f)}\right)_e = \Id
    \quad \forall e \not\in \partial f,
    \\
    &
    \prod_{m=1}^{M_f} \chi_m^{(f)} = +1,
        \\
    \end{aligned}
    \right. \forall f.
\label{eq:opt-problem}
\end{align}

We emphasize that all constraints in this formulation can be written as linear constraints on the inputs assigned to each measurement. Hence, Eq.~\eqref{eq:opt-problem} defines an integer linear program (ILP), whose formulation can be found in Appendix~\ref{app:ILP-Formulation}.
This provides a systematic way of searching for submeasurements capable of self-testing the toric code using standard ILP solvers, such as Gurobi \cite{gurobi}. As we show next, this ILP can also be easily adapted to self-test other stabilizer subspaces.

Furthermore, in Appendix \ref{app:proof-noise}, we prove a more general form of Proposition~\ref{prop:robust-prop}, which can be used to show that every solution of Eq.~\eqref{eq:opt-problem} is noise robust. More precisely, its robustness is

\begin{equation}
    \delta(\epsilon) =  \frac{\sqrt\epsilon}{2} \left[\sum_v\left(\sqrt{2}J_v + \sum_{\partial e \ni v} K_e\right) + \sqrt{2} \sum_f M_f\right].
\end{equation}

In the remainder of this section, we illustrate the usefulness of this formulation through some applications.

\subsection{Minimal lattice size for bounded communication}
\label{sec:toric-mip}

\begin{figure}
    \centering
    \includegraphics[width=0.95\linewidth]{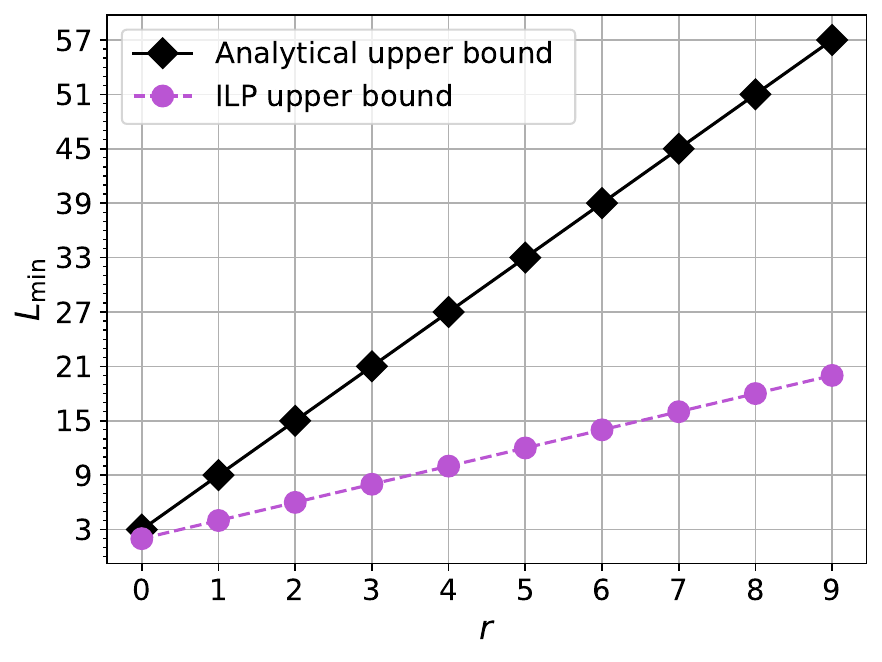}
    \caption{Upper bounds on the minimum lattice size required for self-testing against $r$ rounds of classical communication. Black diamonds correspond to the analytical construction presented in Section~\ref{sec:comm}, while purple circles indicate the minimum lattice sizes obtained numerically via the integer linear program. The corresponding explicit constructions are given in Appendix~\ref{app:minimal_lattices} and in \cite{froes2026code}.}
    \label{fig:Lmin-plot-tc}
\end{figure}

In Section~\ref{sec:comm}, we showed a protocol to self-test the toric code subspace against an arbitrary classical communication distance. However, the price for tolerating a larger communication distance was an increase in the lattice size. This naturally raises the question of whether the lattice sizes required by our construction are optimal.

Although we already established that the construction in Result~\ref{res:not-com-res} has the optimal scaling, it does not imply that it achieves the smallest possible lattice size for each fixed value of $r$. The ILP formulation in Eq.~\eqref{eq:opt-problem} then gives us a tool to find a better upper bound on $L_{\min}$ for these cases, by searching for solutions on progressively smaller lattices.

To fully solve the problem in Eq.~\eqref{eq:opt-problem}, one would need to solve for all the anticommutation and stabilizer certificates simultaneously. To make the computation manageable, we reduce the feasible set by imposing that every party lying on a horizontal edge has the same canonical context, while every party lying on a vertical edge also shares the same canonical context, which may differ from the horizontal one. Therefore, one needs to search for only two anticommutation certificates, taking a horizontal and a vertical edge as the respective target parties. These two canonical contexts are also imposed in the constraints for the star and plaquette stabilizers. It is then sufficient to find a single certificate for an arbitrary vertex and an arbitrary face, respectively, since all remaining certificates can be obtained by translation symmetry of the lattice. Finally, we restrict our search to using only four submeasurements per certificate.

The results are shown in Fig.~\ref{fig:Lmin-plot-tc}, where we compare the resulting upper bounds on $L_{\min}(r)$ with those obtained from the analytical construction. Despite this restricted search space, the ILP is able to find smaller lattices for every value of $r$ tested, showing that the explicit analytical construction is not optimal. 

Perhaps unsurprisingly, for the non-communicating case we find a solution on the smallest possible lattice, namely $L=2$. This construction, however, crucially relies on the periodic boundary conditions of the lattice (see Fig.~\ref{fig:patterns_L2} in Appendix~\ref{app:minimal_lattices}). Consequently, extending this solution in the same way as in Section~\ref{sec:comm} would yield solutions only for lattices whose size is exactly of the form $L=4r+2$. In particular, unlike the constructions in Results~\ref{res:not-com-res} and~\ref{res:comm-res}, this solution could not be extended to arbitrary larger lattices.

Furthermore, all solutions found have a lattice size of the form $L=2r+2$, which we conjecture to be the minimum possible size for each communication distance $r$. All constructions are available in \cite{froes2026code}, and those with $L\leq12$ are shown in Appendix~\ref{app:minimal_lattices}.

\subsection{Surface Codes}
\label{sec:surface-mip}

\begin{figure}[t]
    \centering
    \includegraphics[width=0.75\linewidth]{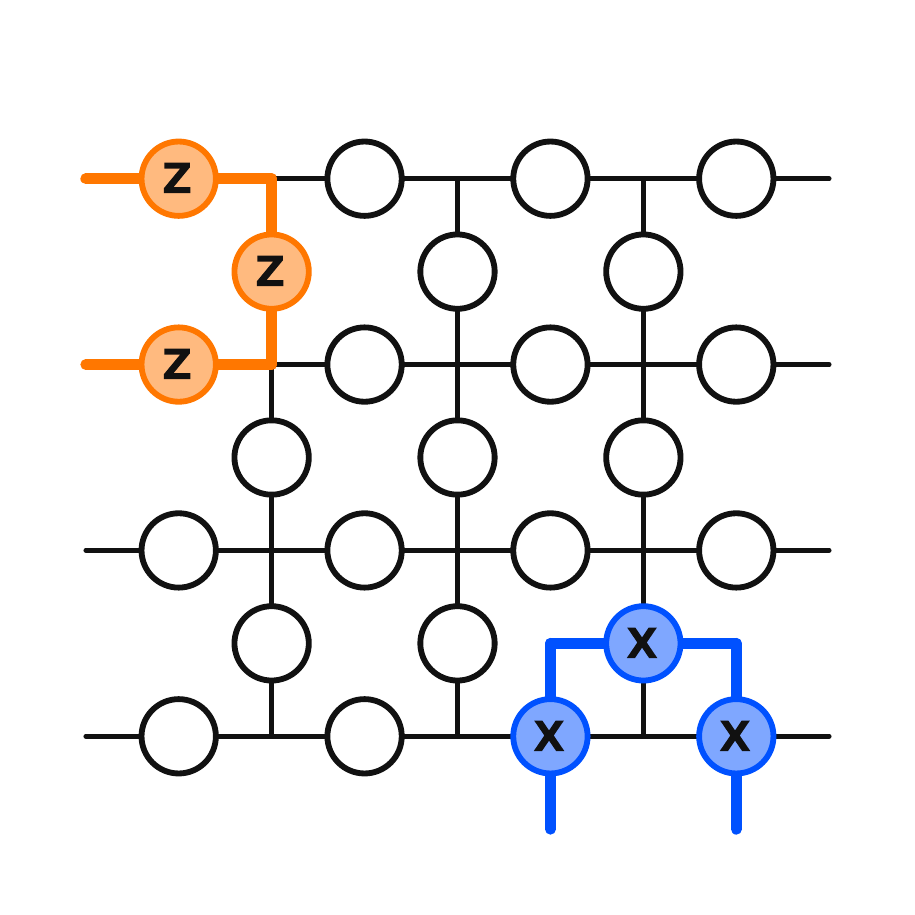}
    \caption{$4\times4$ surface code with rough boundaries on the left and right and smooth boundaries on the top and bottom. Stabilizers located at the boundaries correspond to a truncated plaquette at a rough boundary and a truncated star at a smooth boundary. As a consequence, primal (dual) loops can be open on rough (smooth) boundaries.}
    \label{fig:boundaries}
\end{figure}

So far, we only needed to solve the problem for the party on one edge, and the results for the other parties could be inferred from the translational symmetry of the lattice. However, we can also investigate some cases where this symmetry is lost, and a new set of submeasurements is needed for every party. 

To illustrate this, we consider another quantum error-correcting code closely related to the toric code, namely the surface code \cite{bravyi1998}. Its construction is very similar to that of the toric code, with the stabilizer again given by Eqs.~\eqref{eq:star-def} and \eqref{eq:plaquette-def}. The main difference is that, instead of imposing periodic boundary conditions, we introduce two distinct types of open boundaries, shown in Fig.~\ref{fig:boundaries}. The so-called \textit{smooth boundary} truncates the star stabilizers associated with vertices lying on that boundary. In particular, a vertex $v$ on a smooth boundary has only three edges connected to it, and the corresponding star stabilizer is (see the blue path in Fig.~\ref{fig:boundaries})
\begin{equation}
    A_v = X_{e_1} X_{e_2} X_{e_3}, \quad \partial {e_1},\partial {e_2}, \partial {e_3} \ni v.
\end{equation}
Conversely, a \textit{rough boundary} truncates the plaquette stabilizers adjacent to it, so that a plaquette $f$ touching a rough boundary is supported only on the three edges of its boundary that remain within the lattice (see the orange path in Fig.~\ref{fig:boundaries})
\begin{equation}
    B_f = Z_{e_4} Z_{e_5} Z_{e_6}, \quad {e_4},{e_5},{e_6} \in \partial f.
\end{equation}
In this way, the contractible loops defined for the toric code do not need to be closed anymore. A primal loop going over a rough boundary can be represented as an open path in the primal lattice. Similarly, a dual loop going over a smooth boundary is represented by an open path in the dual lattice (see Fig.~\ref{fig:boundaries}). The surface code subspace is then defined as the simultaneous +1 eigenspace of all stabilizers associated with this modified lattice. 

\begin{figure*}[t]
\centering
\includegraphics[width=\textwidth]
{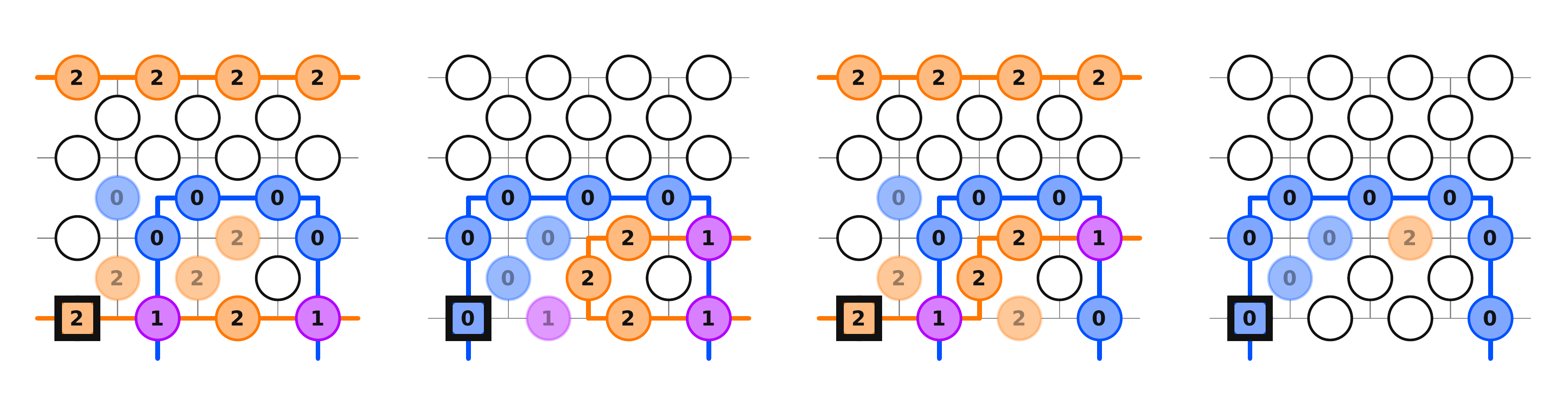}
\caption{ A $4\times4$ surface code with the inputs given by the referee to each party in measurements 1 through 4, from left to right, which are used for certifying the anticommutation relation for the party $H(0,0)$, highlighted with a square. The boundaries correspond to rough boundary conditions on the left and right, and smooth boundaries on the top and bottom.}
\label{fig:surfacecode}
\end{figure*}

Whereas the toric code encodes two logical qubits, the surface code on a square patch with alternating rough and smooth boundary conditions encodes a single logical qubit. Nonetheless, the lack of periodic boundary conditions allows the surface code to be implemented experimentally in a planar topology, requiring only 2D local interactions \cite{Horsman2012}. This has led to its use in several recent experimental demonstrations on superconducting qubit processors \cite{Krinner2022,Google2023,Google2024}.

For self-testing the surface code, we can once again use Proposition~\ref{prop: self-testing-criteria}, since its proof does not rely on any of the differences between the surface code and toric code subspaces. Moreover, the ILP formulation presented in Eq.~\eqref{eq:opt-problem} can also be applied directly, provided that the definitions of smooth and rough boundaries are taken into account. However, these boundaries break translational symmetry, so the argument used in the previous section no longer applies. In particular, it is no longer sufficient to search for only two certificates for the anticommutation relations and one for each type of stabilizer. Instead, parties at different positions relative to the boundaries may require different certificates, and each case must therefore be considered separately. Finding the corresponding submeasurements would thus require solving Eq.~\eqref{eq:opt-problem} in its full generality. Again, to make the search more tractable, we restrict the feasible set. This time, we fix the same canonical context presented in Section~\ref{sec:comm} for every party, independently of its orientation. In this way, we only need to search for the set of submeasurements certifying the anticommutation relation for each target party, while the plaquette and star stabilizers can be certified using the analogues of Eqs.~\eqref{eq:star-submeasurement} and~\eqref{eq:plaquette-submeasurement}. Moreover, the certificate associated with each target party is independent of the others and can therefore be searched for separately. Finally, we again restrict our search to only use $K_t=4$.

As a result, for each party we identified four submeasurements capable of certifying the anticommutation relation in a $4\times4$ surface code with smooth boundaries on the top and bottom, and rough on the left and right borders. More importantly, these submeasurements are also robust to 1 round of classical communication. Figure~\ref{fig:surfacecode} shows one such certificate for a party $H(0,0)$ located on an edge of a rough boundary. The submeasurements for the other parties can be found in \cite{froes2026code}. Therefore, this set of submeasurements, along with submeasurements analogous to the ones in Eqs.~\eqref{eq:star-submeasurement} and \eqref{eq:plaquette-submeasurement}, can be used to self-test this surface code, even when allowing a single round of classical communication to the untrusted parties.

\subsection{Graph states}
\label{sec:graph-states-mip}

Although the ILP framework we introduced was motivated by the toric and surface codes, the construction itself relies on features that are not specific to these subspaces. This suggests that similar optimization techniques may be useful in other self-testing problems involving stabilizers composed of Pauli $X$ and $Z$ observables. 

To exemplify, we demonstrate its application to graph states. Formally, a graph state $\ket{G}$ 
is associated with an undirected graph $G=(V,E)$, where each vertex $v\in V$ corresponds to a qubit. For every vertex $v$, we also define the stabilizer
\begin{equation}
G_v = X_v \prod_{u \in N(v)} Z_u,
\end{equation}
where $N(v)$ denotes the neighborhood of the vertex $v$. The state $\ket{G}$ is then the unique state stabilized by all of these operators
\vspace{1em}
\begin{equation}
    G_v \ket{G} = + \ket{G} \quad \forall v \in V. \\[1em]
\end{equation}

In \cite{Meyer2026}, the authors investigated the self-testing of graph states in the same communication scenario considered here, where the communication graph is given by the underlying graph of the state. For some families of symmetric graph states, the authors were able to construct protocols that self-test the states directly. In particular, they provide analytical constructions for cycle graph states and graph states defined on the honeycomb lattice, which relied on the fact that every party was equivalent under some symmetry. For arbitrary graph states, however, their general construction is indirect and requires additional auxiliary qubits that are not part of the target state. More precisely, the target graph is first embedded into a larger \textit{inflated graph}, in which each edge is replaced by a chain of auxiliary vertices whose length depends on the allowed communication range. The corresponding inflated graph state can then be self-tested in the presence of classical communication, after which the auxiliary qubits are projectively measured out through a procedure referred to as \textit{deflation}, leaving the desired graph state on the original vertices. Thus, although this construction applies to arbitrary graph states, it requires additional parties and physical subsystems beyond those belonging to the state that one ultimately wants to certify.

\begin{figure*}
    \centering
    \includegraphics[width=\textwidth]{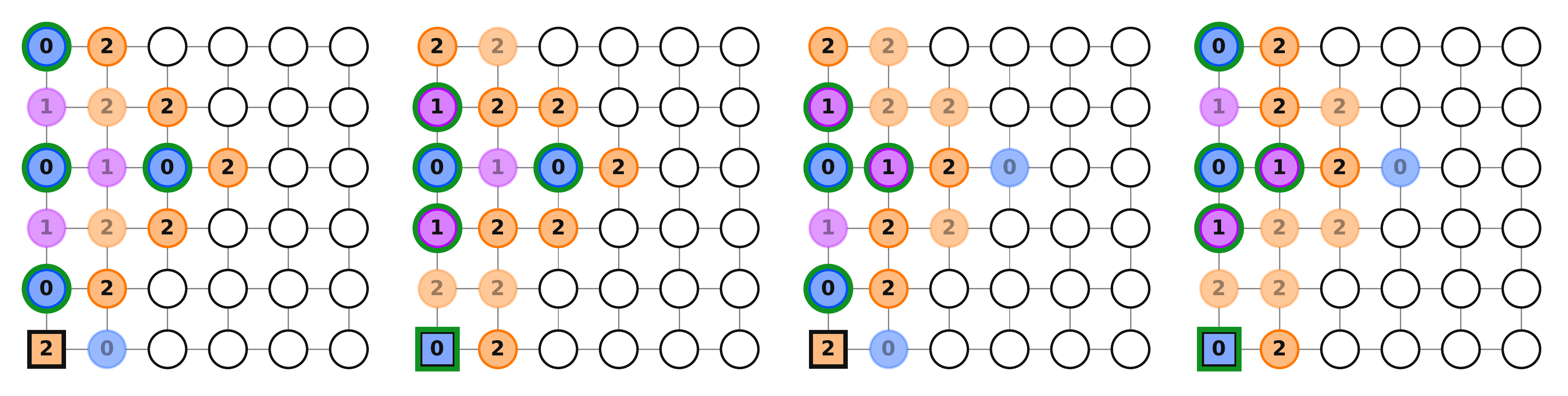}
    \caption{Submeasurements used to certify the anticommutation relation for the party associated with vertex $(0,0)$, highlighted by a square, in a $6\times 6$ cluster state. Vertices marked by thicker green circles, together with the highlighted square when applicable, indicate the stabilizer $G_v$ selected in the corresponding submeasurement. The parties participating in the submeasurement correspond to those being acted on by the product of these stabilizers.}
    \label{fig:cluster-states}
\end{figure*}

Here, we extend the results of \cite{Meyer2026} by showing that certain graph states can be self-tested directly, without requiring either auxiliary vertices or the symmetries exploited in their direct constructions.
The analogue of Proposition~\ref{prop: self-testing-criteria} for graph states was already established in \cite{Meyer2026}.
We can therefore formulate an ILP analogous to Eq.~\eqref{eq:opt-problem} for them by replacing the (sub)measurements used to certify the star and plaquette stabilizers, together with their associated constraints, with (sub)measurements that certify the corresponding graph state stabilizers. Meanwhile, the constraints associated with the anticommutation certificate remain conceptually the same; one just needs to adapt them to the geometrical constraints of the graph state stabilizers (see Appendix \ref{app:graph-state} for a discussion).  More precisely, the new constraints are given by
\begin{align}
    \left\{
    \begin{aligned}
        &  S^{(v)}_j = \left.\left(M^{(v)}_j \right)\right\rvert_{I^{(v)}_j} \quad \forall j \in [J_v] \\
        & \left\langle S^{(v)}_j \right\rangle = \mu_j^{(v)}
        \quad \mu_j^{(v)} \in \{\pm 1\},
        \\
        & \prod_{j=1}^{J_v} \left(S_j^{(v)}\right)_v = X_v^{\Gamma_X}, 
        \\
        & \prod_{j=1}^{J_v} \left(S_j^{(v)}\right)_u = Z_u^{\Gamma_Z} \quad \forall u \in N(v)\\
    & 
     \prod_{j=1}^{J_v} \left(S_j^{(v)}\right)_u = \Id
    \quad \forall u \not \in N(v)\cup\{v\}
    \\
    &
    \prod_{j=1}^{J_v} \mu_j^{(v)} = +1,
        \\
    \end{aligned}
    \right. \forall v.
\label{eq:opt-problem-graph}
\end{align}

To avoid solving the full ILP for all parties simultaneously, we can use the same strategy as for the surface code and fix a common canonical context for every party. Here, however, it is no longer convenient to assign the same input as the target to every party in its context. For example, for bipartite graphs, the stabilizers can instead be certified using a single submeasurement if we define the canonical context associated with the observable $X^{\Gamma_X}_t$ as the one where each party $u$, located at distance at most $r$ from the target $t$, receives an input according to
\begin{equation}
    x^{(t)}_k(u) = \left \{
    \begin{aligned}
        &0, \quad d(u,t) \equiv 0 \mod2, \\
        &2, \quad d(u,t) \equiv 1 \mod2 .
    \end{aligned}
    \right.
    \label{eq:canonical-context-X-graph}
\end{equation}
Similarly, for the $Z^{\Gamma_Z}_v$ observable, we assign the inputs according to 
\begin{equation}
    x^{(t)}_k(u) = \left \{
    \begin{aligned}
        &2, \quad d(u,t) \equiv 0 \mod2, \\
        &0, \quad d(u,t) \equiv 1 \mod2 .
    \end{aligned}
    \right.
    \label{eq:canonical-context-Z-graph}
\end{equation}
In this way, we only need one submeasurement $S^{(v)}$ per vertex $v$ to certify the stabilizers
\begin{equation}
    S^{(v)} =  X_v^{\Gamma_X} \prod_{u \in N(v)} Z_u^{\Gamma_Z},
\end{equation}
and the rest of the inputs composing the respective measurements $M^{(v)}$ are either obtained by Eqs.~\eqref{eq:canonical-context-X-graph}--\eqref{eq:canonical-context-Z-graph}, or are non important and set to 1.

A particular class of interest among the graph states is the so-called \textit{cluster states}. These are universal resources in measurement-based quantum computing \cite{Raussendorf2001, Raussendorf2003}, and are also intimately related to quantum error correction, in particular with surface codes \cite{Raussendorf2006}. Recently, this connection to error correction has been explored to demonstrate noise robustness in unconditional quantum advantage in shallow circuits \cite{Bravyi2020-fo, Caha2026-js}.

We therefore illustrate the strength of our method by constructing a self-test for a two-dimensional cluster state that remains robust in the presence of one round of communication. In Fig.~\ref{fig:cluster-states}, we display the submeasurements used to certify the required anticommutation relation for a $6\times 6$ cluster state, with the target party located at vertex $(0,0)$. The complete set of submeasurements for the other parties is provided in \cite{froes2026code}.

\section{Discussion \label{sec:disc}}

We have shown that quantum error-correcting code subspaces can be self-tested even when the untrusted devices secretly exchange classical messages over bounded distances. For the toric code, we obtain a noise-robust analytical construction that tolerates the parties communicating classical messages over distances proportional to the lattice size, achieving the optimal asymptotic scaling. We also introduce a systematic search based on Integer Linear Programming, which yields protocols on smaller toric-code lattices and extends our approach to the surface code and a two-dimensional cluster state, both with robustness to nearest-neighbor classical communication.

A potential application of our work is to construct self-testing protocols tailored to the specifications of an experimental setup, using the ILP framework presented here. The search can incorporate the geometry and connectivity of the devices, together with bounds on the communication distance imposed by the timing of the experiment. This flexibility could be particularly useful in multipartite settings, where enforcing spacelike separation between every pair of measurement events is demanding. By explicitly accounting for the classical communication available to the devices, such protocols could help address the locality loophole in multipartite certification, replacing the assumption of complete isolation with experimentally justified restrictions on information propagation.

An important motivation for extending the present framework beyond classical communication comes from proposed modular architectures for fault-tolerant quantum computation, in which a surface code is distributed across separate parties connected to their neighbors~\cite{Nickerson2013,Buonacorsi2019}. 
Operating the code in such architectures requires neighboring parties to exchange both quantum and classical information, for instance when performing distributed stabilizer measurements. 
It is therefore natural to ask for certification protocols robust not only against classical communication, but also when the untrusted devices can secretly exchange \emph{quantum} messages, as the experimental architecture will include quantum communication channels between neighbouring parties.
Hence, a natural direction for future work is thus to extend our bounded-distance communication model to also include the possibility of secrete quantum communication.

Even within the classical communication setting, an important distinction is whether the target resource can be certified without introducing auxiliary parties. The authors in \cite{Meyer2026} provide a construction for arbitrary graph states using auxiliary parties whose qubits are measured out to prepare the final certified state. Its direct constructions, which certify a qubit at every party, were tailored to exploit the symmetries of cycle and honeycomb graph states. Our direct self-test of a $6\times6$ cluster state, robust to nearest-neighbor classical communication, shows that certification without auxiliary parties is possible beyond these examples. This suggests that direct self-testing may extend to broader classes of graph states, motivating a general characterization of the resources that admit such protocols and the communication distances they can tolerate.

\begin{acknowledgments}

The authors are grateful to Patrick Emonts, Owidiusz Makuta, Ivan Šupić, Eloïc Vallée and Isadora Veeren for helpful discussions.

This project was partially supported by the ANR for the JCJC grants LINKS (No. ANR-23-CE47-0003), the T-ERC QNET (No. ANR24-ERCS-0008), the project QUANTINT, as well as the European Union's Horizon 2020 Research and Innovation Programme under QuantERA Grant Agreements No. 731473 and No. 101017733. 
The Institute Quantum-Saclay and the project QuanTEdu-France are acknowledged for the funding of this project.
\end{acknowledgments}

\section*{Statement on AI usage}
The authors used GPT-5.6 Sol and Claude Opus 5.5 to assist with implementing and debugging the integer linear programs formulated in this work. All AI-assisted code was reviewed and verified by the authors, who take full responsibility for the content and results of this work.

\section*{Data availability}
All the data used in Section \ref{sec:MILP} is available in \cite{froes2026code}. 

\vspace{2em}
\bibliographystyle{quantum}
\bibliography{ref}

\onecolumn\newpage
\appendix

\section{Proof of Proposition \ref{prop: self-testing-criteria}}
\label{app:proof-prop1}
\begin{proof}
Let $\widetilde{X}_e$ and $\widetilde{Z}_e$ denote the two binary operators acting on $\mathcal{H}_e$, which we implicitly extend by acting as the identity on all other systems, including the purifying system $\mathcal{H}_P$. The observed statistics from the self-testing protocol certify
\begin{equation}
    \left\{\widetilde{X}_e,\widetilde{Z}_e\right\}\ket{\Psi} = 0,
    \label{eq:anticomm-eq-app}
\end{equation}
and 
\begin{align}
    &\prod_{\partial e \ni v} \widetilde{X}_e \ket{\Psi} = +1\ket{\Psi} \quad \forall v,\\
    &\prod_{ e \in \partial f} \widetilde{Z}_e \ket{\Psi} = +1\ket{\Psi} \quad \forall f,
\end{align}
where $\ket{\Psi} \in \bigotimes_e\mathcal{H}_e \otimes \mathcal{H}_P$ is the purification of the shared state of the untrusted parties. 
We can then use the so-called SWAP isometry, introduced in \cite{Mayers2004} and shown in Fig.~\ref{fig:swap-isometry}. For a party $e$, the isometry is $\Phi_e: \mathcal{H}_e \rightarrow \mathcal{H}'_e \otimes \mathcal{H}_e''$, and  is given by
\begin{equation}
    \Phi_e  =  \ket{0} \otimes \left( \frac{\Id+\widetilde{Z}_e}{2} \right)  +  \ket{1} \otimes \left(\widetilde{X}_e\frac{\Id-\widetilde{Z}_e}{2} \right),
\end{equation}
and $\dim(\mathcal{H}'_e) = 2$.

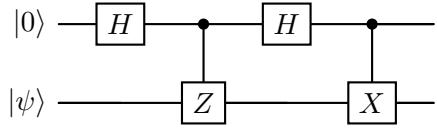
\begin{figure}[t]
    \centering
    \begin{quantikz}
        \lstick{$\ket{0}$} 
            & \gate{H} 
            & \ctrl{1} 
            & \gate{H} 
            & \ctrl{1}
            & \qw
        \\
        \lstick{$\ket{\psi}$}
            & \qw
            & \gate{Z}
            & \qw
            & \gate{X}
            & \qw
    \end{quantikz}
    \caption{Swap isometry used to extract the ideal qubit encoded in the physical system into an auxiliary register.}
    \label{fig:swap-isometry}
\end{figure}

Now, using the anticommutation relation from Eq.~\eqref{eq:anticomm-eq-app}, the full isometry $\Phi = \bigotimes_e \Phi_e$ satisfies 
\begin{align}
    & \prod_{\partial e \ni v} ({X}_e \otimes \Id_{\mathcal{H}''\otimes \mathcal{H}_P}) (\Phi\otimes \Id_P)\ket{\Psi} = (\Phi\otimes \Id_P) \prod_{\partial e \ni v} \widetilde{X}_e \ket{\Psi} = (\Phi\otimes \Id_P)\ket{\Psi} \quad \forall v,
    \label{eq:star-stab-app}
\end{align} 
\begin{align}
    &   \prod_{ e \in \partial f} ({Z}_e \otimes \Id_{\mathcal{H}''\otimes \mathcal{H}_P}) (\Phi\otimes \Id_P)\ket{\Psi} = (\Phi\otimes \Id_P) \prod_{e \in \partial f} \widetilde{Z}_e \ket{\Psi} = (\Phi\otimes \Id_P)\ket{\Psi} \quad \forall f,
    \label{eq:plaquette-stab-app}
\end{align}
where ${X}_e$ and ${Z}_e$ denote the Pauli operators acting on the auxiliary qubit in $\mathcal{H}'_e$. 

Now, following the proof in \cite{Baccari2020}, we derive the most general form of the state $(\Phi\otimes \Id_P)\ket{\Psi}$. To this end, we consider the Schmidt decomposition of $(\Phi\otimes \Id_P)\ket{\Psi}$ with respect to the bipartition between $\mathcal{H}'$ and $\mathcal{H}''\otimes\mathcal{H}_P$
\begin{equation}
    (\Phi\otimes \Id_P) \ket{\Psi} = \sum_i \omega _i \ket{\eta_i} \ket{\zeta_i}.
\end{equation}

From Eqs.~\eqref{eq:star-stab-app} and \eqref{eq:plaquette-stab-app}, we can conclude that
\begin{align}
    \prod_{\partial e \ni v} {X}_e \ket{\eta_i}     = \ket{\eta_i} \quad \forall v, \\
    \prod_{ e \in \partial f} {Z}_e \ket{\eta_i}     = \ket{\eta_i} \quad \forall f, 
\end{align}
for all $i$. Thus, each $\ket{\eta_i}$ is in the toric code subspace, and we can write
\begin{equation}
    \ket{\eta_i} = \sum_{\kappa=1}^4 \widetilde{c}_{i,\kappa}\ket{\psi_\kappa},
\end{equation}
where  $\{\ket{\psi_\kappa}\}_\kappa$ forms an orthonormal basis of the toric code subspace, and $\sum_\kappa \lvert \widetilde{c}_{i,\kappa}\rvert ^2 = 1$. Finally, this implies that there exist unnormalized states ${\big| \widetilde\xi_\kappa} \big\rangle$ such that 
\begin{equation}
    (\Phi\otimes \Id_P) \ket{\Psi} = \sum_{\kappa=1}^4 \ket{\psi_\kappa}\big| \widetilde{\xi}_\kappa \big\rangle = \sum_{\kappa=1}^4 c_\kappa\ket{\psi_\kappa}\ket{\xi_\kappa} ,
\end{equation}
where the $\ket{\xi_\kappa}$ are normalized, $c_\kappa \geq 0$ and $\sum_\kappa c_\kappa^2 = 1$.
\end{proof}

\section{Generalized construction for arbitrary communication distance}

\label{app:generalization}

Here, we show how to generalize the construction from the main text to show that a toric code of lattice size $L \geq 6r+3$ can be self-tested even if the untrusted parties have access to $r$ rounds of classical communication. Before we begin, the following lemma will be useful, and we prove it in Appendix \ref{app:proof_lemma}.
\begin{lemma}[Distance in the communication graph]
\label{lemma: distance-lemma}
For $a\in\mathbb{R}$, we define
\begin{equation}
    \lvert a\rvert_L
    :=
    \min_{n\in\mathbb{Z}}\lvert a+nL\rvert.
\end{equation}
Let $e$ and $e'$ be two parties in the communication graph of the toric code. If they have the same orientation, namely,
\begin{equation}
    e=H(i,j),\quad e'=H(p,q),
\end{equation}
or
\begin{equation}
    e=V(i,j),\quad e'=V(p,q),
\end{equation}
then their graph distance is
\begin{equation}
    d(e,e')
    =
    2\max\left\{
        \lvert i-p\rvert_L,
        \lvert j-q\rvert_L
    \right\}.
\end{equation}
If instead $e$ and $e'$ have different orientations $e=H(i,j)$ and $e' = V(p,q)$, then
\begin{equation}
    d(e,e')
    =
    2\max\left\{
        \left\lvert i-p+\frac{1}{2}\right\rvert_L,
        \left\lvert j-q-\frac{1}{2}\right\rvert_L
    \right\}.
\end{equation}
\end{lemma}

\noindent We can now begin by considering a lattice where the parties are located, and of size $L=3s$, where $s=2r+1$, and a coarse-grained lattice of size $L'=3$ embedded on it. In the notation introduced in the main text, this coarse-grained lattice has its horizontal edges described as 
\begin{equation}
    H'(i,j) = \left\{(si,sj),(si+s,sj)\right\},
\end{equation}
and vertical edges
\begin{equation}
    V'(i,j) = \left\{(si,sj),(si,sj+s)\right\},
\end{equation}
where the vertices on the right-hand side are with respect to the original lattice. We can implement the same loop operators defined in Fig.~\ref{fig:measurements-vertical} on the coarse-grained lattice, associated with the four submeasurements $S^{(t)}_k$ for a given target party $t$. These are shown diagrammatically in Fig.~\ref{fig:generalized-loops}. Namely, for the target party $t = V(2s,s+r)$, the five loops composing these submeasurements consist of the three primal loops 
\begin{align}
    \gamma_1 = \{H'(1,0),V'(2,0),V'(2,1), H'(1,2),V'(1,1),V'(1,0)\}, 
    \label{eq:first-loop-app}
\end{align}
\begin{equation}
    \gamma_2 = \{H'(1,1),V'(2,1),H'(1,2),V'(1,1)\}, 
\end{equation}
\begin{equation}
    \gamma_3 = \{H'(1,0),V'(2,0),H'(1,1),V'(1,0)\}, 
\end{equation}
and the two dual loops 
\begin{align}
    \gamma^\star_1 = \{V'(2,0), H'(2,1), H'(2,2),V'(2,2), H'(1,2), H'(1,1)\},    
\end{align}
\begin{equation}
    \gamma^\star_2 = \{V'(2,0), H'(2,1), V'(2,1), H'(1,1)\}.
    \label{eq:last-loop-app}
\end{equation} 
We use the following convention to associate edges of the coarse-grained lattice with segments of the original lattice. If $H'(i,j)\in\gamma$, then the same loop, expressed on the original lattice, contains the following $s$ consecutive edges
\begin{equation}
H(si,sj),H(si+1,sj),\dots,H(si+s-1,sj) \in \gamma.
\end{equation}
Similarly, if $V'(i,j) \in \gamma$, then
\begin{equation}
V(si,sj),V(si,sj+1),\dots,V(si,sj+s-1) \in \gamma.
\end{equation}
For the dual loop, we say that if $H'(i,j) \in \gamma^\star$, then
\begin{equation}
H(si+r,sj-r),H(si+r,sj-r+1),\dots,H(si+r,sj+r) \in \gamma^\star,
\end{equation}
whereas $V'(i,j) \in \gamma^\star$ implies that
\begin{equation}
V(si-r,sj+r),V(si-r+1,sj+r),\dots,V(si+r,sj+r) \in \gamma^\star.
\end{equation}

Thus, the submeasurements used in the self-testing of the toric code, and which are robust to $r$ rounds of communication, can be described as follows. The parties participating in the submeasurements $S^{(t)}_k$  are those in $I^{(t)}_k = \mathcal{Z}_k \cup \mathcal{X}_k$, where 
\begin{equation} \left\{
    \begin{aligned}
        &\mathcal{Z}_1 = \gamma_1,\\
        &\mathcal{Z}_2 = \emptyset,\\
        &\mathcal{Z}_3 = \gamma_2,\\
        &\mathcal{Z}_4 = \gamma_3,\\
    \end{aligned}
    \right. \quad \text{ and } \quad \left\{
    \begin{aligned}
        &\mathcal{X}_1 = \gamma_1^\star,\\
        &\mathcal{X}_2 = \gamma_2^\star,\\
        &\mathcal{X}_3 = \gamma_1^\star,\\
        &\mathcal{X}_4 = \gamma_2^\star,\\
    \end{aligned}
    \right.,
\end{equation}
and the corresponding inputs for these parties are assigned according to
\begin{equation}
    x_k (e) = \left \{ 
    \begin{aligned}
    &0, \quad  e \in \mathcal{X}_k \setminus \mathcal{Z}_k\\ 
    &1, \quad e \in \mathcal{X}_k \cap \mathcal{Z}_k\\ 
    &2, \quad e \in \mathcal{Z}_k \setminus \mathcal{X}_k.\\ 
    \end{aligned}
    \right.
    \label{eq:assig-loop}
\end{equation}

\noindent As was the case in the main text, these submeasurements satisfy 
\begin{equation}
\begin{aligned}
    S^{(t)}_1 \ket{\psi} = - \ket{\psi}, \\
    S^{(t)}_2 \ket{\psi} = + \ket{\psi}, \\
    S^{(t)}_3 \ket{\psi} = - \ket{\psi}, \\
    S^{(t)}_4 \ket{\psi} = - \ket{\psi}, \\
\end{aligned}
\end{equation}
and 
\begin{equation}
    S_1^{(t)}S_2^{(t)}S_3^{(t)}S_4^{(t)}   = Z^{\Gamma}_t X^{\Gamma}_t Z^{\Gamma}_t X^{\Gamma}_t,
    \label{eq:ZXZX-app-RE}
\end{equation}
for the reference experiment. Crucially, the last relation relies on the fact that, in the RE, the observable performed by a given party is independent of the inputs assigned to its neighbors, as they do not communicate their input. For the PE, however, communication allows the observable of each party to depend on the entire context up to a distance $r$ available to it. Therefore, to recover the analogue of Eq.~\eqref{eq:ZXZX-app-RE} in the PE, the inputs within a distance $r$ of every party whose observables are paired across different submeasurements must be modified so that they match, and not simply set to $1$. Moreover, in order to certify the star and plaquette stabilizers using the same operators appearing in Eq.~\eqref{eq:ZXZX-app-RE}, we must fix the relevant observables to the canonical contexts introduced in the main text, in which every party within distance at most $r$ of the target party $t$ receives the same input from $t$.

\begin{figure*}[t]
    \centering
    \includegraphics[width=0.94\textwidth]{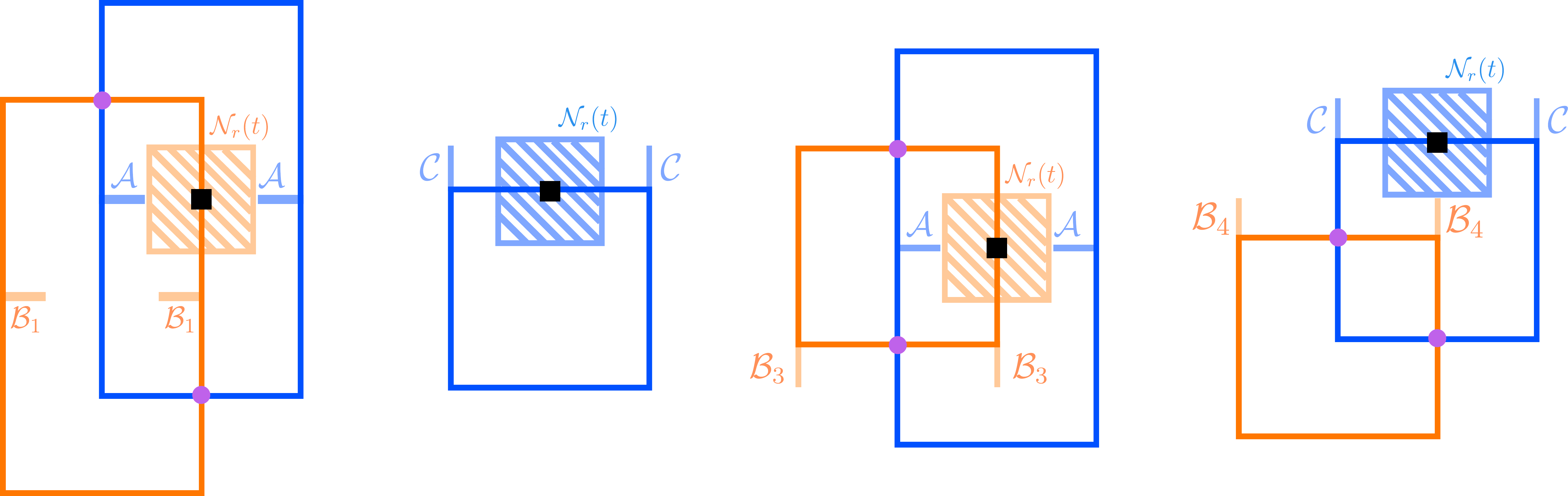}
    \caption{
    Diagrammatic representation of the input assignments for the four measurements $M_k^{(t)}$, $k=1,2,3,4$. The submeasurements involve the parties on the primal loops $\gamma_1,\gamma_2,\gamma_3$ (dark orange) and those on edges crossed by the dual loops $\gamma_1^\star,\gamma_2^\star$ (dark blue). The additional sets of parties defined in Eq.~\eqref{eq:extra-sets} are indicated in light blue and light orange, receiving inputs $0$ and $2$, respectively. The black square marks the target party $t$, while the hatched square indicates the neighborhood in which the canonical context is imposed. Finally, the coorindates for parties $\mathcal{A}, \mathcal{B}_{k}, \text{and } \mathcal{C}$ are explicitly provided in Eq.~\eqref{eq:extra-sets}.
    }
    \label{fig:generalized-loops}
\end{figure*}

To make the simulation of these submeasurements by the PE imply the conditions of Proposition~\ref{prop: self-testing-criteria}, we then need to change the inputs of the following parties. For the parties located within distance $r$ from the target $t$, the canonical context enforces that
\begin{equation}
    x^{(t)}_k(e)= \left\{
    \begin{aligned}
        &0, \quad  k\in\{2,4\},\\
        &2, \quad  k\in\{1,3\}.\\
    \end{aligned}
    \right.
    \label{eq:canonical-context-app}
\end{equation}
For the inputs regarding the rest of the parties that are neither in $I^{(t)}_k$ nor in the canonical context, we define the following set of parties (represented in Fig.~\ref{fig:generalized-loops})
\begin{align}
\label{eq:extra-sets}
\mathcal{A} &= \{V(2s-r+a,\,s+r),\; V(2s+r-a,\,s+r)\}_{a=0}^{q_\perp-1}, \nonumber\\
\mathcal{B}_1 &= \{H(s+a,\,s),\; H(2s-1-a,\,s)\}_{a=0}^{q_\perp-1},\nonumber\\
\mathcal{B}_3 &= \{V(s,\,s-1-a),\; V(2s,\,s-1-a)\}_{a=0}^{q_\perp-1},\\
\mathcal{B}_4 &= \{V(s,\,s+a),\; V(2s,\,s+a)\}_{a=0}^{q_\perp-1},\nonumber
\\
\mathcal{C} &= \{H(s+r,\,2s-r+a),\; H(2s+r,\,2s-r+a)\}_{a=0}^{q_\parallel-1}, \nonumber
\end{align}
where
\begin{equation}
      q_{\perp}
    :=
    \left\lfloor\frac{r+1}{2}\right\rfloor \quad \text{ and } \quad  q_{\parallel}
    :=
    \left\lfloor\frac{r}{2}\right\rfloor,
\end{equation}
such that their input obeys
\begin{equation}
    x_k(e) = \left\{ 
    \begin{aligned}
        &0, \quad e\in \mathcal{A}, \text{ } k\in \{1,3\} \text{ or } e \in \mathcal{C},\text{ } k \in \{2,4\}, \\
        &2, \quad e \in B_k, \\
        &1, \quad \text{otherwise},
    \end{aligned}
    \right.
    \label{eq:change-input-ass}
\end{equation}
where $\mathcal{B}_2=\emptyset$. Also note that the sets $\mathcal{A},\mathcal{B}_k,\mathcal{C}$ are disjoint from $I^{(t)}_k$ and from $\mathcal{N}_r(t) = \{\ell: d(t,\ell) \leq r\}$, and on $\mathcal{N}_r(t)\cap I^{(t)}_k$ the inputs assigned according to Eq.~\eqref{eq:canonical-context-app} and Eq.~\eqref{eq:assig-loop} agree, so the assignment is well defined for every party. Moreover, the measurements for all other vertical targets follow by translation of the coordinate system, and for horizontal targets by a rotation by $90^\circ$ of the loops defined here. 

The remainder of this section is then devoted to showing that these are indeed the only parties that need their input assignment modified in order to ensure that the context of each relevant party is identical whenever its observables are paired across different submeasurements.

We start by reducing the set of parties that need to be considered and can possibly affect the cancellations needed to conclude Eq.~\eqref{eq:ZXZX-app-RE}. Notably, many parties neither participate in any of the primal or dual loops defining the submeasurements nor belong to the canonical context of the target party. Their inputs can therefore affect neither the expectation values of the submeasurements nor the contexts that must be matched. Hence, the inputs of all such parties may be set to $1$ in every measurement. Once the canonical context around the target party is fixed, it therefore remains only to consider parties lying on at least one of the loops defined in Eqs.~\eqref{eq:first-loop-app}-\eqref{eq:last-loop-app}. For this purpose, Lemma~\ref{lemma: distance-lemma} can be used to determine which of these parties may lie within one another's communication contexts. 

The first case we examine concerns the relative orientations of the line segments composing these loops. Consider a party $e$ lying on a vertical segment of a primal loop $\gamma$ and a party $e'$ lying on a horizontal segment of a dual loop $\gamma^\star$. Their coordinates are of the form
\begin{equation}
e = V(si, sj+n) \quad n \in [0,1,\dots,s-1],
\end{equation}
\begin{equation}
e' = H(sp+r, sq-r+m) \quad m \in [0,1,\dots, s-1].
\end{equation}
Therefore
\begin{equation}
d(e,e') = 2\max \left\{\left\lvert sp+r-si+\frac{1}{2}\right\rvert_L, \left\lvert sq-r+m-sj-n-\frac{1}{2}\right\rvert_L\right\}.
\end{equation}
However, notice that
\begin{equation}
\left\lvert sp+r-si+\frac{1}{2}\right\rvert_L \geq \left\lvert r + \frac{1}{2} - s \lvert i-p\rvert_{3}\right\rvert_L = s\left\lvert  \frac{1}{2} -  \lvert i-p\rvert_{3}\right\rvert_{3} \geq r+\frac{1}{2},
\end{equation}
where in the last inequality we used the fact that $i-p$ is an integer. Thus, we can conclude that they cannot be in each other's context, as
\begin{equation}
d(e,e') \geq s > r.
\end{equation}
The converse case, in which $e$ lies on a horizontal edge of a primal loop and $e'$ lies on a vertical edge of a dual loop, follows analogously. 

A similar restriction applies to parties lying on loop segments with the same orientation. In particular, if two parties $e$ and $e'$ lie on distinct segments of the same orientation that are not adjacent to each other, then their distance is also greater than the communication distance. We show this explicitly for the case in which both parties lie on horizontal edges, $H(si+n,sj)$ and $H(sp+m,sq)$, and the case for two vertical edges follows in the same way. In that case, the distance of these two such edges is
\begin{equation}
    d(e,e') = 2\max \left\{\left\lvert si+n-sp-m\right\rvert_L, \left\lvert sj-sq \right\rvert_L  \right\}.
    \label{eq:dist-1}
\end{equation}
However, note that for them to lie on distinct and non-adjacent coarse-grained edges, we must have $j\neq q$ once the size of the coarse-grained lattice is only $L'=3$, and thus
\begin{equation}
    \left\lvert sj-sq \right\rvert_L = s\left\lvert j-q \right\rvert_3 \geq s.
\end{equation}
which implies that 
\begin{equation}
    d(e',e) \geq 2s > r.
    \label{eq:dist-2}
\end{equation}

Therefore, the only parties that can affect the distance $r$ context of another party fall into one of the following categories: they belong to the same loop segment, they lie on two intersecting loop segments, or they belong to two adjacent loop segments.

The first two cases do not require any additional modifications to the input assignments. Indeed, the coarse-grained loops have exactly the same pairing structure as the loops used in the $r=0$ construction shown in Eq.~\eqref{eq:cancelations}. Therefore, whenever one of these structures appears within the context of a given party, the same structure also appears within the corresponding context in the paired submeasurement. Consequently, these cases cannot give rise to any mismatch between the corresponding contexts.

It remains, therefore, to consider parties close to the endpoints of adjacent segments. This is precisely the situation responsible for the context mismatches explained in the main text for $r=1$. Depending on the submeasurement, a loop reaching a coarse-grained vertex may continue along one adjacent segment in the same direction or along a perpendicular one. A party sufficiently close to this vertex can then use these inputs to distinguish which continuation was chosen and prevent the self-testing.

To fix this, we first consider two adjacent segments of different orientations. Without loss of generality, we can translate their common endpoint to $(0,0)$ and consider the parties
\begin{equation}
    e_a = H(a-1,0),
    \qquad
    e'_b = V(0,b-1),
\end{equation}
where $a,b\in [1,2,\dots,s]$ encodes the position of the edges present in these segments, starting from the common endpoint. Lemma~\ref{lemma: distance-lemma} gives
\begin{align}
    d(e_a,e'_b)
    &=
    2\max\left\{
        \left|a-\frac{1}{2}\right|_L,
        \left|b-\frac{1}{2}\right|_L
    \right\}\\
    &=
    2\max\{a,b\}-1.
\end{align}
Hence, the two parties can belong to one another's context only if
\begin{equation}
    \max\{a,b\}
    \leq
    \left\lfloor\frac{r+1}{2}\right\rfloor.
\end{equation}
Equivalently, the party closest to the endpoint will have
\begin{equation}
    q_{\perp}
    =
    \left\lfloor\frac{r+1}{2}\right\rfloor
\end{equation}
parties from the adjacent segment in its context. 

This observation leads to the first required modification of the measurement assignments. Consider a coarse-grained edge that belongs to two different loops defining a pair of submeasurements, so that the observables associated with all parties along this edge cancel when the two submeasurements are multiplied. If the two loops differ on a coarse-grained edge adjacent to the one under consideration, then parties close to their common endpoint may distinguish the two submeasurements by detecting whether the parties along the perpendicular segment received the corresponding loop inputs. To prevent them from accessing this information, we also assign the appropriate loop input to the first $q_{\perp}$ parties along the segment that is absent from the paired submeasurement. Importantly, these parties are not included in the support of the submeasurement, and their outputs are therefore disregarded by the referee. In this way, all parties sufficiently close to the common endpoint receive the same context in both submeasurements, and thus cannot distinguish between the two loops.

This modification is required in two situations, corresponding to the two cases identified in the main text for $r=1$. First, consider the two dual loops $\gamma_1^\star$ and $\gamma_2^\star$. These loops share the segments associated with $H'(1,1)$ and $H'(2,1)$, but differ in the continuation through $V'(2,1)$, which belongs to $\gamma_2^\star$ but not to $\gamma_1^\star$. Parties sufficiently close to either endpoint of $V'(2,1)$ may therefore detect whether this segment is present and distinguish the two loop configurations. Thus, in the submeasurements containing $\gamma_1^\star$, we also assign input $0$ to the first $q_{\perp}$ parties associated with the segment defined by $V'(2,1)\in \gamma_1^\star$ from both endpoints, while disregarding their outputs. The relevant parties on the segments associated with $H'(1,1)$ and $H'(2,1)$ then see the same context, in the perpendicular direction, independently of whether $\gamma_1^\star$ or $\gamma_2^\star$ is being considered by the referee. This corresponds exactly to the change in the input assignment to the parties belonging to $\mathcal{A}$, shown in Eq.~\eqref{eq:change-input-ass} (see also Fig.~\ref{fig:generalized-loops}).

The same type of mismatch arises for the three primal loops $\gamma_1$, $\gamma_2$, and $\gamma_3$, and correspond to the changes in the inputs of the parties in the sets $\mathcal{B}_k$ in Eq.~\eqref{eq:change-input-ass}, and represented in  Fig.~\ref{fig:generalized-loops}. At the vertex $(s,s)$, the relevant adjacent segments are those associated with $H'(1,1)$, $V'(1,0)$, and $V'(1,1)$, while at the vertex $(2s,s)$, they are those associated with $H'(1,1)$, $V'(2,0)$, and $V'(2,1)$. The three loop configurations differ by which of these segments is absent. The loop $\gamma_1$ does not contain $H'(1,1)$, $\gamma_2$ does not contain $V'(1,0)$ or $V'(2,0)$, and $\gamma_3$ does not contain $V'(1,1)$ or $V'(2,1)$. We therefore assign input $2$ to the first $q_{\perp}$ parties along each corresponding missing continuation. More explicitly, this modification is made along $H'(1,1)$ for $\gamma_1$, along $V'(1,0)$ and $V'(2,0)$ for $\gamma_2$, and along $V'(1,1)$ and $V'(2,1)$ for $\gamma_3$. In this way, the parties lying sufficiently close to either vertex see the same relevant continuations in their contexts, regardless of which of the three primal loops is being probed by the referee.

There is one further possibility when $r\geq 2$, which accounts for the modifications related to the final set $\mathcal{C}$ in Eq.\eqref{eq:change-input-ass}. Two adjacent coarse-grained segments may have the same orientation and form a straight continuation of one another. Again translating their common endpoint to the origin, consider, for example,
\begin{equation}
    e_a = V(0,-a),
    \qquad
    e'_b = V(0,b-1).
\end{equation}
Since the two parties have the same orientation, Lemma~\ref{lemma: distance-lemma} gives
\begin{equation}
    d(e_a,e'_b)
    =
    2(a+b-1).
\end{equation}
Therefore, a party on one segment can detect inputs on the other segment only within the first
\begin{equation}
    q_{\parallel}
    =
    \left\lfloor\frac{r}{2}\right\rfloor
\end{equation}
edges beyond the common endpoint. Notice that this contribution vanishes for $r=1$, which explains why no corresponding modification is required in the explicit construction of Fig.~\ref{fig:measurements-1-round}. For $r\geq 2$, however, the distance-$r$ context of parties lying on the segments associated with the coarse-grained edges $H'(1,1)$ and $H'(2,1)$ can extend beyond their endpoints and reach the adjacent collinear segments in the loop $\gamma^\star_1$. These parties may therefore detect whether the dual loop continues straight beyond that endpoint, allowing them to distinguish the generalizations of $\gamma^\star_1$ and $\gamma^\star_2$, even after adding the continuation of these loops in the perpendicular direction. To prevent this, in the submeasurement where the straight continuation is absent, we assign the input $0$ to the first $q_{\parallel}$ parties along that missing continuation. As a result, the parties on $H'(1,1)$ and $H'(2,1)$ see the same straight continuation in their contexts in the relevant paired submeasurements, and can no longer distinguish between the two loops.

Finally, by inspection of the coarse-grained loops in Eqs.~\eqref{eq:first-loop-app}--\eqref{eq:last-loop-app} (see also Fig.~\ref{fig:generalized-loops}), one can see that these are the only adjacent segments whose presence differs between submeasurements and whose local observables are paired according to Eq.~\eqref{eq:cancelations}. Hence, the modifications above exhaust all possible context mismatches.

The only remaining observation is to show that the addition of the canonical context does not introduce any further context mismatches. This follows directly from the distance bounds derived above. In particular, parties lying on loop segments parallel or adjacent to the segment containing the target party are too far away to detect any of the additional inputs used to impose the canonical context. For parties lying on the same loop segment as the target, no additional mismatch arises, since whenever the canonical context falls within distance $r$ of a given party, the same context also appears in the corresponding paired submeasurement that cancels that party's observable.

After these modifications to the input assignments defining the measurements $M^{(t)}_k$ and the corresponding $\widetilde{M}^{(t)}_k$, the submeasurements in the PE satisfy
\begin{equation}
\widetilde{S}^{(t)}_1\widetilde{S}^{(t)}_2\widetilde{S}^{(t)}_3\widetilde{S}^{(t)}_4 = 
\widetilde{Z}^{\Gamma_Z}_t
\widetilde{X}^{\Gamma_X}_t
\widetilde{Z}^{\Gamma_Z}_t
\widetilde{X}^{\Gamma_X}_t.
\end{equation}

The remainder of the proof for self-testing then proceeds exactly as in the main text. In particular, the construction shown above does not depend on the boundary conditions of the lattice. Therefore, we can self-test against $r$ rounds of communication in any lattice in which we can embed the construction shown here.

\section{Proof of Lemma \ref{lemma: distance-lemma}}
\label{app:proof_lemma}
\begin{proof}
We can associate each edge of the toric code lattice with the midpoint of its two vertices. Accordingly, a horizontal edge is represented by the coordinates
\begin{equation}
    H(i,j) := \left(i+\frac{1}{2},j\right),
\end{equation}
while a vertical edge is represented by
\begin{equation}
    V(i,j) := \left(i,j+\frac{1}{2}\right).
\end{equation}

Now, following Eq.~\eqref{eq:neigh-hor}, the four neighbors of a vertical edge $V(i,j)$ are
\begin{equation}
    H(i-1,j)=\left(i-\frac{1}{2},j\right),
\end{equation}
\begin{equation}
    H(i,j)=\left(i+\frac{1}{2},j\right),
\end{equation}
\begin{equation}
    H(i,j+1)=\left(i+\frac{1}{2},j+1\right),
\end{equation}
and
\begin{equation}
    H(i-1,j+1)=\left(i-\frac{1}{2},j+1\right).
\end{equation}

Hence, a single round of communication corresponds to a displacement of the form
\begin{equation}
    D=\left(\frac{a}{2},\frac{b}{2}\right),
    \qquad
    a,b\in\{\pm 1\}.
    \label{eq:displacement-cond}
\end{equation}
The same set of possible displacements is obtained when moving from a horizontal edge to any of its neighboring vertical edges. Thus, the distance between any two edges is given by the minimum number of such displacements necessary for going from one edge to the other. 

Now, let $e$ and $e'$ be two edges of the lattice, and denote their relative displacement by $e-e'=(\Delta x,\Delta y)$.
From Eq.~\eqref{eq:displacement-cond}, each round of communication can change either coordinate by at most $1/2$. Therefore, any path connecting $e$ and $e'$ must contain at least $2|\Delta x|_L$ steps to account for the displacement along the horizontal direction, and at least $2|\Delta y|_L$ steps to account for the displacement along the vertical direction, where we use $\lvert \cdot \rvert_L$ to account for the periodicity of the lattice. Consequently, we can lower bound the graph distance between $e$ and $e'$ by
\begin{equation}
    d(e,e')
    \geq
    2\max\left\{
        |\Delta x|_L,
        |\Delta y|_L
    \right\}.
\end{equation}
 
In fact, this bound is tight. To see this, consider first the case where $ |\Delta x|_L \geq |\Delta y|_L$. Let $\sign(\Delta x)$ and $\sign(\Delta y)$ denote the signs of the two components of the displacement, i.e.
\begin{equation}
    \sign(\Delta \mu) =
    \begin{cases}
        1, & \Delta \mu \geq 0,\\
        -1, & \Delta \mu < 0,
    \end{cases}
    \qquad \mu\in\{x,y\}.
\end{equation}
We first take $2|\Delta y|_L$ steps of the form
\begin{equation}
    \left(\frac{\sign(\Delta x)}{2},\frac{\sign(\Delta y)}{2}\right).
\end{equation}
Their total displacement is $\left(\sign(\Delta x) |\Delta y|_L,\Delta y\right)$, so that the required displacement along the $y$ direction has already been achieved. The remaining displacement is therefore $\left(\sign(\Delta x)\bigl(|\Delta x|_L-|\Delta y|_L\bigr),0\right)$.
This can be obtained by taking $|\Delta x|_L-|\Delta y|_L$ pairs of steps of the form
\begin{equation}
    \left(\frac{\sign(\Delta x)}{2},\frac{1}{2}\right),
    \qquad
    \left(\frac{\sign(\Delta x)}{2},-\frac{1}{2}\right),
\end{equation}
since each pair produces the net displacement $ (\sign{(\Delta_x)},0)$.
Therefore, the total number of steps in this path is
\begin{align}
    2|\Delta y|_L
    +2\bigl(|\Delta x|_L-|\Delta y|_L\bigr)
    =
    2|\Delta x|_L.
\end{align}
An analogous construction applies to the case where $|\Delta y|_L \geq |\Delta x|_L$, with the roles of the two coordinates interchanged, yielding a path of length $2|\Delta y|_L$. Hence, in both cases, the lower bound is achievable, and we conclude that
\begin{equation}
    d(e,e')
    =
    2\max\left\{
        |\Delta x|_L,
        |\Delta y|_L
    \right\}.
\end{equation}
In the case that the edges have the same orientation, this has the form
\begin{equation}
    d(e,e')
    =
    2\max\left\{
        |i-p|_L,
        |j-q|_L
    \right\}.
\end{equation}
On the other hand, if the two edges have different orientations, let
$e=H(i,j)$ and $e'=V(p,q)$. From their midpoint coordinates
\begin{equation}
    e-e'
    =
    \left(
        i-p+\frac{1}{2},
        j-q-\frac{1}{2}
    \right),
\end{equation}
and therefore
\begin{equation}
    d(e,e')
    =
    2\max\left\{
        \left|i-p+\frac{1}{2}\right|_L,
        \left|j-q-\frac{1}{2}\right|_L
    \right\}.
\end{equation}
\end{proof}

\section{ILP Formulation}
\label{app:ILP-Formulation}
In this section, we show how the problem in Eq.~\eqref{eq:opt-problem} can be formulated as an integer linear program. Let us start by showing the constraints related to the submeasurements $S^{(t)}_k$ shown in Eqs.~\eqref{eq:eig-cond}-\eqref{eq:prd-const}, which are used to produce the anticommutation certificate. Following Eq.~\eqref {eq:op-const}, we need to encode all the pairings in Eqs.~\eqref{eq:consecutive-pair}-\eqref{eq:nested-pair} as linear constraints over integer variables for every non-target party.  First, the identity-separated pairing $\mathcal{A}\Id\mathcal{A}\Id$ discussed in the main text can be formally written as follows. Let $\pi_1^{(t)}(e)\in\{0,1\}$ be a binary variable indicating whether the
pairing $\mathcal{A}\Id\mathcal{A}\Id$ is selected for a non-target party $e$ in the set of submeasurements that certify the anticommutation relations for the target party $t$. The constraints shown in Eq.~\eqref{eq:alpha-const} and \eqref{eq:gamma-const} can be written as
\begin{equation}
\left\{
\begin{aligned}
&\lvert \Gamma_1^{(t)}(e,r)-\Gamma_3^{(t)}(e,r)\rvert 
    \leq 2\left(1-\pi_1^{(t)}(e)\right),
\\
&\alpha_1^{(t)}(e)\geq \pi_1^{(t)}(e),
\\
&\alpha_3^{(t)}(e)\geq \pi_1^{(t)}(e),
\\
&\alpha_2^{(t)}(e)\leq 1-\pi_1^{(t)}(e),
\\
&\alpha_4^{(t)}(e)\leq 1-\pi_1^{(t)}(e).
\end{aligned}
\right.
\label{eq:pairing-A1A1}
\end{equation}
where the abuse of notation $\lvert \Gamma_1^{(t)}(\ell,r)-\Gamma_3^{(t)}(\ell,r)\rvert$ should be interpreted elementwise, i.e. 
\begin{equation}
     \lvert x_1^{(t)}(\ell)-x_3^{(t)}(\ell)\rvert \leq 2\left(1-\pi_1^{(t)}(e)\right) \qquad \forall \ell : d(e,\ell) \leq r.
\end{equation} 
Notably, the first constraint implies that $\lvert x_1^{(t)}(\ell) - x_3^{(t)}(\ell)\rvert \leq 2$ if $\pi_1^{(t)}(e) = 0$, which does not constrain these variables, as the possible inputs are $\{0,1,2\}$. In the same way, the variables $\alpha_k^{(t)}(e) \in \{0,1\}$ remain unconstrained if $\pi_1^{(t)}(e) = 0$. On the other hand, if $\pi_1^{(t)}(e) = 1$, we recover Eqs.~\eqref{eq:alpha-const} and \eqref{eq:gamma-const}. Similarly, the other identity-separated pairing
$\Id\mathcal{A}\Id\mathcal{A}$ is encoded by introducing
$\pi_2^{(t)}(e)\in\{0,1\}$ and imposing
\begin{equation}
\left\{
\begin{aligned}
&\lvert \Gamma_2^{(t)}(e,r)-\Gamma_4^{(t)}(e,r)\rvert
\leq 2\left(1-\pi_2^{(t)}(e)\right),
\\
&\alpha_2^{(t)}(e)\geq \pi_2^{(t)}(e),
\\
&\alpha_4^{(t)}(e)\geq \pi_2^{(t)}(e),
\\
&\alpha_1^{(t)}(e)\leq 1-\pi_2^{(t)}(e),
\\
&\alpha_3^{(t)}(e)\leq 1-\pi_2^{(t)}(e).
\end{aligned}
\right.
\label{eq:pairing-1A1A}
\end{equation}
Crucially, we do not let $\mathcal{A} = \Id$ in these pairings; otherwise, there would be no need to constrain the context for the party $e$. For the other two pairings, we also need to ensure that at least one of $\mathcal{A}$ and $\mathcal{B}$ is not the identity and participates in a submeasurement.  Therefore, to encode the consecutive pairing
$\mathcal{A}\mathcal{A}\mathcal{B}\mathcal{B}$, we separate it into
three cases:
\begin{align}
\mathcal{A}\mathcal{A}\Id\Id \text{ case: }\qquad
&\left\{
\begin{aligned}
&\lvert \Gamma^{(t)}_1(e,r)-\Gamma^{(t)}_2(e,r)\rvert
\leq 2\left(1-\beta_{\mathrm{consec},1}^{(t)}(e)\right),
\\
&\alpha_1^{(t)}(e)\geq \beta_{\mathrm{consec},1}^{(t)}(e),
\\
&\alpha_2^{(t)}(e)\geq \beta_{\mathrm{consec},1}^{(t)}(e),
\\
&\alpha_3^{(t)}(e)\leq 1-\beta_{\mathrm{consec},1}^{(t)}(e),
\\
&\alpha_4^{(t)}(e)\leq 1-\beta_{\mathrm{consec},1}^{(t)}(e),
\end{aligned}
\right.
\\[1ex]
\Id\Id\mathcal{B}\mathcal{B} \text{ case: }\qquad
&\left\{
\begin{aligned}
&\lvert \Gamma^{(t)}_3(e,r)-\Gamma^{(t)}_4(e,r)\rvert
\leq 2\left(1-\beta_{\mathrm{consec},2}^{(t)}(e)\right),
\\
&\alpha_3^{(t)}(e)\geq \beta_{\mathrm{consec},2}^{(t)}(e),
\\
&\alpha_4^{(t)}(e)\geq \beta_{\mathrm{consec},2}^{(t)}(e),
\\
&\alpha_1^{(t)}(e)\leq 1-\beta_{\mathrm{consec},2}^{(t)}(e),
\\
&\alpha_2^{(t)}(e)\leq 1-\beta_{\mathrm{consec},2}^{(t)}(e),
\end{aligned}
\right.
\\[1ex]
\mathcal{A}\mathcal{A}\mathcal{B}\mathcal{B} \text{ case: }\qquad
&\left\{
\begin{aligned}
&\lvert \Gamma^{(t)}_1(e,r)-\Gamma^{(t)}_2(e,r)\rvert
\leq 2\left(1-\beta_{\mathrm{consec},3}^{(t)}(e)\right),
\\
&\lvert \Gamma^{(t)}_3(e,r)-\Gamma^{(t)}_4(e,r)\rvert
\leq 2\left(1-\beta_{\mathrm{consec},3}^{(t)}(e)\right),
\\
&\alpha_1^{(t)}(e)\geq \beta_{\mathrm{consec},3}^{(t)}(e),
\\
&\alpha_2^{(t)}(e)\geq \beta_{\mathrm{consec},3}^{(t)}(e),
\\
&\alpha_3^{(t)}(e)\geq \beta_{\mathrm{consec},3}^{(t)}(e),
\\
&\alpha_4^{(t)}(e)\geq \beta_{\mathrm{consec},3}^{(t)}(e),
\end{aligned}
\right.
\\[1ex]
&\beta_{\mathrm{consec},1}^{(t)}(e)
+\beta_{\mathrm{consec},2}^{(t)}(e)
+\beta_{\mathrm{consec},3}^{(t)}(e)
=
\pi_3^{(t)}(e),
\qquad
\pi_3^{(t)}(e),\beta_{\mathrm{consec},i}^{(t)}(e)\in\{0,1\}.
\label{eq:consecutive-ILP}
\end{align}
where the last constraint is to ensure that one of these cases is satisfied if $\pi^{(t)}_3(e)=1$, and unconstrained otherwise.
Similarly, the nested pairing
$\mathcal{A}\mathcal{B}\mathcal{B}\mathcal{A}$ is separated into
three cases:
\begin{align}
\mathcal{A}\Id\Id\mathcal{A}\text{ case: } \qquad
&\left\{
\begin{aligned}
&\lvert \Gamma^{(t)}_1(e,r)-\Gamma^{(t)}_4(e,r)\rvert
\leq 2\left(1-\beta_{\text{nested},1}^{(t)}(e)\right),
\\
&\alpha_1^{(t)}(e)\geq \beta_{\text{nested},1}^{(t)}(e),
\\
&\alpha_4^{(t)}(e)\geq \beta_{\text{nested},1}^{(t)}(e),
\\
&\alpha_2^{(t)}(e)\leq 1-\beta_{\text{nested},1}^{(t)}(e),
\\
&\alpha_3^{(t)}(e)\leq 1-\beta_{\text{nested},1}^{(t)}(e),
\end{aligned}
\right.
\\[1ex]
\Id\mathcal{B}\mathcal{B}\Id\text{ case: } \qquad
&\left\{
\begin{aligned}
&\lvert \Gamma^{(t)}_2(e,r)-\Gamma^{(t)}_3(e,r)\rvert
\leq 2\left(1-\beta_{\text{nested},2}^{(t)}(e)\right),
\\
&\alpha_2^{(t)}(e)\geq \beta_{\text{nested},2}^{(t)}(e),
\\
&\alpha_3^{(t)}(e)\geq \beta_{\text{nested},2}^{(t)}(e),
\\
&\alpha_1^{(t)}(e)\leq 1-\beta_{\text{nested},2}^{(t)}(e),
\\
&\alpha_4^{(t)}(e)\leq 1-\beta_{\text{nested},2}^{(t)}(e),
\end{aligned}
\right.
\\[1ex]
\mathcal{A}\mathcal{B}\mathcal{B}\mathcal{A}\text{ case: } \qquad
&\left\{
\begin{aligned}
&\lvert \Gamma^{(t)}_1(e,r)-\Gamma^{(t)}_4(e,r)\rvert
\leq 2\left(1-\beta_{\text{nested},3}^{(t)}(e)\right),
\\
&\lvert \Gamma^{(t)}_2(e,r)-\Gamma^{(t)}_3(e,r)\rvert
\leq 2\left(1-\beta_{\text{nested},3}^{(t)}(e)\right),
\\
&\alpha_1^{(t)}(e)\geq \beta_{\text{nested},3}^{(t)}(e),
\\
&\alpha_2^{(t)}(e)\geq \beta_{\text{nested},3}^{(t)}(e),
\\
&\alpha_3^{(t)}(e)\geq \beta_{\text{nested},3}^{(t)}(e),
\\
&\alpha_4^{(t)}(e)\geq \beta_{\text{nested},3}^{(t)}(e),
\end{aligned}
\right.
\\[1ex]
&\beta_{\text{nested},1}^{(t)}(e)
+\beta_{\text{nested},2}^{(t)}(e)
+\beta_{\text{nested},3}^{(t)}(e)
=
\pi_4^{(t)}(e),
\qquad
\pi_4^{(t)}(e),\beta_{\text{nested},i}^{(t)}(e)\in\{0,1\}.
\label{eq:nested-ILP}
\end{align}

 These pairings can also be straightforwardly generalized for the cases where one has $K_t$ submeasurements, and we denote by $P(K_t)$ the total number of possible pairings in this case. In general, to enforce that a party appearing in at least one submeasurement gets assigned a pairing, it is easy to see that the following constraint must be satisfied for every party $e\neq t$
\begin{equation}
    \sum_{k=1}^{K_t} \alpha_k^{(t)}(e) \leq K_t \sum_{i=1}^{P(K_t)} \pi_i^{(t)}(e).
\end{equation}
The converse must also be true, as no pairing is needed if the party is not in any submeasurement; hence, the following constraint is also necessary
\begin{equation}
    \sum_{i=1}^{P(K_t)} \pi^{(t)}_i(e) \leq P(K_t) \sum_{k=1}^{K_t} \alpha^{(t)}_k(e).
\end{equation}

For the constraint in Eq.~\eqref{eq:op-alt}, we need to enforce that the target party will end up with alternating $2$ and $0$ inputs, as well as the correct canonical context $\Gamma^{c}$, which we done by $\Gamma^{c}_x(t,r)$ for the canonical context for the $X$ operator. and $\Gamma^{c}_z(t,r)$, for the $Z$ operator. For four submeasurements, this amounts to the constraints
\begin{equation}
\left\{
    \begin{aligned}
     &\Gamma^{(t)}_2(t,r) = \Gamma^{(t)}_4(t,r) = \Gamma^c_x(t,r),  \quad x_2^{(t)}(t) = x_4^{(t)}(t) = 0 , \\
     &\Gamma^{(t)}_1(t,r) = \Gamma^{(t)}_3(t,r) = \Gamma^c_z(t,r), \quad x_1^{(t)}(t) = x_3^{(t)}(t) = 2,  \\
     & \alpha^{(t)}_1(t) = \alpha^{(t)}_2(t) = \alpha^{(t)}_3(t) = \alpha^{(t)}_4(t) = 1 ,
\end{aligned}
\right. ,
\label{eq:ZXZX-app}
\end{equation}
For a larger number of submeasurements, the construction can be generalized to allow additional pairings to appear at the target party and to permit the canonical operators to occur in different submeasurements. For example, one may need to allow words such as $Z\Id XZX$ or $ZXZ \mathcal{A}\mathcal{A}X$. It is also worth noting that the choice between obtaining the certificate $ZXZX$ or $XZXZ$ is arbitrary, since any set of submeasurements producing one of these two words can be reordered to produce the other.

Now, the observables appearing in a given submeasurement, namely those
for which $\alpha_k^{(t)}(e)=1$, must arise from the product of primal
and dual contractible loop operators, so that the resulting submeasurement
is, up to a possible sign, a stabilizer of the toric code and therefore
satisfies Eq.~\eqref{eq:eig-cond}. Recall that primal loops are generated
by products of plaquette stabilizers, while dual loops are generated by
products of star stabilizers.

To encode this condition in an integer linear program, let
$p_k^{(t)}(f)\in\{0,1\}$ indicate whether the plaquette stabilizer
associated with face $f$ is included in the product defining
$S_k^{(t)}$. Similarly, let $q_k^{(t)}(v)\in\{0,1\}$ indicate whether
the star stabilizer associated with vertex $v$ is included. For each
edge $e$, denote by $f_1(e)$ and $f_2(e)$ its two adjacent faces, and
by $v_1(e)$ and $v_2(e)$ its two endpoint vertices. An edge $e$ belongs to the resulting primal loop precisely when an odd
number of its two adjacent plaquettes is selected, and it belongs to
the resulting dual loop precisely when an odd number of its two endpoint
stars is selected. We encode these two parities by the binary variables
$\theta_k^{(t)}(e)$ and $\widehat{\theta}_k^{(t)}(e)$, respectively.
The edge participates in the submeasurement if and only if it belongs
to at least one of these two loops, i.e.,
$\alpha_k^{(t)}(e)=1$ if and only if
$\theta_k^{(t)}(e)=1$ or $\widehat{\theta}_k^{(t)}(e)=1$.
These conditions can be written using auxiliary integer variables as
\begin{equation}
\left\{
\begin{aligned}
&p^{(t)}_k(f_1(e)) + p^{(t)}_k(f_2(e))
    = \theta_{k}^{(t)}(e) +2g_{k}^{(t)}(e),
\\
&q^{(t)}_k(v_1(e)) + q^{(t)}_k(v_2(e))
    = \widehat{\theta}_{k}^{(t)}(e)
    +2\widehat{g}_{k}^{(t)}(e),
\\
&\alpha_k^{(t)}(e)\geq \theta_k^{(t)}(e),
\\
&\alpha_k^{(t)}(e)\geq \widehat{\theta}_k^{(t)}(e),
\\
&\alpha_k^{(t)}(e)
    \leq \theta_k^{(t)}(e)+\widehat{\theta}_k^{(t)}(e).
\end{aligned}
\right.
\end{equation}
Furthermore, the input given by that party must follow these loops. More precisely, a party $e$ receives the input $0$ if it is in a dual loop only, an input $1$ if it is in both a dual and a primal loop, and $2$ if it is only in a primal loop. More importantly, if $\alpha^{(t)}_k(e) = 0$, we must not impose any restrictions on $e$'s input. This can be written as the following constraint
\begin{equation}
    \lvert x^{(t)}_k(e) - \theta_{k}^{(t)}(e)  + \widehat{\theta}_k^{(t)}(e) -1 \rvert \leq 2\left(1-\alpha_{k}^{(t)}(e)\right).
\end{equation}

The sign factor of these submeasurements corresponds to the number of intersections between the primal and dual loops. For each intersection, we get a $\pm i$ factor, depending on the order of the performed operations. The number of intersections can be computed by counting the number of parties receiving the input $x^{(t)}_k(e) = 1$. For that, let $y^{(t)}_k(e) \in \{0,1\}$ be a binary variable indicating whether the input received by $e$ is one and participates in the submeasurement. More precisely,  $y^{(t)}_k(e)$ must be subject to
\begin{equation}
\left \{
\begin{aligned}
     &\lvert y^{(t)}_k(e) - 1\rvert \leq \left(1-\alpha^{(t)}_k(e)\right) + \left(1-\theta^{(t)}_k(e)\right) + \left(1-\widehat{\theta}^{(t)}_k(e)\right),
     \\
      & y^{(t)}_k(e)  \leq \alpha^{(t)}_k(e), \\
      & y^{(t)}_k(e)  \leq \theta^{(t)}_k(e), \\
      & y^{(t)}_k(e)  \leq \widehat{\theta}^{(t)}_k(e). \\
\end{aligned}
\right.
\end{equation}
The number of intersections between the primal and dual loop operators
is necessarily even. Indeed, the primal and dual loops are generated
by products of plaquette and star stabilizers, respectively, and these
operators commute. Hence, their supports overlap on an even number of
edges. The same argument applies to the surface code. Therefore, the sign associated with these submeasurements is $\lambda_k^{(t)} = \langle S^{(t)}_k \rangle$, which can be written as an ILP constraint by observing that 
\begin{equation}
    \lambda_k^{(t)}  = (-1)^{\sigma_k^{(t)}},
\end{equation}
where $\sigma_k^{(t)}\in\{0,1\}$ is a binary variable subject to the constraint
\begin{equation}
       \sum_e y^{(t)}_k(e) - 2\sigma_k^{(t)} = 4h^{(t)}_k.
\end{equation}
Hence, the constraint in Eq.~\eqref{eq:prd-const} is equivalent to
\begin{equation}
    \sum_{k=1}^{K_t} \sigma^{(t)}_k = 1 + 2\tau^{(t)},
\end{equation}
where both $h^{(t)}_k$ and $\tau^{(t)}$ are auxiliary integer variables, 

With these constraints, we have completed the ILP formulation regarding the
submeasurements $S_k^{(t)}$. In particular, the constraints above enforce
all the conditions in Eqs.~\eqref{eq:eig-cond}-\eqref{eq:prd-const},
and therefore guarantee that the resulting submeasurements provide the
anticommutation certificate required for the self-testing protocol. To conclude the full formulation of Eq.~\eqref{eq:opt-problem}, we must also specify the constraints related to the submeasurements $S^{(j)}_v$ and $S^{(m)}_f$. However, these are very similar to the ones for anticommutation. The main difference worth noting is regarding the constraint
\begin{equation}
    \prod_{j=1}^{J_v} S^{(v)}_j = \prod_{\partial e \ni v} X_e^{\Gamma_X}.
\end{equation}
Contrary to the relation in Eq.~\eqref{eq:op-alt}, the order in which each term on the right-hand side appears is not defined. This means that for every party $\partial e \ni v$, we must allow pairings similar to those shown above. For example, for $J_v = 4$ submeasurements, one must allow the pairings
\begin{equation}
    \begin{aligned}
       & X^{\Gamma_X} \mathcal{A}  \mathcal{A} \Id, \\
        &X^{\Gamma_X} \mathcal{A}  \Id \mathcal{A}, \\
        &X^{\Gamma_X} \Id \mathcal{A}  
        \mathcal{A}, \\
        &\Id X^{\Gamma_X} \mathcal{A} \mathcal{A},
    \end{aligned}
    \qquad 
    \begin{aligned}
       & \mathcal{A}  \mathcal{A} X^{\Gamma_X} \Id, \\
        & \Id \mathcal{A}  \mathcal{A} X^{\Gamma_X}, \\
        & \mathcal{A}  \Id \mathcal{A} X^{\Gamma_X}, \\
        &\mathcal{A}   \mathcal{A} \Id X^{\Gamma_X}, 
    \end{aligned}
    \quad    
    \begin{aligned}
       & X^{\Gamma_X} \Id  \Id \Id, \\
        &\Id X^{\Gamma_X}   \Id \Id, \\
        &\Id \Id X^{\Gamma_X}  \Id, \\
        &\Id \Id \Id X^{\Gamma_X}.
    \end{aligned}
    \qquad 
\end{equation}
An analogous construction is also required for the certification of the
plaquette stabilizers. In this case, the local pairings are defined so
that the product of the corresponding submeasurements reduces to the
desired canonical $Z$ operator on each edge of the plaquette and to the
identity on all other parties. These pairings can be encoded using the
same construction introduced above.

\begin{remark}
    In general, an ILP admits multiple equivalent formulations, which may differ substantially in their computational runtime. The formulation presented above is therefore not unique and was chosen primarily for clarity and pedagogical purposes. A variety of standard modeling techniques can be used to strengthen the formulation and improve the performance of generic ILP solvers; see, e.g., \cite{Vielma2015}. The numerical results reported in Section \ref{sec:MILP} were obtained using a more tailored implementation, in which the search space was reduced by exploiting properties of the problem and by imposing additional constraints explained in the main text that accelerate the solution process. The implementation of techniques to improve the performance of the ILP without restricting the search space is then left to future work.
\end{remark}

\subsection{Graph states}
\label{app:graph-state}

In this section, we show how to modify the above constraints for self-testing graph states. Notably, the general idea of the constraints remains the same as in the toric code. Indeed, all the constraints regarding the pairings, as well as Eq.~\eqref{eq:ZXZX-app}, remain exactly the same. 

The first difference arises when enforcing that each submeasurement is obtained from a product of the graph stabilizers, rather than from the primal and dual loop construction considered above. Let $q_k^{(t)}(v)\in{0,1}$ indicate whether the graph stabilizer associated with vertex $v$ is included in the product defining $S_k^{(t)}$. Recall that selecting the stabilizer centered at $v$ contributes an $X$ operator at $v$, while every selected stabilizer centered at a neighboring vertex $u\in N(v)$ contributes a $Z$ operator at $v$. Consequently, the input $x_k^{(t)}(v)$ depends both on $q_k^{(t)}(v)$ and on the parity of the selected stabilizers in the neighborhood of $v$. More specifically, if an even number of neighboring vertices satisfy $q_k^{(t)}(u)=1$, all the corresponding $Z$ contributions cancel. In this case, the input of party $v$ is $0$ if $q_k^{(t)}(v)=1$, corresponding to an $X$ operator, and is unconstrained otherwise. Conversely, if an odd number of neighboring vertices satisfy $q_k^{(t)}(u)=1$, a single $Z$ contribution remains at $v$. The input is then $2$ if $q_k^{(t)}(v)=0$, corresponding to a $Z$ operator, and $1$ if $q_k^{(t)}(v)=1$, corresponding to a $Y$ operator.

To encode this parity, we introduce a binary variable $z_k^{(t)}(v)\in\{0,1\}$, together with an auxiliary integer variable $g_k^{(t)}(v)$, such that
\begin{equation}
\sum_{u\in N(v)} q_k^{(t)}(u)
=
z_k^{(t)}(v)
+
2g_k^{(t)}(v).
\end{equation}
Thus, $z_k^{(t)}(v)=1$ if and only if an odd number of neighboring stabilizers are selected. Now the party $v$ will participate in a given submeasurement if and only if it is influenced by one of the selected stabilizers. More precisely, we have
\begin{equation}
    \left \{
    \begin{aligned}
        & \alpha^{(t)}_k(v) \geq q_k^{(t)}(v) \\
        & \alpha^{(t)}_k(v) \geq z_k^{(t)}(v) \\
        & \alpha^{(t)}_k(v) \leq q_k^{(t)}(v) +  z_k^{(t)}(v) \\
        \end{aligned}
    \right.
\end{equation}
Furthermore, we constrain the input of each party depending on whether it participates or not in the submeasurement by imposing
\begin{equation}
\left|
x_k^{(t)}(v)
-
\left(1+z_k^{(t)}(v)-q_k^{(t)}(v)\right)
\right|
\leq
2\left(1-\alpha_k^{(t)}(v)\right).
\end{equation}
Indeed, when $\alpha_k^{(t)}(v)=1$, the constraint enforces
\begin{equation}
x_k^{(t)}(v)
=
1+z_k^{(t)}(v)-q_k^{(t)}(v),
\end{equation}
yielding inputs $0$, $1$, and $2$ for the local Pauli operators $X$, $Y$, and $Z$, respectively. On the other hand, when $\alpha_k^{(t)}(v)=0$, necessarily $q_k^{(t)}(v)=z_k^{(t)}(v)=0$, and the constraint reduces to $\lvert x_k^{(t)}(v)-1\rvert\leq 2$. Since $x_k^{(t)}(v)\in\{0,1,2\}$, this leaves the input completely unconstrained, as required.

The second difference to the toric code arises when we need to compute the sign of the total submeasurement. As before, the sign of the submeasurement depends on the appearance of the Pauli $Y$, substituting the cases in which a party would otherwise be acted on by both $X$ and $Z$. Let $y_k^{(t)}(v)\in\{0,1\}$ indicate whether party $v$ receives input $1$ and participates in the submeasurement. From the construction above, this occurs precisely when both the stabilizer centered at $v$ is selected and an odd number of neighboring stabilizers are selected, i.e., when $q_k^{(t)}(v)=z_k^{(t)}(v)=1$. Hence,
\begin{equation}
\left\{
\begin{aligned}
&y_k^{(t)}(v)\leq q_k^{(t)}(v),\\
&y_k^{(t)}(v)\leq z_k^{(t)}(v),\\
&y_k^{(t)}(v)\geq q_k^{(t)}(v)+z_k^{(t)}(v)-1.
\end{aligned}
\right.
\end{equation}
Therefore, $n_y^{(t)} = \sum_v y_k^{(t)}(v)$ gives the number of $Y$ substitutions appearing in $S_k^{(t)}$. In contrast to the toric code case, however, this number is not the only factor contributing to the final sign of the submeasurement. As an example, let us compute the product of two consecutive vertices $1$ and $2$ in a three-vertex line graph
\begin{align}
    G_1G_2 &= (X_1Z_2)(Z_1X_2Z_3) \nonumber\\
    &= (-iY_1)(+iY_2)Z_3 \\
    & = + Y_1Y_2Z_3. \nonumber
\end{align}
   
On the other hand, the product of all the stabilizers is
\begin{align}
    G_1G_2G_3&= (X_1Z_2)(Z_1X_2Z_3) (Z_2X_3)\nonumber\\
    &= (-iY_1)(-X_2)(+iY_3) \\
    & = - Y_1X_2Y_3. \nonumber
\end{align}
So even though both products of stabilizers have the same number of Y's, they nonetheless have different signs attributed to them. To understand the other factor contributing to the sign, note that every time two vertices $u$ and $v$ sharing an edge are selected, the resulting operator has the form  
\begin{align}
    G_{u}G_v &= (X_uZ_vZ^{N(u)})(X_vZ_uZ^{N(v)}) \nonumber\\
    &= -X_uX_v Z_u Z_v Z^{N(u)} Z^{N(v)}.
\end{align}
Therefore, the product of all the stabilizers associated with the vertices $u$ in a set $Q$ is 
\begin{equation}
    \prod_{u \in Q} G_u = (-1)^{m} \prod_{u \in Q} X_u \prod_{u\in N(Q)_{\text{odd}}} Z_u,
\end{equation}
where $N(Q)_{\text{odd}}$ denotes the set of vertices that are in the neighborhood of an odd number of vertices in the set $Q$, and $m$ is the number of edges in the graph connecting two vertices that are inside $Q$. We can then further simplify this by substituting $XZ = -iY$, yielding
\begin{equation}
    \prod_{u \in Q} G_u = (-1)^{m + n_y/2} \prod_{\substack{u \in Q\\ u \not \in N(Q)_{\text{odd}}}} X_u \prod_{\substack{u\in N(Q)_{\text{odd}} \\ u \not \in Q}} Z_u \prod_{u \in Q \cap N(Q)_{\text{odd}}} Y_u ,
\end{equation}
where we used the fact that $(-1)^{3n_y/2}  = (-1)^{n_y/2}$. Thus, to compute the total overall sign of the submeasurements, we also need to account for the number of edges $e$ such that $q^{(t)}_k(v)=1$ for all $v \in \partial e$. For that, let $\eta^{(t)}_k(u,v) \in \{0,1\}$, we can make $\eta^{(t)}_k$ encode exactly this property for the edge $(u,v)$ by imposing
\begin{equation}
    \left \{
    \begin{aligned}
        &\eta^{(t)}_k(u,v) \leq q_k^{(t)}(u), \\
        &\eta^{(t)}_k(u,v) \leq q_k^{(t)}(v), \\
        &\eta^{(t)}_k(u,v) \geq q_k^{(t)}(u) + q_k^{(t)}(v)-1.  \\
    \end{aligned}
    \right.
\end{equation}
We then have
\begin{equation}
    m^{(t)}_k = \sum_{(u,v) \in E} \eta^{(t)}_k(u,v).
\end{equation}

Finally, to impose 
\begin{equation}
 \langle S^{(t)}_k\rangle = \lambda_k^{(t)} = (-1)^{\sigma^{(t)}_k},   
\end{equation}
we add the constraint
\begin{equation}
    \sum_{v \in V} y^{(t)}_k(v) + 2 \sum_{(u,v) \in E} \eta^{(t)}_k(u,v) - 2 \sigma^{(t)}_k = 4h^{(t)}_k,
\end{equation}
where $h^{(t)}_k$ is an auxiliary integer variable. All the other constraints to define the stabilizer certificates follow similarly.

\section{Proof of Proposition \ref{prop:robust-prop}}
\label{app:proof-noise}
Here, we prove a more general form of Proposition~\ref{prop:robust-prop}.
\begin{proposition}
Let
    \begin{equation}
    \left \{\left \{S^{(t)}_k\right\}_k^{K_t}\right\}_{t},\quad \left \{\left \{S^{(v)}_j\right\}_j^{J_v}\right\}_{v}, \quad \text{and}\quad \left \{\left \{S^{(f)}_m\right\}_m^{M_f}\right\}_{f}    
    \end{equation}
    be submeasurements that are solutions to Eq.~\eqref{eq:opt-problem}.
    For any PE $\epsilon$-simulating these submeasurements, there exists an isometry $\Phi = (\bigotimes_i \Phi_i)\otimes \Id_P$ such that
\begin{align}
     \lVert \Phi  \ket{\Psi}  - (\Pi_{TC}\otimes \Id ) \Phi \ket{\Psi}  \rVert_2 \leq \delta(\epsilon),
\end{align}
where
\begin{equation}
    \delta(\epsilon) =  \frac{\sqrt\epsilon}{2} \left[\sum_v\left(\sqrt{2}J_v + \sum_{\partial e \ni v} K_e\right) + \sqrt{2} \sum_f M_f\right].
\end{equation}
\end{proposition}
\begin{proof}
To prove the noise robustness of the constructions shown in the main text, we closely follow the technique introduced in \cite{McKague2014}, which was also used in \cite{Meyer2026}. Crucially, it relies on two identities which we will use repeatedly.

First, we use the fact that if $\bra{\Psi} M \ket{\Psi} \geq 1-\epsilon$, then 
\begin{equation}
    \lVert \ket{\Psi} - M \ket{\Psi}\rVert_2 \leq \sqrt{2 \epsilon}.
\end{equation}
Secondly, we also use the fact that if $\lVert \ket{\Psi} - M\ket{\Psi} \rVert_2 \leq \alpha$, $\lVert \ket{\Psi} - N\ket{\Psi} \rVert_2 \leq \beta $, and $\lVert M \rVert_\infty = 1$, we have
\begin{equation}
    \lVert \ket{\Psi} - MN\ket{\Psi}\rVert_2 \leq \alpha + \beta.
\end{equation}

We then start by noting that if the PE $\epsilon$-simulates the RE, we have 
\begin{equation}
    \left \lvert \left \langle \widetilde{S}^{(t)}_k\right \rangle - \lambda_k \right \rvert \leq \epsilon, 
\end{equation}
which implies
\begin{equation}
    \left \lVert \ket{\Psi} - \lambda_k \widetilde{S}_k^{(t)} \ket{\Psi}\right \rVert_2 \leq \sqrt{2\epsilon}.
\end{equation}
Therefore
\begin{equation}
    \left \lVert \ket{\Psi} - \prod_{k=1}^{K_t}\lambda_k \widetilde{S}_k^{(t)} \ket{\Psi}\right \rVert_2 \leq \sqrt{2\epsilon} K_t.
\end{equation}
Because $\widetilde{S}^{(t)}_k$ are solutions of Eq.~\eqref{eq:opt-problem}, we have
\begin{equation}
    \left \lVert \ket{\Psi} +  \widetilde{Z}_t\widetilde{X}_t\widetilde{Z}_t\widetilde{X}_t \ket{\Psi}\right \rVert_2  = \left\lVert \left\{\widetilde{Z}_t,\widetilde{X}_t\right\} \ket{\Psi}\right \rVert_2\leq \sqrt{2\epsilon} K_t.
    \label{eq:anticomm-bound}
\end{equation}
A similar calculation for the $\epsilon$-simulation of the submeasurements $S^{(v)}_j$ and $S^{(f)}_m$ yield
\begin{align}
    &\left \lVert \ket{\Psi} -  \widetilde{A}_v \ket{\Psi}\right \rVert_2 \leq \sqrt{2 \epsilon} J_v, \\
    &\left \lVert \ket{\Psi} -  \widetilde{B}_f \ket{\Psi}\right \rVert_2 \leq \sqrt{2 \epsilon} M_f.
\end{align}

We can now use the bound in Eq.~\eqref{eq:anticomm-bound} to show
\begin{align}
    \lVert \bigl[({X}_e \otimes \Id_{\mathcal{H}''\otimes\mathcal{H}_P}) (\Phi_e\otimes \Id_P) - (\Phi_e\otimes \Id_P) &\widetilde{X}_e \bigr]\ket{\Psi} \rVert_2 ^2 \nonumber \\
    =& \left\lVert \ket{0}\otimes\left( -\frac{1}{2}\left\{\widetilde{Z}_e,\widetilde{X}_e\right\} \ket{\Psi}\right) +\ket{1} \otimes \left(\frac{\widetilde{X}_e}{2}\left\{\widetilde{Z}_e,\widetilde{X}_e\right\} \ket{\Psi}\right)\right \rVert_2 ^2\nonumber \\
     = & \left\lVert \ket{0}\otimes\left( -\frac{1}{2}\left\{\widetilde{Z}_e,\widetilde{X}_e\right\} \ket{\Psi}\right)\right \rVert_2^2 + \left\lVert\ket{1} \otimes \left(\frac{\widetilde{X}_e}{2}\left\{\widetilde{Z}_e,\widetilde{X}_e\right\} \ket{\Psi}\right)\right \rVert_2 ^2\nonumber\\
    &\leq \epsilon K_e^2,
\end{align}
where $\Phi_e$ is SWAP isometry acting on party $e$, and trivially on the other parties (see Appendix \ref{app:proof-prop1}), and the second equality comes from the orthogonality between $\ket{0}$ and $\ket{1}$. On the other hand, the analogous relation for Pauli $Z$ operators does not require the use of the bound on the anticommutator and is simply
\begin{equation}
    \left\lVert \left[({Z}_e \otimes \Id_{\mathcal{H}''\otimes \mathcal{H}_P}) (\Phi_e\otimes \Id_P) - (\Phi_e\otimes \Id_P) \widetilde{Z}_e\right ]\ket{\Psi}\right \rVert_2 = 0.
\end{equation}
Now, for two different edges $e$ and $e'$ incident on $v$, we have
\begin{align}
\Bigl\lVert \Bigl[ (X_eX_{e'}\otimes \Id_{\mathcal{H}'' \otimes \mathcal{H}_P})&(\Phi_e  \Phi_{e'}\otimes \Id_P) -
(\Phi_e  \Phi_{e'}\otimes \Id_P)\widetilde X_e\widetilde X_{e'}
\Bigr]\ket{\Psi}
\Bigr\rVert_2 \nonumber \\ 
&\leq
\Bigl\lVert
(X_e\otimes \Id_{\mathcal H'' \otimes \mathcal{H}_P})
\Bigl[
(X_{e'}\otimes \Id_{\mathcal H'' \otimes \mathcal{H}_P})(\Phi_e  \Phi_{e'}\otimes \Id_P)-
(\Phi_e  \Phi_{e'}\otimes \Id_P)\widetilde X_{e'}
\Bigr]\ket{\Psi}
\Bigr\rVert_2
\nonumber \\ &  \qquad \text{  } \qquad \text{  } \qquad \text{  }   +
\Bigl\lVert
\Bigl[
(X_e\otimes \Id_{\mathcal H''\otimes \mathcal{H}_P})(\Phi_e  \Phi_{e'}\otimes \Id_P)
-
(\Phi_e  \Phi_{e'}\otimes \Id_P)\widetilde X_e
\Bigr]
\widetilde X_{e'}\ket{\Psi}
\Bigr\rVert_2
\nonumber\\
&=
\Bigl\lVert
\Bigl[
(X_{e'}\otimes \Id_{\mathcal H'' \otimes \mathcal{H}_P})(\Phi_e  \Phi_{e'}\otimes \Id_P)
-
(\Phi_e  \Phi_{e'}\otimes \Id_P)\widetilde X_{e'}
\Bigr]\ket{\Psi}
\Bigr\rVert_2
& \nonumber\\
&\qquad \text{  } \qquad \text{  } \qquad \text{  }  +
\left\lVert
\left[
(X_e\otimes \Id_{\mathcal H'' \otimes \mathcal{H}_P})(\Phi_e  \Phi_{e'}\otimes \Id_P)
-
(\Phi_e  \Phi_{e'}\otimes \Id_P)\widetilde X_e
\right]\ket{\Psi}
\right\rVert_2
\nonumber\\
&\leq
\sqrt{\epsilon}\left(K_e+K_{e'}\right).
\end{align}
Here, the equality follows from the unitarity of $X_e$ and $\widetilde X_{e'}$, together with the fact that the observables associated with different edges act on distinct subsystems. Repeating the same argument for all edges satisfying $\partial e\ni v$ gives
\begin{equation}
\left\lVert
\left[
(A_v\otimes \Id_{\mathcal H'' \otimes \mathcal{H}_P})(\Phi\otimes \Id_P)
-
(\Phi\otimes \Id_P)\widetilde A_v
\right]\ket{\Psi}
\right\rVert_2
\leq
\sqrt{\epsilon}\sum_{\partial e\ni v}K_e.
\end{equation}
Therefore, we have 
\begin{align}
        \left\lVert
\left( \Id - A_v \otimes \Id_{\mathcal{H}'' \otimes \mathcal{H}_P}
\right)(\Phi\otimes \Id_P) \ket{\Psi}
\right\rVert_2 &\leq
\left \lVert
(\Phi\otimes \Id_P) \ket{\Psi} - (\Phi\otimes \Id_P) \widetilde{A}_v \ket{\Psi}
\right\rVert_2  \nonumber \\
& \qquad \text{ }\qquad \text{ }  \quad+  \left\lVert
\left[
(A_v\otimes \Id_{\mathcal H'' \otimes \mathcal{H}_P})(\Phi\otimes \Id_P)
-
(\Phi\otimes \Id_P)\widetilde A_v
\right]\ket{\Psi}
\right\rVert_2 \nonumber \\
&\leq \sqrt{2\epsilon} J_v + \sqrt{\epsilon}\sum_{\partial e \ni v}K_e.
\end{align}

A similar calculation also implies
\begin{equation}
\left\lVert
\left[
(B_f\otimes \Id_{\mathcal H'' \otimes \mathcal{H}_P})(\Phi\otimes \Id_P)
-
(\Phi\otimes \Id_P)\widetilde B_f
\right]\ket{\Psi}
\right\rVert_2 = 0, 
\end{equation}
and 
\begin{align}
        \left\lVert
\left( \Id - B_f \otimes \Id_{\mathcal{H}'' \otimes \mathcal{H}_P}
\right)(\Phi\otimes \Id_P) \ket{\Psi}
\right\rVert_2 &\leq
\left \lVert
(\Phi\otimes \Id_P) \ket{\Psi} - (\Phi\otimes \Id_P) \widetilde{B}_f \ket{\Psi}
\right\rVert_2 \nonumber \\ 
& \qquad \text{ }\qquad \text{ }  \quad + \left\lVert
\left[
(B_f\otimes \Id_{\mathcal H''\otimes \mathcal{H}_P})(\Phi\otimes \Id_P)
-
(\Phi\otimes \Id_P) \widetilde B_f
\right]\ket{\Psi}
\right\rVert_2 \nonumber \\
&\leq \sqrt{2\epsilon} M_f.
\end{align}

Finally, we can conclude
\begin{align}
     \Bigl\lVert
    (\Phi\otimes \Id_P)\ket{\Psi}
    -
    &(\Pi_{TC}\otimes \Id_{\mathcal H'' \otimes \mathcal{H}_P})(\Phi\otimes \Id_P)\ket{\Psi}
    \Bigr\rVert_2
     \nonumber \\&= \left \lVert (\Phi\otimes \Id_P) \ket{\Psi} - \prod_v \frac{[\Id + (A_v \otimes \Id_{\mathcal{H''}})]}{2}\prod_f \frac{[\Id + (B_f \otimes \Id_{\mathcal{H''}})]}{2} (\Phi\otimes \Id_P) \ket{\Psi}\right \rVert_2 \leq 
    \delta(\epsilon),
\end{align}
where
\begin{equation}
    \delta(\epsilon) =  \frac{\sqrt\epsilon}{2} \left[\sum_v\left(\sqrt{2}J_v + \sum_{\partial e \ni v} K_e\right) + \sqrt{2} \sum_f M_f\right].
\end{equation}
\end{proof}

\section{Classical simulation with local communication}
\label{sec:communicating-classical-strategies}

In this Appendix, we examine whether previous constructions proposed to self-test the toric code remain valid when the parties have access to some amount of classical communication. In Subsection~\ref{sec:toric-bell-one-round1}, we analyze the Bell inequality introduced in \cite{Baccari2020} and show that its quantum bound can be attained by a classical strategy when the parties have access to only a single round of classical communication. In Subsection~\ref{sec:toric-bell-one-round2}, we consider the CSS submeasurement game proposed in \cite{Hart2025}. There, we show that the statistics of a maximal winning strategy of the game can be simulated classically with a single round of communication.

\subsection{One-round local strategies achieve the quantum bound in  \cite{Baccari2020}}
\label{sec:toric-bell-one-round1}

Throughout this appendix, we follow the notation and conventions of \cite{Baccari2020}, where the scenario is defined by $N = 2L^2$ parties placed on the edges of an $L \times L$ torus. The associated stabilizers are the vertex and plaquette operators
\begin{equation}
  S_v = \prod_{i \in v} X^{(i)},
  \qquad
  S_p = \prod_{i \in p} Z^{(i)},
\end{equation}
where $i\in v$ denotes the four edges connected to the vertex $v$, and $i \in p$, the four edges defining a plaquette $p$. The corresponding Bell inequality is
\begin{equation}
  I_N^{\mathrm{tor}}
  =
  \sum_v \left \langle \widetilde S_v \right \rangle
  +
  \sum_p \left \langle \widetilde S_p \right \rangle,
\end{equation}
where the operators $\widetilde{S}_v$ and  $\widetilde{S}_p$ are obtained by choosing a distinguished edge $j$, and substituting the observables $X^{(j)}$ and $Z^{(j)}$ by the tilted observables $(A^{j}_0+A^{j}_1)/\sqrt{2}$ and $(A^{j}_0-A^{j}_1)/\sqrt{2}$, while every other edge substitutes the observables $X^{(i)}$ and $Z^{(i)}$ for $A_0^{(i)}$ and $A_1^{(i)}$, respectively. More explicitly
\begin{equation}
   X^{(i)}
  =
  \begin{cases}
    A_0^{(i)}, & i \neq j, \\[1mm]
    \frac{A_0^{(j)}+A_1^{(j)}}{\sqrt{2}}, & i=j,
  \end{cases}
  \qquad
  Z^{(i)}
  =
  \begin{cases}
    A_1^{(i)}, & i \neq j, \\[1mm]
    \frac{A_0^{(j)}-A_1^{(j)}}{\sqrt{2}}, & i=j .
  \end{cases}
\end{equation}
It was shown in \cite{Baccari2020} that the maximal quantum value of this expression is
\begin{equation}
   I_N^{\mathrm{tor}} \overset{Q}{\leq} 
  |V|+|P|
  =
  N,
\end{equation}
and that attaining this value self-tests the toric code subspace in the standard non-communicating scenario.

We now show that, once nearest-neighbor classical communication is allowed according to the same communication graph introduced in the main text, the same maximal value can already be attained by a classical strategy using a single round of input communication. 

Since communication makes the local output of a party potentially depend on inputs outside the support of the correlator being evaluated, one must specify a full input string for each correlator.

\begin{theorem}
  \label{thm:toric-bell-one-round}
  The maximal quantum value of $I_N^{\mathrm{tor}}$ can be attained by a classical strategy using one
  round of nearest-neighbor classical communication. More precisely,
  there exists a one-round classical strategy satisfying
  \begin{equation}
    \left \langle \widetilde S_v \right \rangle = 1
    \qquad \forall v,
  \end{equation}
  and
  \begin{equation}
    \left \langle \widetilde S_p \right \rangle = 1
    \qquad \forall p.
  \end{equation}
\end{theorem}
\begin{proof} We prove this by showing an explicit strategy attaining the quantum bound. For every party $i \neq j$, the strategy is to deterministically output $a_i = +1$, independently of the input, which for every vertex $v$ and plaquette $p$ not containing the party $j$ already satisfies
\begin{equation}
    \left \langle \widetilde S_v \right \rangle = \left \langle \prod_{i \in v} a_i\right \rangle= 1, \qquad \left \langle \widetilde S_p \right \rangle  = \left \langle \prod_{i \in p} a_i\right \rangle= 1.
  \end{equation}
  It therefore remains only to specify the strategy of the distinguished party $j$. Let $p$ be one of the two plaquettes containing the edge $j$. We denote by
  \begin{equation}
    \partial_j p
    =
    \left\{
      i \in p \setminus\{j\}
      :
      i \text{ shares a vertex with } j
    \right\}
  \end{equation}
  the two edges of $p$ which are adjacent to $j$, as $|\partial_j p|=2$. Moreover, both edges in $\partial_j p$ are neighbors of $j$ in the communication graph. Hence, after one round of communication, party $j$ knows the inputs $x_i$ for every $i \in \partial_j p$, for each of the two plaquettes $p \ni j$. Define
  \begin{equation}
    C_j
    =
    \begin{cases}
      1,
      &
      \exists p \ni j, \; \forall i \in \partial_j p, \; x_i = 1 \\
      0,
      &
      \text{otherwise}.
    \end{cases}
  \end{equation}
  Thus, $C_j$ depends only on information available to $j$ after one round of communication. Let $R \in\{-1,+1\}$ be a random variable satisfying $\mathbb{E}[R]=1/\sqrt{2}$, that is
  \begin{equation}
    \Pr[R=-1]
    =
    \frac{1-1/\sqrt{2}}{2},
    \qquad
    \Pr[R=+1]
    =
    \frac{1+1/\sqrt{2}}{2}.
  \end{equation}
The strategy for the party $j$ is then to output
  \begin{equation}
    a_j
    =
    \begin{cases}
      -R,
      &
      x_j=1 \text{ and } C_j=1,
      \\[1mm]
      +R,
      &
      \text{otherwise}.
    \end{cases}
  \end{equation}
  
  In the following, we show that this strategy attains the quantum value for the remaining correlators. Let us first consider the correlators associated with a vertex $v$ containing $j$. The three other edges of $v$ will thus receive the input $0$. Consider either of the two plaquettes $p$ containing $j$. Exactly one of the two edges in $\partial_j p$ is also incident on the vertex $v$. Since this edge belongs to $v \setminus\{j\}$, it must have received the input $0$. Consequently, for both plaquettes $p \ni j$, it is impossible that all the inputs on $\partial_j p$ are equal to $1$, and therefore
  \begin{equation}
    C_j=0.
  \end{equation}
  Hence party $j$ outputs $+R$, independently of whether its own input is $0$ or $1$. Since every other party outputs $+1$, it follows that
  \begin{equation}
    \left \langle
    A_0^{(j)}
    \prod_{i \in v \setminus\{j\}} A_0^{(i)}
    \right \rangle
    =
    \mathbb{E}[R]
    =
    \frac{1}{\sqrt{2}},
  \end{equation}
  and likewise
  \begin{equation}
    \left \langle
    A_1^{(j)}
    \prod_{i \in v \setminus\{j\}} A_0^{(i)}
    \right \rangle
    =
    \mathbb{E}[R]
    =
    \frac{1}{\sqrt{2}}.
  \end{equation}
  Therefore,
  \begin{align}
    \left \langle \widetilde S_v \right \rangle
    &=
    \left \langle
    \frac{A_0^{(j)}+A_1^{(j)}}{\sqrt{2}}
    \prod_{i \in v \setminus\{j\}} A_0^{(i)}
    \right \rangle
    \nonumber\\
    &=
    \frac{1}{\sqrt{2}}
    \left(
      \frac{1}{\sqrt{2}}
      +
      \frac{1}{\sqrt{2}}
    \right)
    \nonumber\\
    &=1.
  \end{align}
  Next, consider the correlators associated with a plaquette $p$ be containing $j$. The three remaining edges of $p$ receive input $1$. In particular, the two edges belonging to $\partial_j p$ both receive input $1$. Hence
  \begin{equation}
    C_j=1.
  \end{equation}
  If the input of $j$ is $0$, then the output is $+R$, whereas if the input
  of $j$ is $1$, the output is $-R$. Therefore
  \begin{equation}
    \left \langle
    A_0^{(j)}
    \prod_{i \in p \setminus\{j\}} A_1^{(i)}
    \right \rangle
    =
    \mathbb{E}[R]
    =
    \frac{1}{\sqrt{2}},
  \end{equation}
  while
  \begin{equation}
    \left \langle
    A_1^{(j)}
    \prod_{i \in p \setminus\{j\}} A_1^{(i)}
    \right \rangle
    =
    -\mathbb{E}[R]
    =
    -\frac{1}{\sqrt{2}}.
  \end{equation}
  Consequently,
  \begin{align}
    \left \langle \widetilde S_p \right \rangle
    &=
    \left \langle
    \frac{A_0^{(j)}-A_1^{(j)}}{\sqrt{2}}
    \prod_{i \in p \setminus\{j\}} A_1^{(i)}
    \right \rangle
    \nonumber\\
    &=
    \frac{1}{\sqrt{2}}
    \left(
      \frac{1}{\sqrt{2}}
      -
      \left(-\frac{1}{\sqrt{2}}\right)
    \right)
    \nonumber\\
    &=1.
  \end{align}

  Which allow us to conclude that this strategy attains
  \begin{equation}
    I_N^{\mathrm{tor}}
    =
    |V|+|P|
    =
    N.
  \end{equation}
\end{proof}

\subsection{The CSS submeasurement game from \cite{Hart2025} is not rigid against $1$ round of communication}
\label{sec:toric-bell-one-round2}

Throughout this subsection, we follow the notation and conventions of \cite{Hart2025}. In the version of their proposed game tailored for the toric code, a party is associated with every edge $e$ of the lattice, and the referee queries the parties by assigning as inputs a pair of bits $(a_e, b_e)$, associated with a Pauli string $P$. Their corresponding local observables are denoted by
\begin{equation}
    E_e, A_e, B_e, C_e,
\end{equation}
for the inputs
\begin{equation}
    (a_e, b_e) = (0,0), (0,1), (1,0), (1,1),
\end{equation}
respectively. In the reference experiment, these correspond exactly to the Pauli operators $\Id, Z, X, Y$. The game proposed in \cite{Hart2025} can then be won if, for all queries of Pauli strings $P$ corresponding to a substring of plus or minus a stabilizer of the toric code, the player outputs a bit $y_i(a_i,b_i)$ such that
\begin{equation}
    \sum_{i \in \Lambda_P} y_i = \frac{1}{2} \sum_{i \in \Lambda_P} a_ib_i  \mod 2,
    \label{eq:win-cond}
\end{equation}
where $\Lambda_P = \supp{P}$. To prove the rigidity of the optimal strategy of this game, the authors in \cite{Hart2025} first use the stabilizer-generator queries that imposes that, when acting on their shared state, we have
\begin{equation}  \label{eq:Hart2025_stabilizer}
    \left \langle \prod_{\partial e \ni  v} A_e \right \rangle = +1, \qquad \left \langle \prod_{e \in \partial p} B_e \right \rangle = +1,
\end{equation}
and subsequently uses the four local queries which correspond to the submeasurements $S_{12}, S_{23}, S_{31}$, and $S_{0}$, which can be diagrammatically represented as 
\begin{equation}\label{eq:Hart2025_anticommutation}
\begin{tikzpicture}[
    redbox/.style={draw=red, thick, minimum size=1cm},
    bluebox/.style={draw=blue, thick, minimum size=1cm},
    bluerect/.style={draw=blue, thick, minimum width=1cm, minimum height=2cm},
    cross/.style={circle, draw, fill=white, inner sep=2pt}
]

\node at (0,0) {$S_{12}\coloneqq$};
\node[redbox]  at (1.5,0) {};
\node[bluebox] at (2,.5) {};
\node[cross] at (1.5,.5) {};
\node[cross] at (2,0) {};

\node at (4,0) {$S_{23}\coloneqq$};
\node[redbox]  at (5.5,0) {};
\node[bluebox] at (6,-.5) {};
\node[cross] at (6,0) {};
\node[cross] at (5.5,-.5) {};

\node at (8,0) {$S_{31}\coloneqq$};
\node[redbox]   at (9.5,0) {};
\node[bluerect] at (10,0) {};
\node[cross] at (9.5,.5) {};
\node[cross] at (9.5,-.5) {};

\node at (12,0) {$S_0\coloneqq$};
\node[redbox] at (13.5,0) {};

\end{tikzpicture}
\end{equation}
as well as its translated versions, to certify the anticommutation relations of these observables. Here, red edges denote the parties applying $A_e$, blue edges correspond to $B_e$, and white circles to $C_e$. For these queries, the win condition in Eq.~\eqref{eq:win-cond} imply that these operators satisfy
\begin{equation}
    \langle S_{12} \rangle = \langle S_{23} \rangle = \langle S_{31} \rangle = -1, \qquad \langle S_{0} \rangle = +1.
\end{equation}
These correlators are then used to show that the only possible winning strategy for their game corresponds to those where the parties share a state in the toric code.
At the end of their rigidity argument, the queries are precisely identified, together with the ones for the stabilizer-generator, as the smallest collection of queries required for this self-testing. The queries associated with all the other Pauli strings are not used. We now show that these correlations can be reproduced by a deterministic classical strategy with one round of nearest-neighbor classical communication.

\begin{theorem}
    The stabilizer-generator correlations Eq.~(\ref{eq:Hart2025_stabilizer}) and all translated and rotated copies of the correlations Eq.~(\ref{eq:Hart2025_anticommutation}) can be reproduced exactly by a deterministic classical strategy using one round of nearest-neighbor communication.
\end{theorem}
\begin{proof} To show that these queries are no longer sufficient to self-test the toric code, it is enough to exhibit a classical strategy that, with one round of classical communication, satisfies the winning condition. For that, let
    \begin{equation}
        \tau_e \in \{E_e,A_e,B_e,C_e\}
    \end{equation}
    denote the observable corresponding to the input received by party $e$, and let $N(e)$ denote its four neighbors in the communication graph. After one round of communication, party $e$ will then learn the inputs received by all $f\in N(e)$.

    To win the game, they then implement the following deterministic strategy. For $O\in\{A,B,C\}$, define
    \begin{equation}
        n_O(e)=\left|\left\{f\in N(e):\tau_f=O_f\right\}\right|.
    \end{equation}
    Each party $e$ outputs $y_e=1$ if one of the following conditions is satisfied:
    \begin{equation}
    \begin{aligned}
        &1.\quad \tau_e=C_e,\qquad n_C(e)=1,
        \qquad e\text{ is a horizontal edge};\\
        &2.\quad \tau_e=A_e,\qquad n_C(e)=2, \qquad n_B(e)=2; \\
        &3.\quad \tau_e=B_e,\qquad n_C(e)=2, \qquad n_A(e)=2; \\        
    \end{aligned}
    \end{equation}
    In every other case, the party outputs $y_e=0$. Notice that this strategy depends only on the input received by the party, the orientation of its edge, and the inputs received by its nearest neighbors, and can therefore be implemented after a single round of communication.

    We now show that this strategy satisfies the winning condition in Eq.~\eqref{eq:win-cond}. First, for the stabilizer-generator queries and for $S_0$, no party receives the input $C$. Therefore, neither of the two conditions above can be satisfied, and every party outputs $0$. Since $a_eb_e=0$ for all parties in these queries, we obtain
    \begin{equation}
        \sum_{e\in\Lambda_P}y_e
        =
        \frac{1}{2}\sum_{e\in\Lambda_P}a_eb_e
        =0
        \pmod 2.
    \end{equation}

    For $S_{12}$ and $S_{23}$, exactly two parties receive the input $C$, and these correspond to perpendicular neighboring edges. Hence both have $n_C(e)=1$, but exactly one of them is horizontal and therefore satisfies condition $1$. No party satisfies condition $2$ or $3$, so exactly one party outputs $1$. Since $a_eb_e=1$ precisely for the two parties receiving $C$, it follows that
    \begin{equation}
        \sum_{e\in\Lambda_P}y_e
        =1
        =
        \frac{1}{2}\sum_{e\in\Lambda_P}a_eb_e
        \pmod 2.
    \end{equation}

    Finally, for $S_{31}$, the two parties receiving $C$ correspond to parallel edges, and therefore neither satisfies condition $1$. Instead, exactly one party receiving $A$ has two neighboring parties receiving $C$ and two receiving $B$, and therefore satisfies condition $2$. Thus, again, exactly one party outputs $1$, and
    \begin{equation}
        \sum_{e\in\Lambda_P}y_e
        =1
        =
        \frac{1}{2}\sum_{e\in\Lambda_P}a_eb_e
        \pmod 2.
    \end{equation}

    The same local argument applies to all translated and rotated copies of these queries, and even for the queries corresponding to swapping the roles of $A_e$ and $B_e$ in Eq.~\eqref{eq:Hart2025_anticommutation}, by using condition 3. Hence, all the correlations used in the rigidity argument of \cite{Hart2025} satisfy the winning condition under this deterministic classical strategy with one round of nearest-neighbor communication.
\end{proof}
\begin{remark}
    The theorem concerns only the collection of queries explicitly used in the rigidity proof of \cite{Hart2025}. It does not imply that the full CSS submeasurement game is one-round classically simulable. In fact, the CSS submeasurement game is defined using all valid submeasurements of the relevant stabilizers, and therefore already contains, at the level of the tested submeasurements, the more extended queries needed to obtain rigidity in the presence of communication, including those used in our construction. The only additional ingredient in our communication-robust protocol is that we also specify the inputs assigned to parties that do not participate in the submeasurement, whose outputs are discarded by the referee but whose inputs can affect the local contexts available after communication. What fails under one round of communication is instead the particular rigidity argument given in \cite{Hart2025}which relies only on two tests which are sufficient to show the anti-commutaion relations, and we have shown that the behaviors of these two tests admit an exact one-round classical simulation. Consequently, their proof establishes rigidity only in the non-communicating setting and cannot be extended directly to one round of communication. Achieving rigidity against communication requires making essential use of additional submeasurement queries already present in their game, together with an appropriate choice of the inputs outside the corresponding submeasurements.
\end{remark}

\section{Smallest lattice sizes found for bounded communication}
\label{app:minimal_lattices}

For completeness, we present the explicit measurement patterns associated with the smallest lattice sizes found by our numerical search for communication restricted to graph distances $r=0,\dots,5$.
All lattices have periodic boundary conditions, and the four panels in each row represent the submeasurements $S_1,S_2,S_3,S_4$, from left to right.
In the anticommutation certificates, the black square marks the target edge.
In the stabilizer certificates, each black square marks a local observable that survives in the final stabilizer product and appears only in the submeasurement contributing that observable. We omit the constructions for $r=6,\dots,9$, once the lattices become too big to be displayed here, and we instead refer for their plots in \cite{froes2026code}.

Since the solver searches only for feasible solutions, some of the constructions presented here admit straightforward simplifications. For example, the submeasurements yielding the anticommutation certificate in Fig.~\ref{fig:anticommutation_h_l4} already use the canonical context defined in the main text. They can therefore be reused to certify anticommutation at vertical edges by just rotating them by 90°, as well as be used to certify the star and plaquette stabilizers using the submeasurements in Eqs.~\eqref{eq:star-submeasurement} and~\eqref{eq:plaquette-submeasurement}, respectively.

% \newpage
\paragraph{$r=0$.}
The smallest lattice found for $r=0$ has side length $L=2$.
The corresponding measurement patterns are shown in Fig.~\ref{fig:patterns_L2}.
\nopagebreak
\begin{figure}[H]
    \centering
    \begin{subfigure}{0.87\linewidth}
        \centering
        \includegraphics[width=\linewidth]{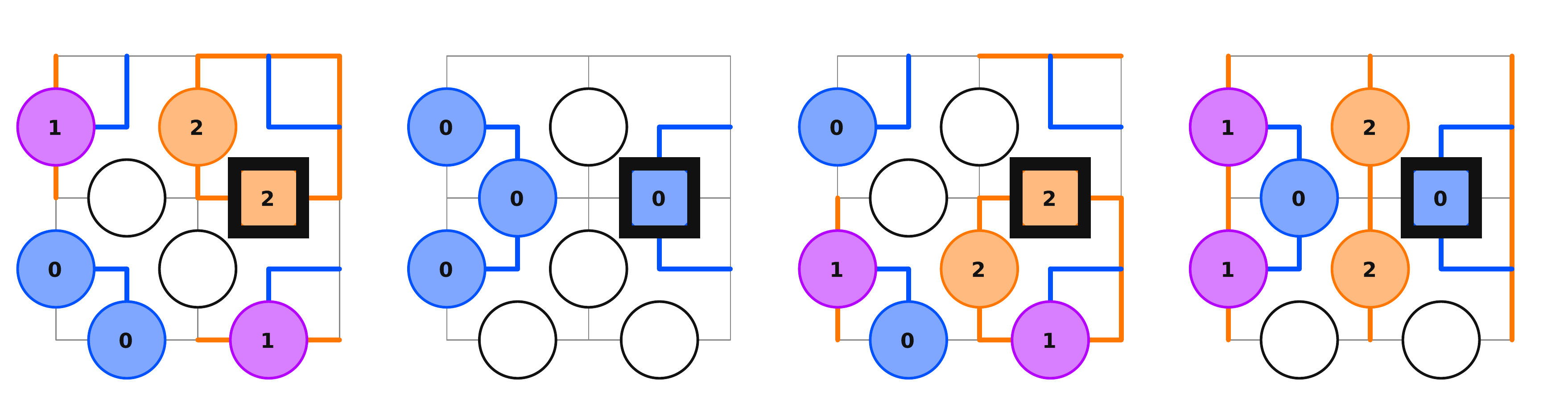}
        \caption{Anticommutation certificate (horizontal edge).}
        \label{fig:anticommutation_h_l2}
    \end{subfigure}
    
    \begin{subfigure}{0.87\linewidth}
        \centering
        \includegraphics[width=\linewidth]{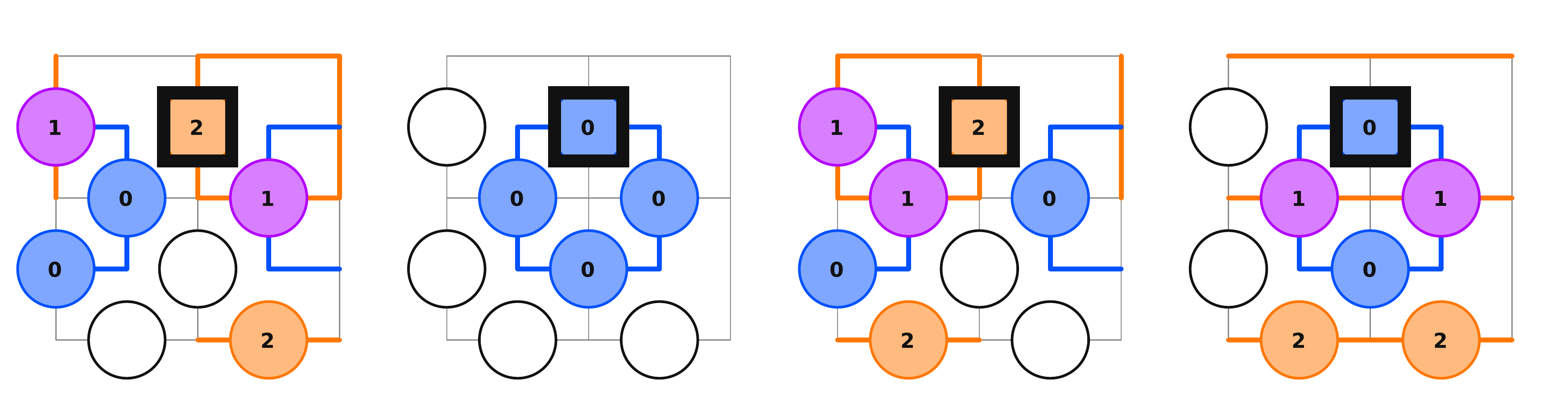}
        \caption{Anticommutation certificate (vertical edge).}
        \label{fig:anticommutation_v_l2}
    \end{subfigure}
    
    \begin{subfigure}{0.87\linewidth}
        \centering
        \includegraphics[width=\linewidth]{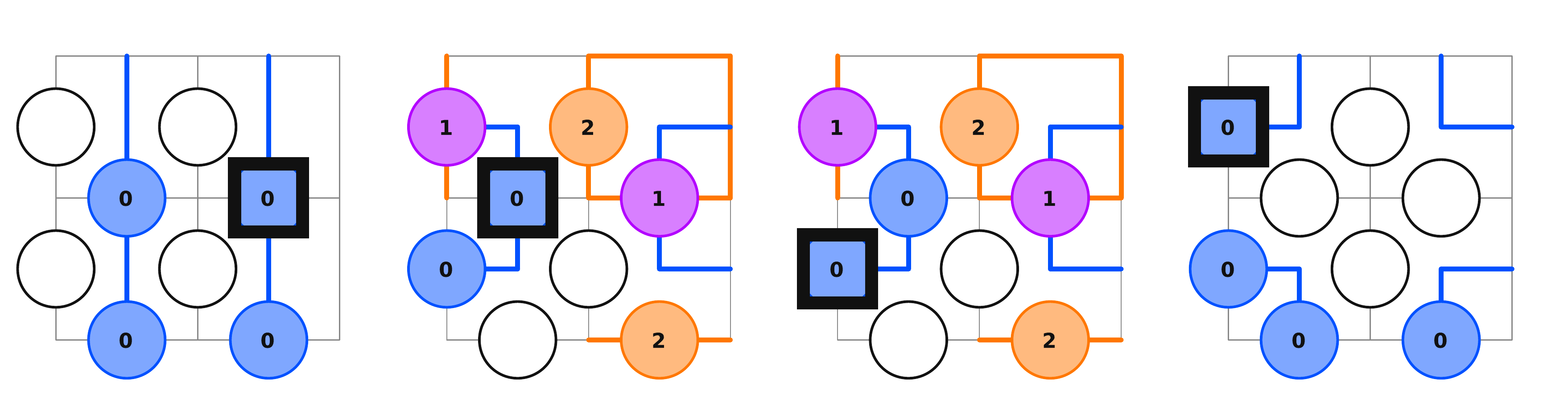}
        \caption{Star stabilizer certificate.}
        \label{fig:star_l2}
    \end{subfigure}
    
    \begin{subfigure}{0.87\linewidth}
        \centering
        \includegraphics[width=\linewidth]{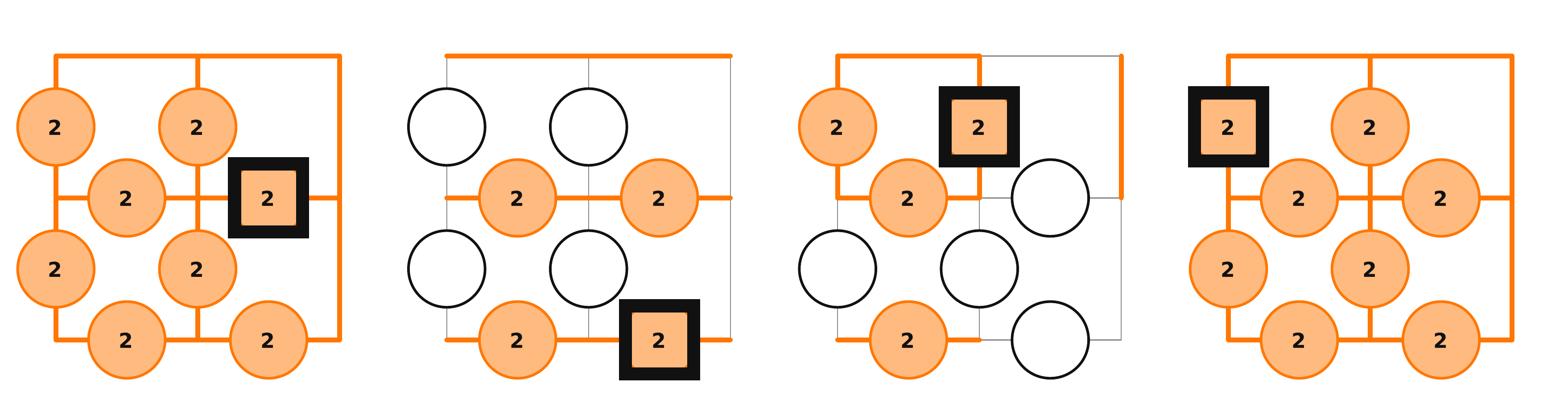}
        \caption{Plaquette stabilizer certificate.}
        \label{fig:plaquette_l2}
    \end{subfigure}
    \caption{Measurement patterns found for $r=0$ rounds of classical communication on a lattice of side length $L=2$.
    The panels show the certificates for (\subref*{fig:anticommutation_h_l2}) and (\subref*{fig:anticommutation_v_l2}) the anticommutation relations, (\subref*{fig:star_l2}) the star stabilizer, and (\subref*{fig:plaquette_l2}) the plaquette stabilizer.
    }
    \label{fig:patterns_L2}
\end{figure}

\newpage
\paragraph{$r=1$.}
The smallest lattice found for $r=1$ has side length $L=4$.
The corresponding measurement patterns are shown in Fig.~\ref{fig:patterns_L4}.
\nopagebreak
\begin{figure}[H]
    \centering
    \begin{subfigure}{0.95\linewidth}
        \centering
        \includegraphics[width=\linewidth]{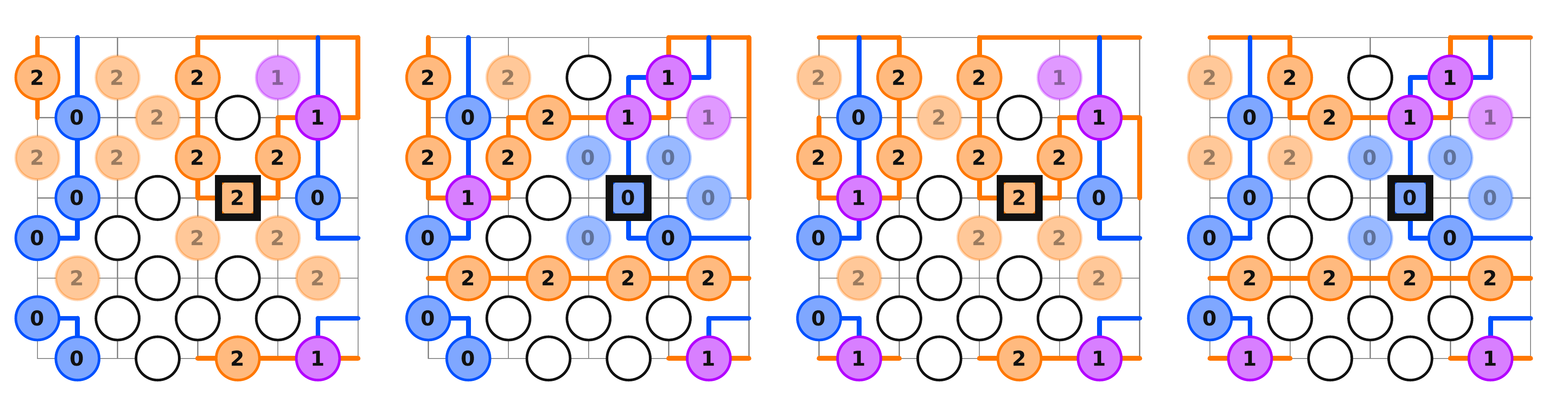}
        \caption{Anticommutation certificate (horizontal edge).}
        \label{fig:anticommutation_h_l4}
    \end{subfigure}
    
    \begin{subfigure}{0.95\linewidth}
        \centering
        \includegraphics[width=\linewidth]{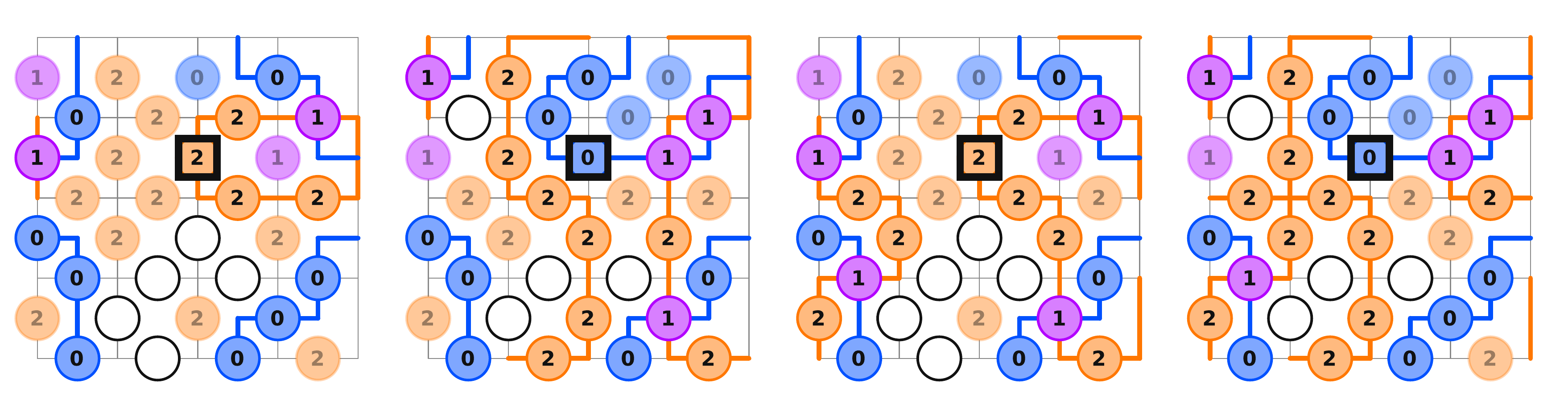}
        \caption{Anticommutation certificate (vertical edge).}
        \label{fig:anticommutation_v_l4}
    \end{subfigure}
    
    \begin{subfigure}{0.95\linewidth}
        \centering
        \includegraphics[width=\linewidth]{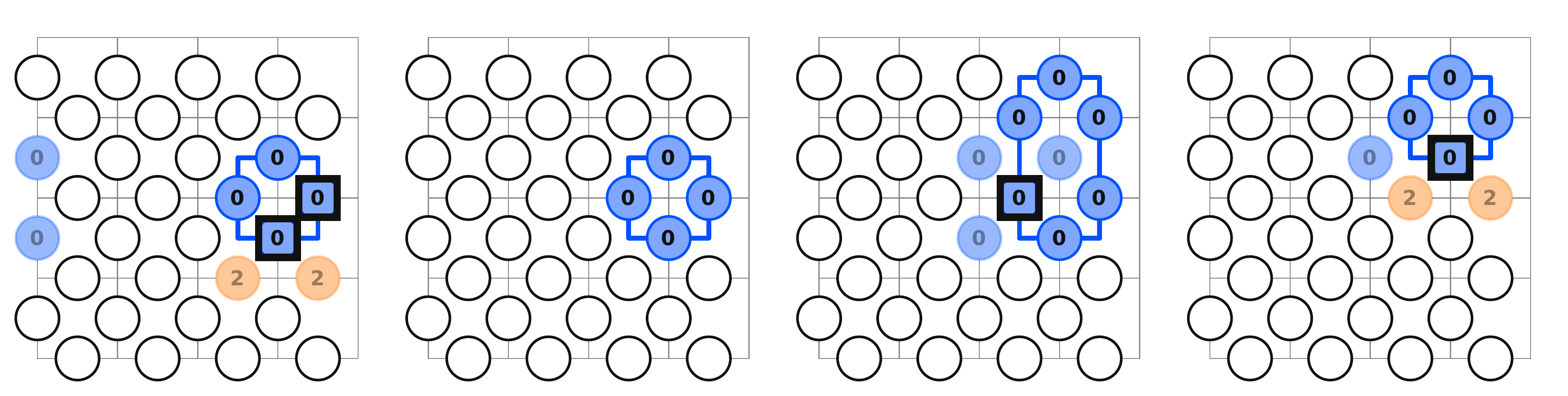}
        \caption{Star stabilizer certificate.}
        \label{fig:star_l4}
    \end{subfigure}
    
    \begin{subfigure}{0.95\linewidth}
        \centering
        \includegraphics[width=\linewidth]{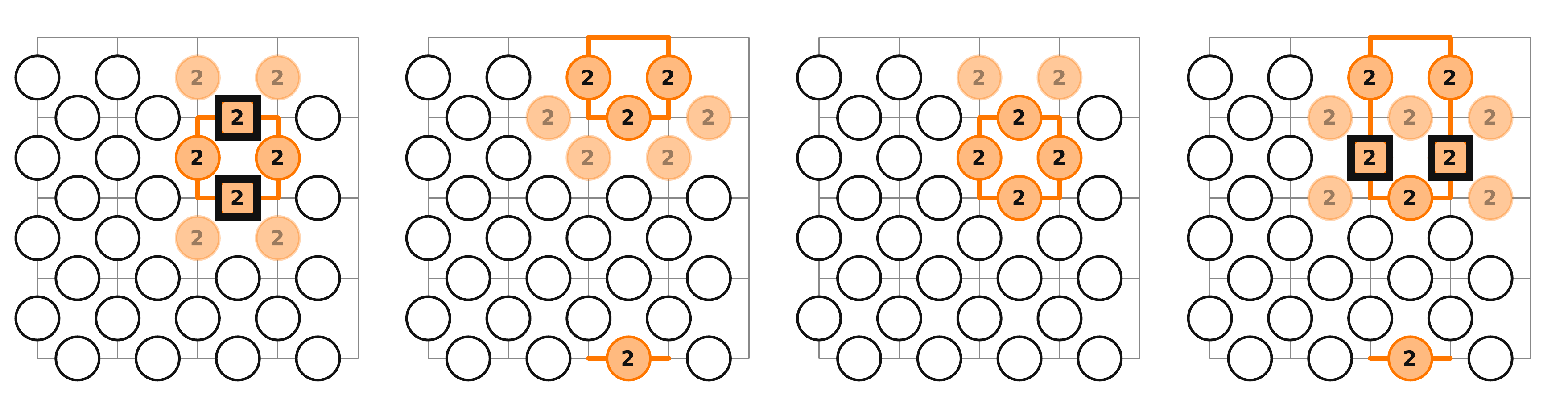}
        \caption{Plaquette stabilizer certificate.}
        \label{fig:plaquette_l4}
    \end{subfigure}
    \caption{Measurement patterns found for $r=1$ rounds of classical communication on a lattice of side length $L=4$.
    The panels show the certificates for (\subref*{fig:anticommutation_h_l4}) and (\subref*{fig:anticommutation_v_l4}) the anticommutation relations, (\subref*{fig:star_l4}) the star stabilizer, and (\subref*{fig:plaquette_l4}) the plaquette stabilizer.
    }
    \label{fig:patterns_L4}
\end{figure}

\newpage
\paragraph{$r=2$.}
The smallest lattice found for $r=2$ has side length $L=6$.
The corresponding measurement patterns are shown in Fig.~\ref{fig:patterns_L6}.
\nopagebreak
\begin{figure}[H]
    \centering
    \begin{subfigure}{0.95\linewidth}
        \centering
        \includegraphics[width=\linewidth]{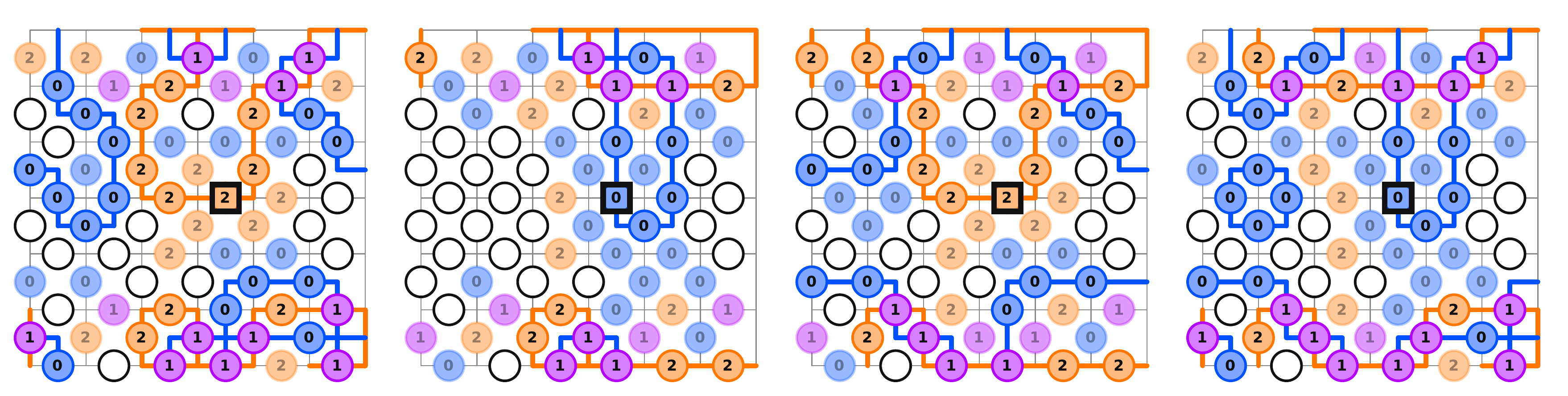}
        \caption{Anticommutation certificate (horizontal edge).}
        \label{fig:anticommutation_h_l6}
    \end{subfigure}
    
    \begin{subfigure}{0.95\linewidth}
        \centering
        \includegraphics[width=\linewidth]{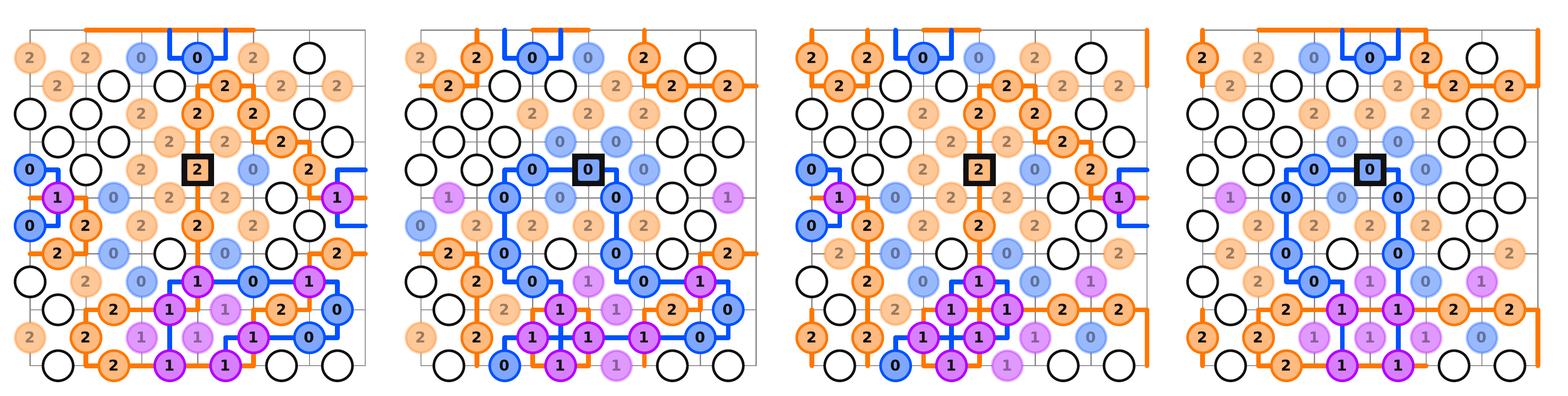}
        \caption{Anticommutation certificate (vertical edge).}
        \label{fig:anticommutation_v_l6}
    \end{subfigure}
    
    \begin{subfigure}{0.95\linewidth}
        \centering
        \includegraphics[width=\linewidth]{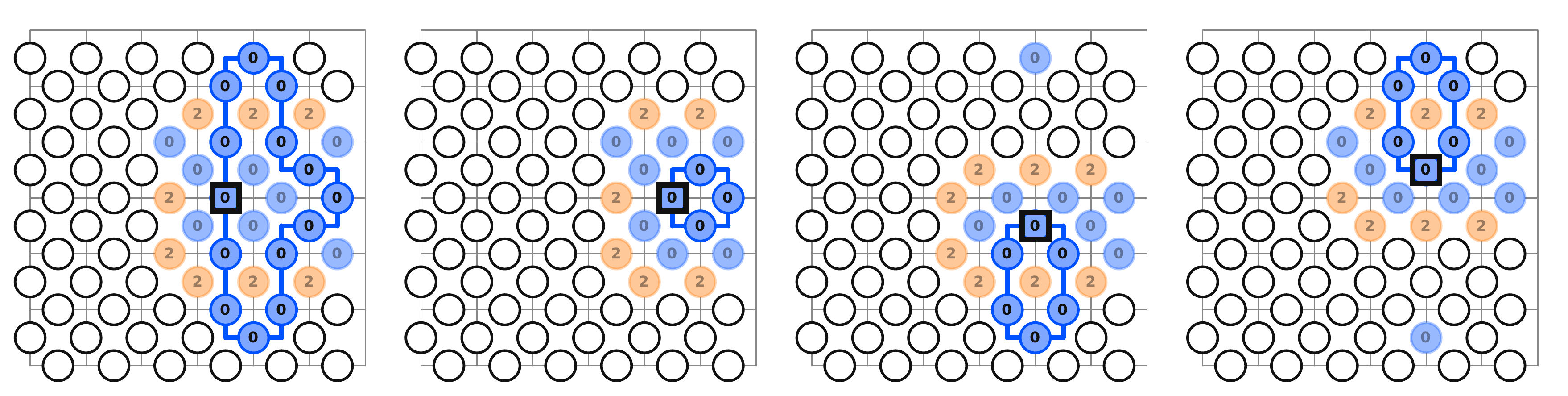}
        \caption{Star stabilizer certificate.}
        \label{fig:star_l6}
    \end{subfigure}
    
    \begin{subfigure}{0.95\linewidth}
        \centering
        \includegraphics[width=\linewidth]{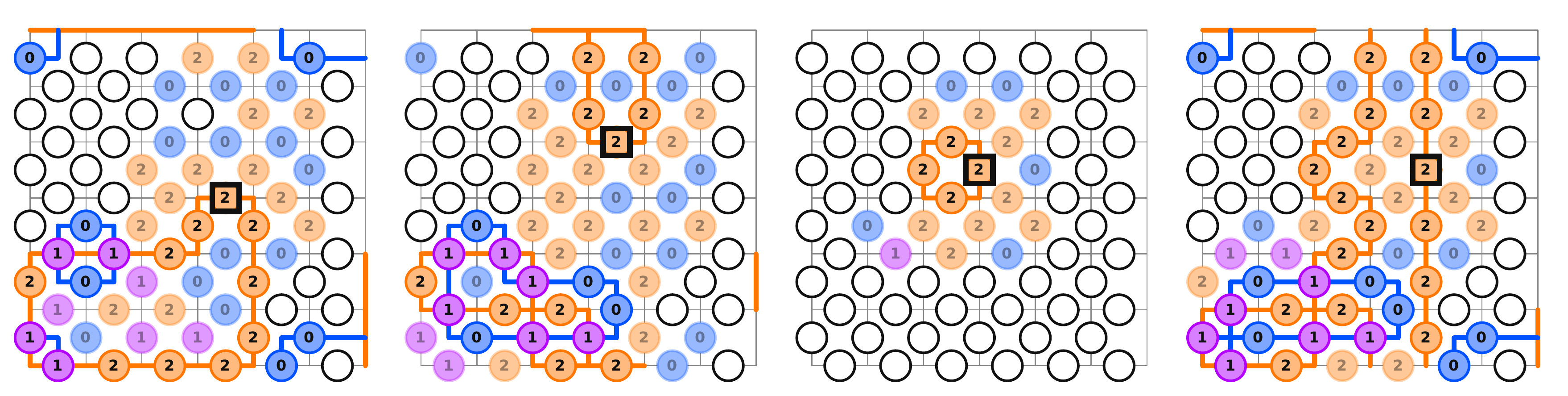}
        \caption{Plaquette stabilizer certificate.}
        \label{fig:plaquette_l6}
    \end{subfigure}
    \caption{Measurement patterns found for $r=2$ rounds of classical communication on a lattice of side length $L=6$.
    The panels show the certificates for (\subref*{fig:anticommutation_h_l6}) and (\subref*{fig:anticommutation_v_l6}) the anticommutation relations, (\subref*{fig:star_l6}) the star stabilizer, and (\subref*{fig:plaquette_l6}) the plaquette stabilizer.
    }
    \label{fig:patterns_L6}
\end{figure}

\newpage
\paragraph{$r=3$.}
The smallest lattice found for $r=3$ has side length $L=8$.
The corresponding measurement patterns are shown in Fig.~\ref{fig:patterns_L8}.
\nopagebreak
\begin{figure}[H]
    \centering
    \begin{subfigure}{0.95\linewidth}
        \centering
        \includegraphics[width=\linewidth]{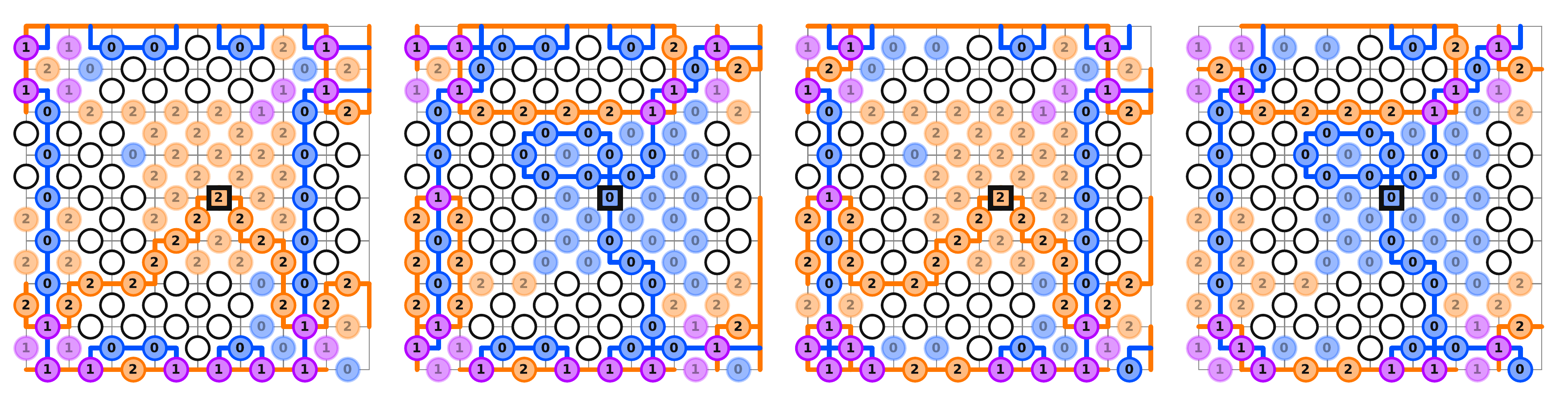}
        \caption{Anticommutation certificate (horizontal edge).}
        \label{fig:anticommutation_h_l8}
    \end{subfigure}
    
    \begin{subfigure}{0.95\linewidth}
        \centering
        \includegraphics[width=\linewidth]{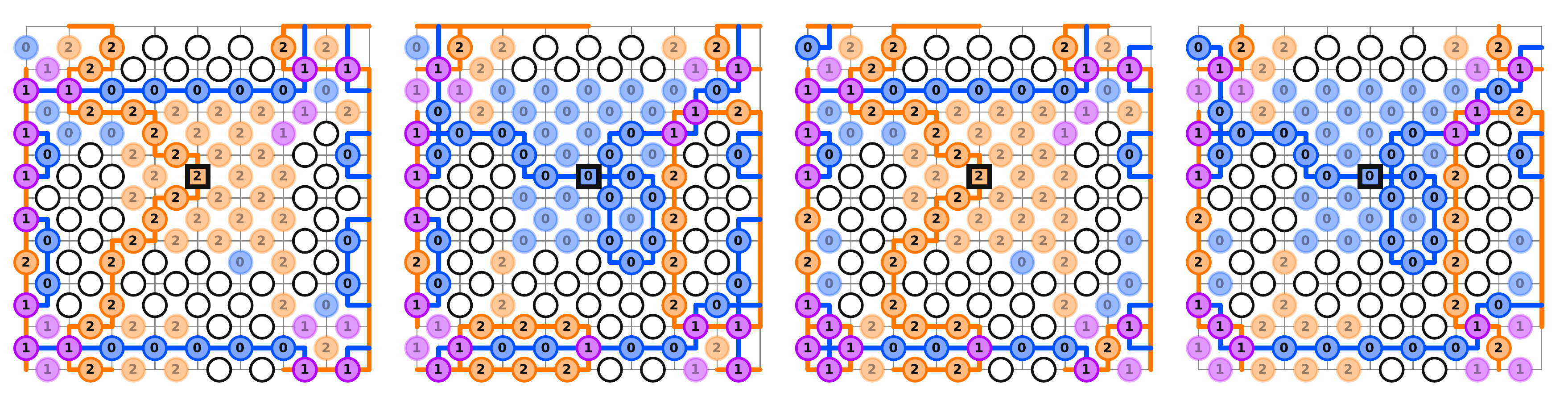}
        \caption{Anticommutation certificate (vertical edge).}
        \label{fig:anticommutation_v_l8}
    \end{subfigure}
    
    \begin{subfigure}{0.95\linewidth}
        \centering
        \includegraphics[width=\linewidth]{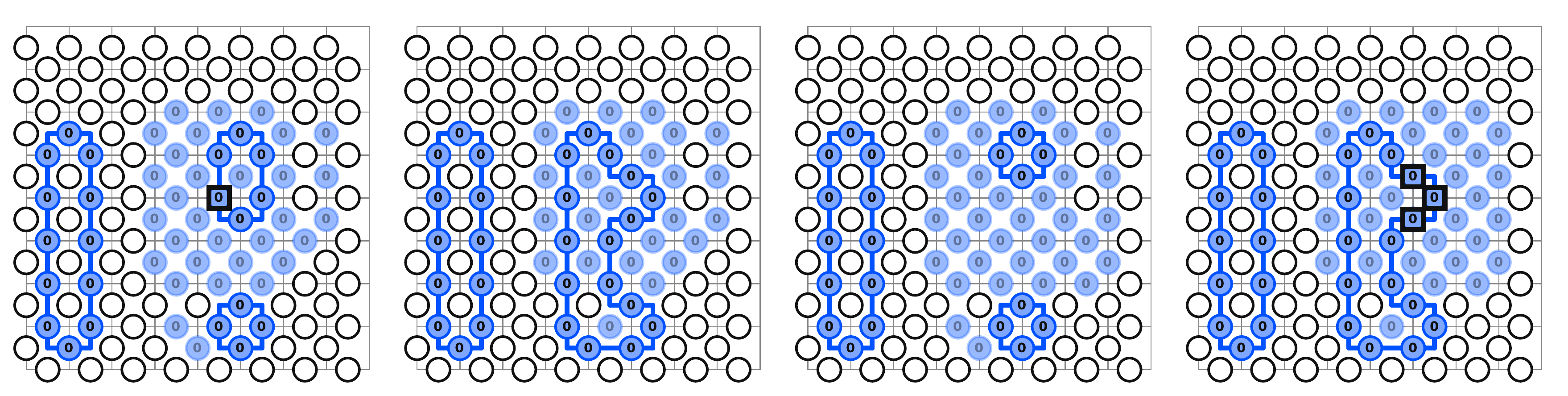}
        \caption{Star stabilizer certificate.}
        \label{fig:star_l8}
    \end{subfigure}
    
    \begin{subfigure}{0.95\linewidth}
        \centering
        \includegraphics[width=\linewidth]{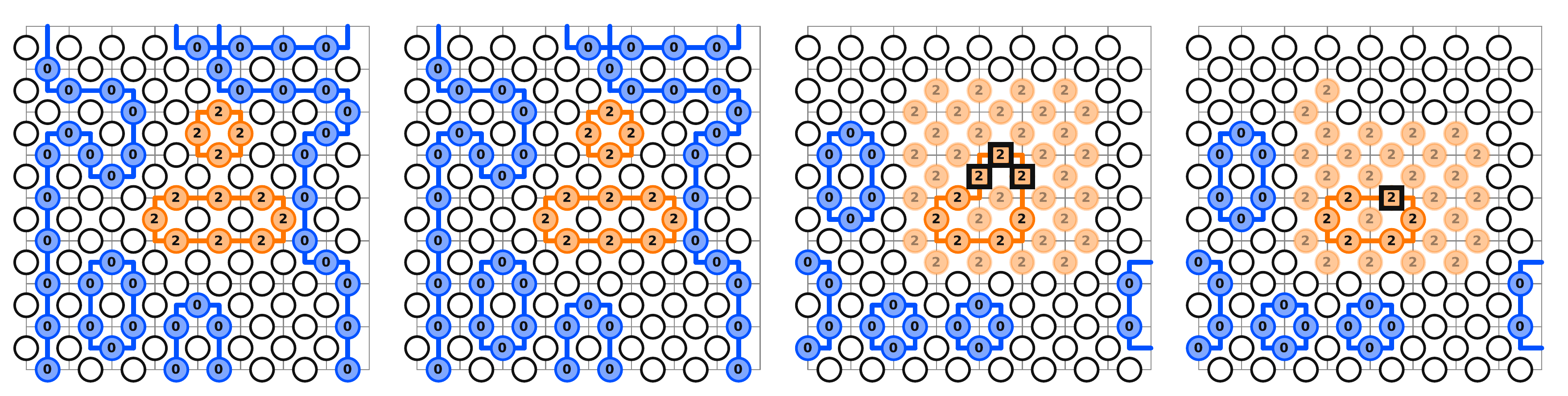}
        \caption{Plaquette stabilizer certificate.}
        \label{fig:plaquette_l8}
    \end{subfigure}
    \caption{Measurement patterns found for $r=3$ rounds of classical communication on a lattice of side length $L=8$.
    The panels show the certificates for (\subref*{fig:anticommutation_h_l8}) and (\subref*{fig:anticommutation_v_l8}) the anticommutation relations, (\subref*{fig:star_l8}) the star stabilizer, and (\subref*{fig:plaquette_l8}) the plaquette stabilizer.
    }
    \label{fig:patterns_L8}
\end{figure}

\newpage
\paragraph{$r=4$.}
The smallest lattice found for $r=4$ has side length $L=10$.
The corresponding measurement patterns are shown in Fig.~\ref{fig:patterns_L10}.
\nopagebreak
\begin{figure}[H]
    \centering
    \begin{subfigure}{0.95\linewidth}
        \centering
        \includegraphics[width=\linewidth]{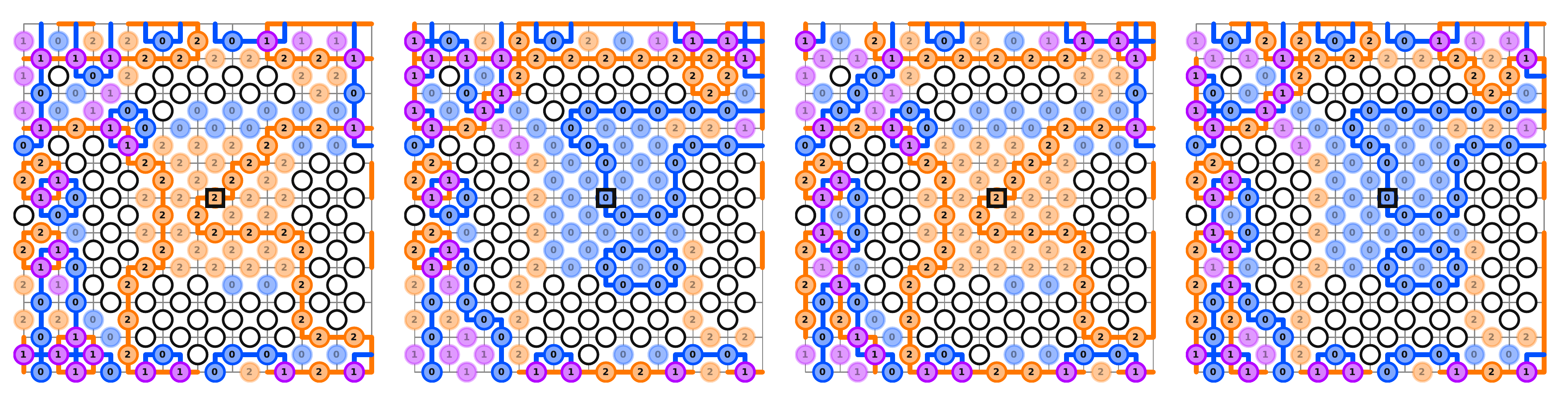}
        \caption{Anticommutation certificate (horizontal edge).}
        \label{fig:anticommutation_h_l10}
    \end{subfigure}
    
    \begin{subfigure}{0.95\linewidth}
        \centering
        \includegraphics[width=\linewidth]{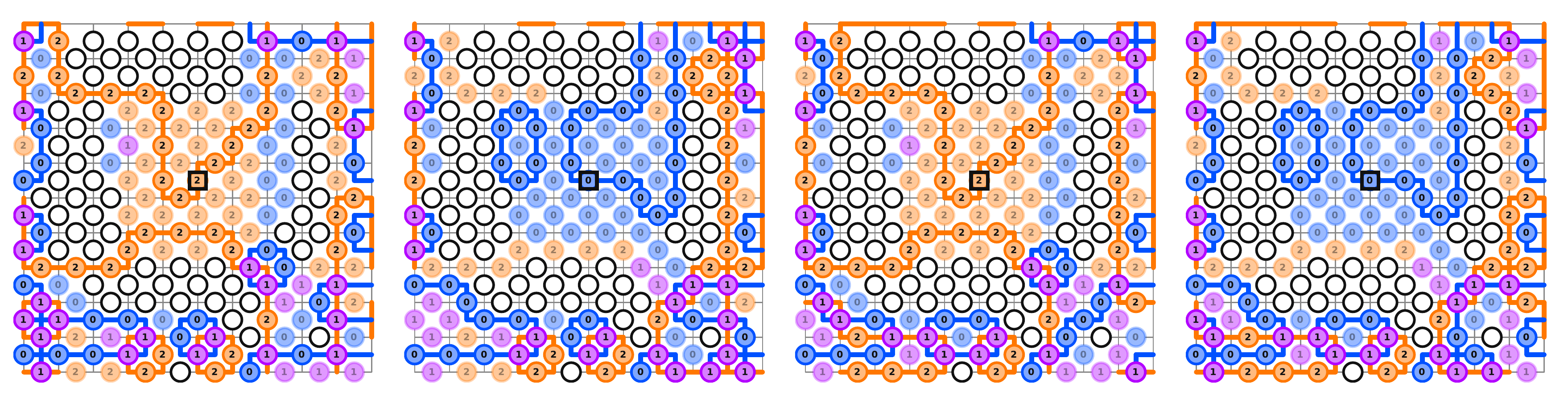}
        \caption{Anticommutation certificate (vertical edge).}
        \label{fig:anticommutation_v_l10}
    \end{subfigure}
    
    \begin{subfigure}{0.95\linewidth}
        \centering
        \includegraphics[width=\linewidth]{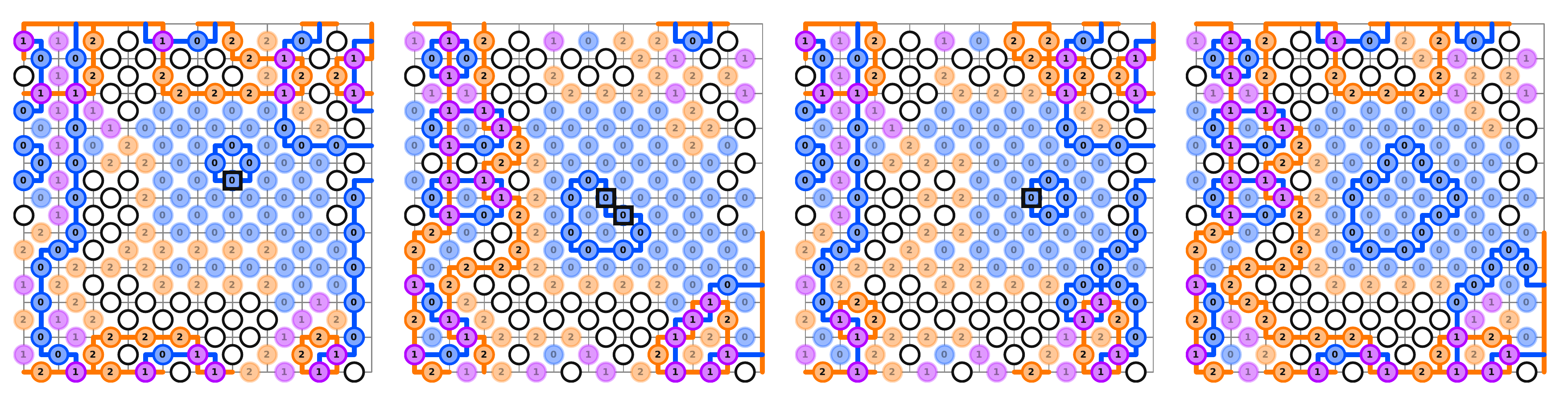}
        \caption{Star stabilizer certificate.}
        \label{fig:star_l10}
    \end{subfigure}
    
    \begin{subfigure}{0.95\linewidth}
        \centering
        \includegraphics[width=\linewidth]{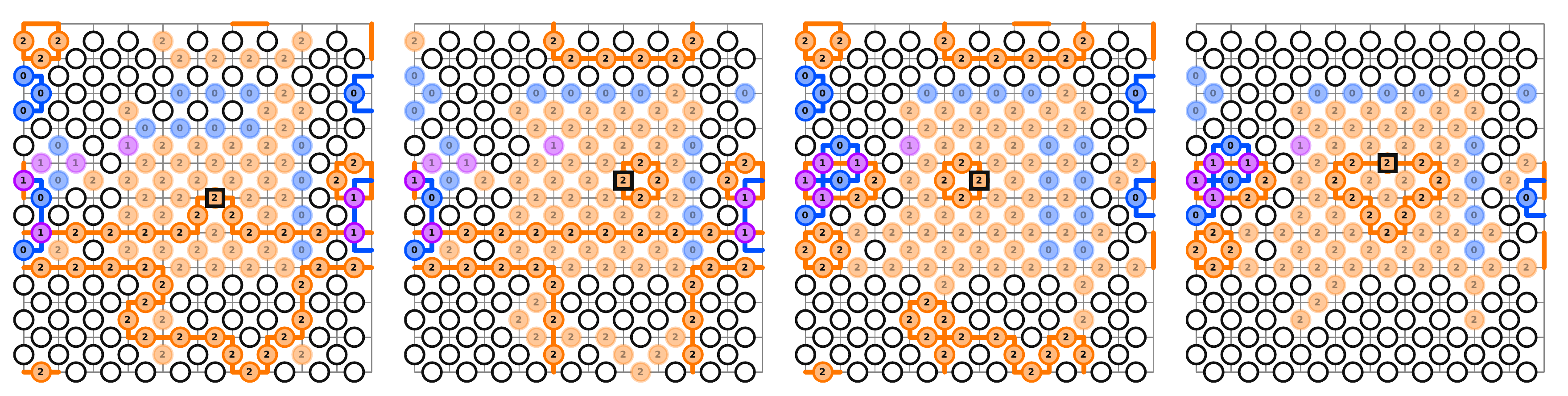}
        \caption{Plaquette stabilizer certificate.}
        \label{fig:plaquette_l10}
    \end{subfigure}
    \caption{Measurement patterns found for $r=4$ rounds of classical communication on a lattice of side length $L=10$.
    The panels show the certificates for (\subref*{fig:anticommutation_h_l10}) and (\subref*{fig:anticommutation_v_l10}) the anticommutation relations, (\subref*{fig:star_l10}) the star stabilizer, and (\subref*{fig:plaquette_l10}) the plaquette stabilizer.
    }
    \label{fig:patterns_L10}
\end{figure}

\newpage
\paragraph{$r=5$.}
The smallest lattice found for $r=5$ has side length $L=12$.
The corresponding measurement patterns are shown in Fig.~\ref{fig:patterns_L12}.
\nopagebreak
\begin{figure}[H]
    \centering
    \begin{subfigure}{0.95\linewidth}
        \centering
        \includegraphics[width=\linewidth]{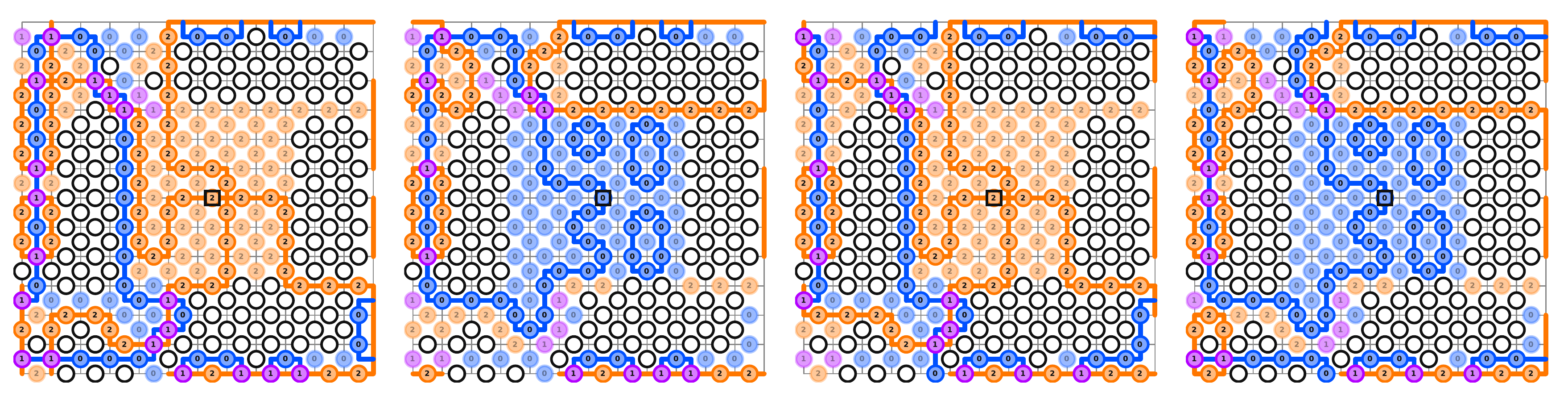}
        \caption{Anticommutation certificate (horizontal edge).}
        \label{fig:anticommutation_h_l12}
    \end{subfigure}
    
    \begin{subfigure}{0.95\linewidth}
        \centering
        \includegraphics[width=\linewidth]{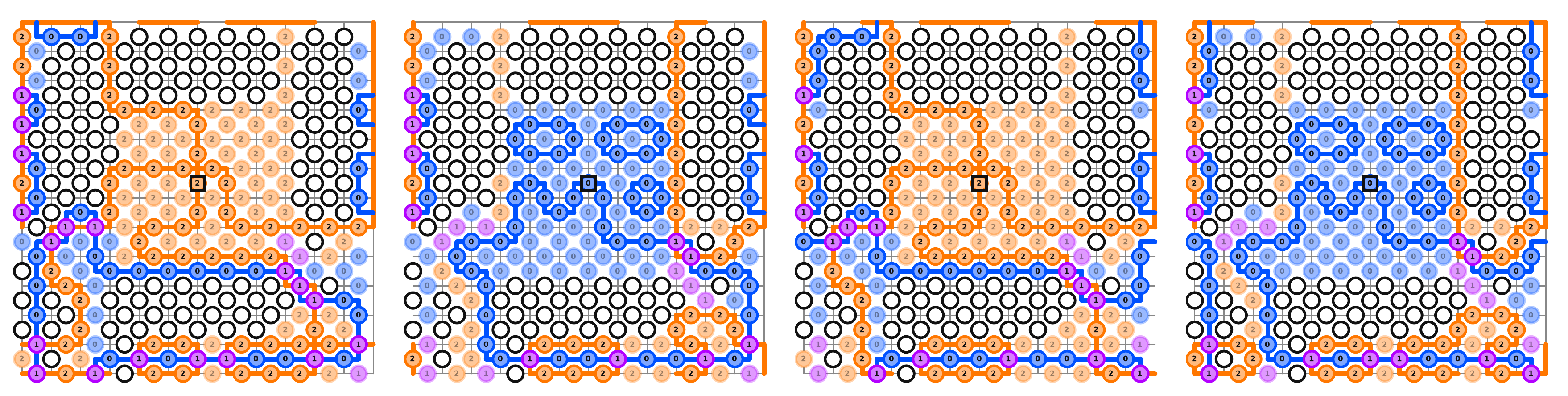}
        \caption{Anticommutation certificate (vertical edge).}
        \label{fig:anticommutation_v_l12}
    \end{subfigure}
    
    \begin{subfigure}{0.95\linewidth}
        \centering
        \includegraphics[width=\linewidth]{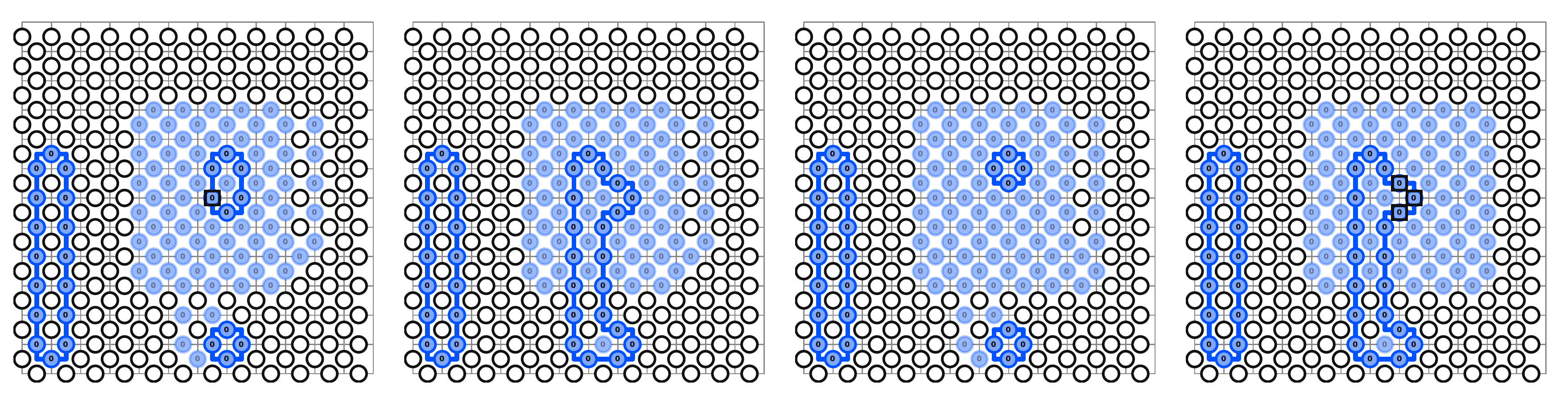}
        \caption{Star stabilizer certificate.}
        \label{fig:star_l12}
    \end{subfigure}
    
    \begin{subfigure}{0.95\linewidth}
        \centering
        \includegraphics[width=\linewidth]{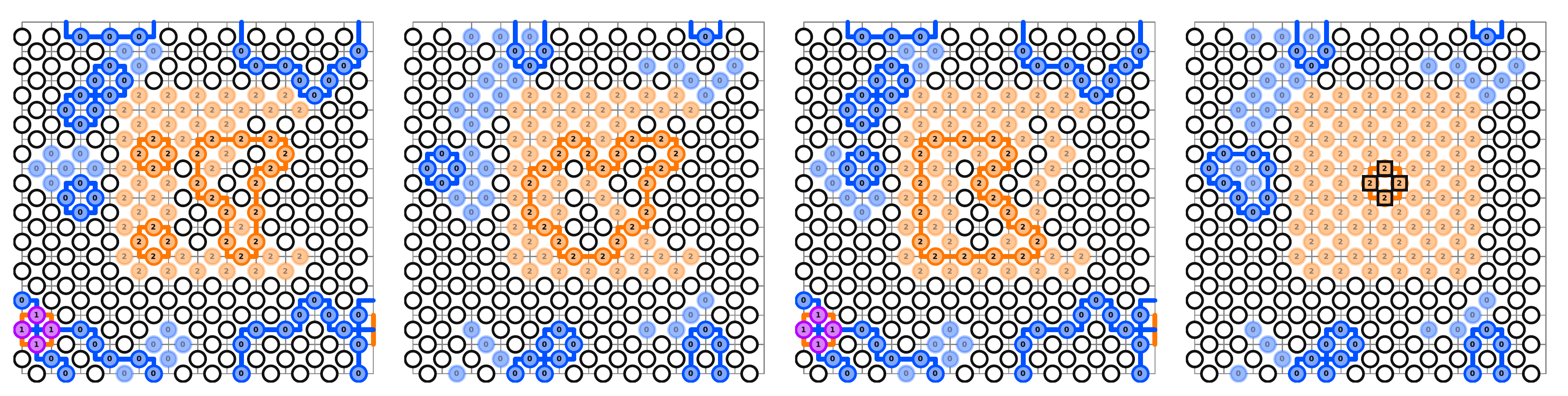}
        \caption{Plaquette stabilizer certificate.}
        \label{fig:plaquette_l12}
    \end{subfigure}
    \caption{Measurement patterns found for $r=5$ rounds of classical communication on a lattice of side length $L=12$.
    The panels show the certificates for (\subref*{fig:anticommutation_h_l12}) and (\subref*{fig:anticommutation_v_l12}) the anticommutation relations, (\subref*{fig:star_l12}) the star stabilizer, and (\subref*{fig:plaquette_l12}) the plaquette stabilizer.
    }
    \label{fig:patterns_L12}
\end{figure}

\end{document}